%% file: CoarseGrainHamiltonianWSINDy.tex
\documentclass[fleqn,10pt]{wlscirep}
\usepackage[utf8]{inputenc}
\usepackage[T1]{fontenc}
\usepackage{
accents,
algorithm2e,
amsmath,
amssymb,
amsthm,
array,
appendix,
bm,
bbm,
color,
comment,
epsf,
enumitem,
float,
graphicx,
hyperref,
listings,
mathtools,
mathrsfs,
MnSymbol,
pifont,
setspace,
textcomp,
tikz,
titletoc,
units,
url,
ulem,
wrapfig,
xcolor,
}
\newtheorem{theorem}{Theorem}
\newtheorem{lemma}{Lemma}

\usetikzlibrary{spy,decorations.fractals}

\input{notation}

\title{Coarse-Graining Hamiltonian Systems Using WSINDy}

\author[1,*]{Daniel A. Messenger}
\author[2,$\dagger$]{Joshua W. Burby}
\author[1,+]{David M. Bortz}
\affil[1]{University of Colorado, Department of Applied Mathematics, Boulder, CO, 80309-0526, USA}
\affil[2]{Los Alamos National Laboratory, Theoretical Division, Los Alamos, NM, 87545, USA}

\affil[*]{daniel.messenger@colorado.edu}
\affil[$\dagger$]{jburby@lanl.gov}
\affil[+]{david.bortz@colorado.edu}


\begin{abstract}
The Weak-form Sparse Identification of Nonlinear Dynamics algorithm (WSINDy) has been demonstrated to offer coarse-graining capabilities in the context of interacting particle systems (\url{https://doi.org/10.1016/j.physd.2022.133406}). In this work we extend this capability to the problem of coarse-graining Hamiltonian dynamics which possess approximate symmetries associated with timescale separation. A smooth $\vep$-dependent Hamiltonian vector field $X_\vep$ possesses an approximate symmetry if the limiting vector field $X_0=\lim_{\vep\to 0}X_\vep$ possesses an exact symmetry. Such approximate symmetries often lead to the existence of a Hamiltonian system of reduced dimension that may be used to efficiently capture the dynamics of the symmetry-invariant dependent variables. Deriving such reduced systems, or approximating them numerically, is an ongoing challenge. We demonstrate that WSINDy can successfully identify this reduced Hamiltonian system in the presence of large perturbations imparted in the $\vep>0$ regime, while remaining robust to extrinsic noise. This is significant in part due to the nontrivial means by which such systems are derived analytically. WSINDy naturally preserves the Hamiltonian structure by restricting to a trial basis of Hamiltonian vector fields. The methodology is computational efficient, often requiring only a single trajectory to learn the global reduced Hamiltonian, and avoiding forward solves in the learning process. In this way, we argue that weak-form equation learning is particularly well-suited for Hamiltonian coarse-graining. Using nearly-periodic Hamiltonian systems as a prototypical class of systems with approximate symmetries, we show that WSINDy robustly identifies the correct leading-order system, with dimension reduced by at least two, upon observation of the relevant degrees of freedom. We also provide a contribution to the literature on averaging theory by proving that first-order averaging at the level of vector fields preserves Hamiltonian structure in nearly-periodic Hamiltonian systems. We provide physically relevant examples, namely coupled oscillator dynamics, the H\'enon-Heiles system for stellar motion within a galaxy, and the dynamics of charged particles.
\end{abstract}

\begin{document}

\flushbottom
\maketitle
%
%
\thispagestyle{empty}


\section{Introduction}\label{sec:intro}

Hamiltonian mechanics is a formulation of classical mechanics that is used to describe non-dissipative systems\cite{AbrahamMarsden1978}. Hamiltonian descriptions of physical systems allow for geometric interpretations which are not immediately present in the Newtonian and Lagrangian formulations of classical mechanics. Expression of the dynamics in terms of a conserved quantity (the {\it Hamiltonian}) is also essential in the formulation of quantum mechanics. In fact, the Schr\"odinger equation comprises an infinite-dimensional Hamiltonian system\cite{Chernoff_1976}. The geometry of phase space allows one to systematically explore conserved quantities, also known as constants of motion, which indicate the presence of {\it symmetry}. A symmetry is map on phase space which commutes with the flow-map of the dynamics, and may be used to reduce the size of phase space. Continuous families of symmetries lead to phase space dimension reduction. While many physical systems do possess quantities which are strictly conserved, often the system possesses {\it approximately conserved quantities} which lead to {\it approximate symmetries}. Such approximate symmetries can still be used to derive a reduced-order system that approximates well the important degrees of freedom in the original system, but such derivations remain an ongoing challenge.

An important example of a system with an approximate symmetry is a charged particle moving in a strong magnetic field. The particle exhibits fast oscillations around the magnetic field lines, which are largely unimportant to measure. Isolating the slow drift motion along and across the field lines is essential in order to efficiently simulate such systems. The concept of {\it adiabatic invariance} provides the necessary approximately-conserved quantities which allow one to analytically derive reduced Hamiltonian systems for the slow dynamics of charged particle motion. Passage to this reduced system involves averaging over a continuous family of {\it approximate} symmetries resulting from the adiabatic invariant (discussed in Section \ref{sec:nearperreview}). 

In practice, full descriptions of Hamiltonian dynamics in the form of governing equations are challenging to identify from experimental data because the fast-scale oscillations are often underresolved. Moreover, given the full set of governing equations, analytically deriving the reduced-dimension system in the presence of symmetries (or approximate symmetries) becomes infeasible for complex systems. In this article we explore the ability of recent weak-form methods to identify sparse equations for the reduced system in the presence of approximate symmetries, directly from time series data on the Hamiltonian system in question. 

Using nearly-periodic systems as a test case, we show in this work that {\it weak-form equation discovery}, specifically the WSINDy algorithm, which interprets the to-be-discovered dynamics using test functions, provides a framework for directly coarse-graining Hamiltonian systems from observation of only the slow modes. The dictionary learning approach provides a wealth of information on the reduced order system, namely, one can in many cases learn the structure of the entire Hamiltonian from a single noisy trajectory. Moreover, in cases where the Hamiltonian structure is not identified correctly, often the dynamics near the level set of the trajectory are captured very accurately. 

Since the weak form is applied at the level of vector fields, we also provide theoretical results showing that first-order averaged vector fields remain Hamiltonian in the nearly-periodic case. While the all-orders averaging theory of Kruskal \cite{
Kruskal1962JournalofMathematicalPhysics} is known to be Hamiltonian, the Hamiltonian structure underlying first-order averaging has never been identified in full generality, thus we provide a self-contained description which, to the best of the authors' knowledge, constitutes a novel contribution to the theory of averaging.

\subsection{Literature Review}\label{sec:nearperreview}

Significant progress has been made recently in the development of data-driven methods for estimation and identification of Hamiltonian systems and other structured dynamics. These methods include dictionary-based Hamiltonian learning\cite{HolmsenEidnesRiemer-Sorensen2023arXiv230506920}, approximation of Hamiltonian systems and symplectic maps by neural networks \cite{burby2020fast,jin2020sympnets,BertalanDietrichMezicEtAl2019Chaos}, and combined dictionary and neural-network based structure preservation in learned dynamics systems  \cite{LeeTraskStinis}. The key difference between these works and ours are (a) we introduce a weak formulation of the discovery problem, which naturally enables discovery from corrupted data, and (b) we demonstrate that the weak form offers direct coarse-graining capabilities, allowing one to identify reduced-order Hamiltonian systems simply by interpreting the dynamics through the action of a suitable class of test functions.

The subject of coarse-graining and reduced-order modeling for Hamiltonian systems has certainly received attention in recent years but this subfield is far from complete. Peng and Mohseni\cite{PengMohseni2016SIAMJSciComput} developed a symplectic reduced-order modeling analogue of proper orthogonal decomposition, which has seen several extensions\cite{SharmaWangKramer2022PhysicaD}.
From a different perspective, Duruisseaux, Burby, and Tang\cite{DuruisseauxBurbyTang2023SciRep} introduced a neural-network architecture to specifically handle systems possessing an
adiabatic invariant, by which the dynamics may be reduced to an approximate Hamiltonian system of lower dimension. The latter is part of a series of recent works aimed at exploiting near-periodicity to design more efficient means of understanding and simulating dynamics relevant to plasma physics\cite{Burby2022LA-UR-22-265241875767,Burby_Hirvijoki_2021,Burby2022PhysicsofPlasmas,BHL_2023,BurbySquire2020JPlasmaPhys}. The current work can be seen as a continuation of this series, with the purpose of offering computationally efficient, noise-robust, and interpretable discovery of reduced Hamiltonian systems to complement neural network-based approaches\cite{DuruisseauxBurbyTang2023SciRep}.

Here we focus on the dictionary learning approach to identify reduced Hamiltonian systems. In particular, we employ the WSINDy algorithm (Weak-form Sparse Identification of Nonlinear Dynamics), which has its roots in the SINDy algorithm\cite{BruntonProctorKutz2016ProcNatlAcadSci}. The weak form has risen to prominence as a way to combat realistic challenges like noisy data and non-smooth dynamics 
\cite{BortzMessengerDukic2023BullMathBiol,MessengerBortz2022arXiv221116000,MessengerDallAneseBortz2022ProcThirdMathSciMachLearnConf,MessengerWheelerLiuEtAl2022JRSocInterface,MessengerBortz2021JComputPhys,MessengerBortz2021MultiscaleModelSimul,TangLiaoKuskeEtAl2023JComputPhys,SchaefferMcCalla2017PhysRevE,BertsimasGurnee2023NonlinearDyn,FaselKutzBruntonEtAl2022ProcRSocMathPhysEngSci,KaptanogluZhangNicolaouEtAl2023NonlinearDyn,WangHuanGarikipati2021ComputMethodsApplMechEng,GurevichReinboldGrigoriev2019Chaos}. Most relevant to this work, WSINDy has been demonstrated to offer coarse-graining capabilities\cite{MessengerBortz2022PhysicaD} in the context of interacting particle systems and homogenization of parabolic PDEs. In a similar vein to the current work, Bramburger, Dylewsky, and Kutz\cite{BramburgerDylewskyKutz2020PhysRevE} demonstrate that SINDy may be used to identify reduced dynamics in slow-fast systems. However structure-preservation is not considered, and the method is restricted to systems exhibiting an identifiable separation of scales. Moreover, the method assumes a short timescale characterized by periodic orbits, each with the same period. While this method is complementary to ours, the techniques developed here are very different, and are designed for a more general class of problems. In particular, while the method in \cite{BramburgerDylewskyKutz2020PhysRevE} depends on identification of the dominant fast timescale, we demonstrate here that representing the dynamics in weak form is sufficient to coarse-grain the fast scales in addition to combating measurement noise, which can easily hinder the identification of a fast timescale. Our method also allows for fast timescale orbits with variable period.

Compared with POD-based methods\cite{PengMohseni2016SIAMJSciComput,SharmaWangKramer2022PhysicaD}, which expand the reduced dynamics in terms of a data-driven basis, dictionary learning allows one to learn representations of the reduced dynamics in a basis that easily generalizes as it does not depend on the training dataset. In the context of Hamiltonian equation discovery, the method we study here identifies the Hamiltonian over all of phase space, and can easily be used to explore unseen energy levels. Moreover, structure preservation is easily enforced at the level of the dictionary (see Section \ref{sec:prelims}). On the other end of the spectrum, neural network-based approaches have the capacity to approximate well the underlying dynamics but lack the interpretability and computational efficiency native to dictionary learning. For Hamiltonian equation discovery, neural-networks typically require many trajectories and do not provide global knowledge of the Hamiltonian\cite{DuruisseauxBurbyTang2023SciRep}. However, structure preservation can be enforced in the neural network architecture\cite{burby2020fast,jin2020sympnets,DuruisseauxBurbyTang2023SciRep}.

Ultimately, the main purpose of this work is to demonstrate that the weak form itself has inherent temporal coarse-graining capabilities, which are especially useful in reduced-order Hamiltonian modeling. We note that combinations of weak-form equation learning with other reduced-order modeling paradigms is possible, as exhibited by previous works in other contexts (.e.g\ POD-based methods\cite{russo2022convergence,russo2023streaming} and Neural Networks\cite{stephany2023weak}), and we leave these synergies to future work.

\subsection{Paper organization}

We include preliminary concepts relevant to the study of Hamiltonian coarse-graining in Section \ref{sec:prelims}, namely, an overview of Hamiltonian systems (\ref{sec:Hsys}) with specific attention paid to nearly-periodic Hamiltonian systems (\ref{sec:nearperreview}). The latter is the prototypical class of approximate-symmetries which we use in the current manuscript to investigate weak-form coarse-graining. Section \ref{sec:1stOrderHtheory} contains theoretical results proving that first-order averaging of nearly-periodic Hamiltonian systems at the level of vector fields preserves Hamiltonian structure in the resulting reduced phase space, justifying a search for Hamiltonian coarse-grained models in the presence of approximate symmetries. In Section \ref{sec:wsindyH} we describe how the WSINDy algorithm may be applied to learn a general Hamiltonian system of the form \ref{sec:Hsys} and we demonstrate in \ref{sec:WSINDy4HS} that for systems with approximate symmetries, WSINDy can be employed to learn multiple relevant models to describe the system in different regimes. The bulk of our findings is presented in the Section \ref{sec:numerics}, where we quantify the performance of WSINDy applied to four physically-relevant nearly-periodic Hamiltonian systems of varying dimension.

\section{Preliminaries}
\label{sec:prelims}

In this Section we review Hamiltonian dynamical systems with specific attention paid to nearly-periodic systems in Section \ref{sec:nearperreview}.

\subsection{Hamiltonian systems}\label{sec:Hsys}

Classically, Hamilton's equations describe the evolution of a point $z=(q,p)$ in phase space $M=\Rbb^{2N}$ along a level set a function $H: M\to \Rbb$ referred to as the {\it Hamiltonian}. Hamilton's original equations are
\begin{equation}\label{HamiltonsEq}
\begin{dcases} \dt{q} = \nabla_p H \\ \dt{p} =- \nabla_q H,\end{dcases}
\end{equation}
and $(q,p)$ were originally associated with the position and momentum of a particle (the dot notation denoting the derivative with respect to time). It can readily be seen from equation \eqref{HamiltonsEq} that along any trajectory $t\to (q(t),p(t))$, $H$ is conserved: $\frac{d}{dt}H(q(t),p(t)) = 0.$
By defining the matrix
\begin{equation}\label{J}
\Jbf = \begin{pmatrix} 0 & {Id}_{\Rbb^N} \\ -{Id}_{\Rbb^N} & 0 \end{pmatrix},
\end{equation}
where ${Id}_{\Rbb^N}$ is the identity in $\Rbb^N$, we can equivalently write \eqref{HamiltonsEq} as $\dt{z} = X_H(z),$
where the {\it Hamiltonian vector field} $X_H$ is defined
\[X_H(z) = \Jbf\nabla H(z).\]
The matrix $\Jbf$ is nonsingular, anti-symmetric, and of even dimension, in other words it is {\it symplectic}. It is this symplectic structure that allows for a significantly more general formulation of Hamiltonian dynamics on arbitrary smooth manifolds using the language of differential forms. For a comprehensive review of the subject, see textbooks\cite{marsden2013introduction,AbrahamMarsden1978} and the exposition\cite{MacKay2020JPlasmaPhys} on differential forms in plasma physics. For the purposes of introducing a widely-applicable weak formulation, we will briefly describe general Hamiltonian systems starting with the following definition. 

\begin{defn}[Symplectic manifold] The pair $(M,\Omega)$ is a {\normalfont{symplectic manifold}} if $M$ is a smooth manifold and $\Omega$ is a closed, nondegenerate differential 2-form on $M$.
\end{defn}

That is, $d\Omega = 0$ ($\Omega$ is closed, see \cite{marsden2013introduction,AbrahamMarsden1978,MacKay2020JPlasmaPhys} for more details on the exterior derivative $d$) and for every $z\in M$, $\Omega_z$ is a bilinear map on the product tangent space at $z$, $T_zM \times T_zM$, that is anti-symmetric ($\Omega_z(X,V) = -\Omega_z(V,X)$ for all $X,V\in T_zM$) and non-degenerate ($ \Omega_z(X,V) = 0$ for all $V\in T_zM$ implies $X=0$).
Non-degeneracy and anti-symmetry imply that the dimension of $M$ must be even. If $\Omega$ is allowed to degenerate (i.e.\ admitting a null space) then $(M,\Omega)$ is referred to as a {\it presymplectic manifold}, which need not have even dimension. However, throughout we assume that $M$ has dimension $2N$ for some $N\geq 1$. We then classify a Hamiltonian system as follows.

\begin{defn}[Hamiltonian system]
Let $(M,\Omega)$ be a $(2N)$-dimensional presymplectic manifold and $H : M \to \Rbb$ a smooth function. The Hamiltonian vector field $X_H$ associated to $(M,\Omega,H)$ is defined by 
\begin{equation}\label{Hdef}
\iota_{X_H}\Omega = dH.
\end{equation}
The tuple $(M,\Omega,H,X_H)$ is referred to as a \normalfont{Hamiltonian system}.
\end{defn}

\noindent In \eqref{Hdef}, $dH$ is the differential of $H$ and $\iota$ denotes the interior product, that is, for a differential $k$-form $\tau$ and vector field $X$ on $M$, $\iota_{X}\tau$ is a $(k-1)$-form given by $(\iota_{X}\tau)_z(V_1,\dots,V_{k-1}) = \tau_z(X,V_1,\dots,V_{k-1})$ for all $V_1,\dots,V_{k-1}\in T_zM$. Equation \eqref{Hdef} will be referred to as {\it Hamilton's equations}, representing a generalization of Hamilton's original equations \eqref{HamiltonsEq} to arbitrary presymplectic manifolds $(M,\Omega)$. Since $\Omega$ is bilinear, Hamilton's equations \eqref{Hdef} can be written
\[\Omega(X_H,V) = dH(V), \quad \forall V\in TM.\]
In this way, $X_H$ is defined implicitly at each point $z\in M$ through the action of elements $\Omega_z(\cdot,V)$ in $(T_zM)^*$, the dual of the tangent space $T_zM$. In the Euclidean setting ($M = \Rbb^{2N})$, if $\Omega$ is symplectic, we can associate $\Omega$ with a quadratic form and use the Euclidean inner product (dot product) to write
\[\Omega_z(X_H,V) = (\Jbf_z^{-1}X_H)\cdot V, \quad dH(V) = \nabla H\cdot V\]
where $\Jbf_z$ is a symplectic matrix for each $z\in M$, and further, supressing the $z$-dependence,
\begin{equation}\label{HdefEuc}
V\cdot X_H = V\cdot \Jbf \nabla H, \quad \forall V\in TM.
\end{equation}
This will play a role in the weak formulation below.

\subsection{Nearly-periodic Hamiltonian systems}\label{sec:nearperreview}
A {\it nearly-periodic} system is a dynamical system that depends smoothly on a small parameter $\vep$ and is periodic in the limit $\vep\to 0$. A nearly-periodic {\it Hamiltonian} system is nearly periodic and is Hamiltonian for all $\vep\geq 0$. In practice, the symplectic form may depend on $\vep$ and degenerate at $\vep=0$, hence why recent developments in nearly-periodic system theory\cite{BurbySquire2020JPlasmaPhys,Burby_Hirvijoki_2021,BHL_2023} work with presymplectic, rather than symplectic Hamiltonian systems (see Examples 3 and 4 below). Nearly-periodic Hamiltonian systems are a prototypical class of dynamics to explore how approximate symmetries may be used to reduce the dimensionality of physical systems. An important example relevant to plasma physics is the dynamics of a charged particle in a strong magnetic field, which is emulated below in Examples 3 and 4.  

First, we introduce several definitions. The limiting symmetry in a nearly-periodic system is defined using a {\it circle action}.

\begin{defn}[Circle Action]
A 1-parameter family of diffeomorphisms $\{\Phi_\theta\ :\ \theta\in \Rbb\}$, on a manifold $M$ is a {\normalfont circle action} if $\Phi_{\theta+2\pi} = \Phi_\theta$, $\Phi_0 = \text{Id}_M$, and $\Phi_{\theta_1+\theta_2} = \Phi_{\theta_1}\circ\Phi_{\theta_2}$.
\end{defn}

\noindent Change with respect to an underlying vector field on $M$ is defined using the {\it Lie derivative}.

\begin{defn}[Lie Derivative]\label{liederiv}
Let $X$ be a vector field on $M$ with associated flow map ${\Phi_t}$ such that $\frac{d}{dt}{\Phi_t}(z) = X({\Phi_t}(z))$ for all $z\in M$. The Lie derivative of a tensor field $\tau$ on $M$ with respect to $X$ is defined by 
\[\CalL_X\tau(z) = \frac{d}{dt}\Big\vert_{t=0}\Phi_t^*\tau(z)\]
where $(\cdot)^*$ is the pullback operation (see e.g.\ \cite[Def.\ 1.7.16]{AbrahamMarsden1978}).
\end{defn}
\noindent For instance, if $f$ is a function on $M$, then $\CalL_Xf(z)$ is the directional derivative of $f$ at $z$ in the direction $X(z)$. If $V$ is a vector field on $M$, then $\CalL_X V$ is the Lie Bracket of $X$ and $V$, given by $\partial_X V - \partial_V X$, with $\partial$ denoting the directional derivative. If $\tau$ is a rank-$k$ differential form, then the Lie derivative is given by Cartan's formula, 
\begin{equation}\label{cartan}
\CalL_X\tau= \iota_X d\tau + d(\iota_X\tau).
\end{equation}
We will occasionally denote vector fields in differential operator notation, in other words with local coordinates $z=(z_1,\dots,z_{2N})$ for $M$ the vector field $V(z) = (V_1(z),\dots,V_{2N}(z))$ may be written $V(z) = \sum_{j=1}^{2N} V_j(z)\partial_{z_j}$. As well as providing a compact notation, this emphasizes the action of vector fields on functions through the $\CalL_Vf = df(V)$.

A nearly periodic Hamiltonian System is then defined as follows.

\begin{defn}[Nearly periodic Hamiltonian System]\label{def:nphs}
An $\vep$-dependent Hamiltonian system $(M,\Omega_\vep,H_\vep,X_\vep)$ is nearly periodic if there exists a function $\omega_0:M\to \Rbb$ and a circle action $\Phi_\theta$ such that
\begin{enumerate}
\item $\Omega_\vep,H_\vep$ depend smoothly on $\vep$
\item $X_0 = \omega_0R_0$ where $R_0$ is the infinitesimal 
generator of a circle action $\Phi_\theta$, i.e.\ $R_0 = \frac{d}{d\theta}\Big\vert_{\theta=0} \Phi_\theta$
\item The limiting angular frequency $\omega_0$ is strictly positive and satisfies $\CalL_{R_0}\omega_0 = 0$.
\end{enumerate}
\end{defn}

For a nearly-periodic Hamiltonian system $(M,\Omega_\vep,H_\vep,X_\vep)$, smoothness in $\vep$ implies that for sufficiently small $\vep$, $\Omega_\vep$ and $H_\vep$ have formal asymptotic expansions 
\[H_\vep(z) = H_0(z) +\vep H_1(z) +\vep^2 H_2(z)+\cdots\]
\[(\Omega_\vep)_z = (\Omega_0)_z +\vep (\Omega_1)_z +\vep^2 (\Omega_2)_z + \cdots\]
From Hamilton's equations \eqref{Hdef}, we find that the Hamiltonian vector $X_\vep$ has a formal power series $X_\vep = X_0+\vep X_1 + \vep^2X_2+\cdots$ with coefficient vector fields given by the infinite family of equations
\begin{align}
\label{leading_order_ham} \iota_{X_0}\Omega_0 &= dH_0 \\
\label{first_order_ham}\iota_{X_1}\Omega_0+\iota_{X_0}\Omega_1 &= dH_1 \\
\label{second_order_ham}\iota_{X_2}\Omega_0+\iota_{X_1}\Omega_1+\iota_{X_0}\Omega_2 &= dH_2 \\
 & \vdots \nonumber
\end{align}
In 1962 M. Kruskal\cite{Kruskal1962JournalofMathematicalPhysics} proved that every nearly periodic system possesses an approximate $U(1)$ symmetry given by a unique vector field $R_\vep$ referred to as the {\it roto-rate}, and defined as follows. 

\begin{defn}[Roto-rate]
The {\normalfont roto-rate} of a nearly periodic vector field $X_\vep$ is a formal power series $R_\vep = R_0 + \vep R_1 +\vep^2R_2+\cdots$ with vector field coefficients $R_k$ such that $R_0 = \omega_0^{-1}X_0$ and to all orders in $\vep$, 
\begin{enumerate}
\item $\CalL_{R_\vep}X_\vep = 0$ 
\item The integral curves of $R_\vep$ are $2\pi$-periodic.
\end{enumerate}
\end{defn}

\noindent Kruskal also showed that, in the Hamiltonian case, $R_\vep$ is itself Hamiltonian, with Hamiltonian $\mu_\vep$ known as the {\it adiabatic invariant}. The adiabatic invariant is defined as a formal power series 
\[\mu_\vep = \mu_0+\vep\mu_1+\vep^2\mu_2+\cdots.\]
Explicit formulas for the first few terms in the series were first found by Burby and Squire\cite{BurbySquire2020JPlasmaPhys}.
Since $R_\vep$ commutes with $X_\vep$ to all orders in $\vep$ ($\CalL_{R_\vep}X_\vep=0$), $\mu_\vep$ is formally conserved by $X_\vep$, hence the notion of $R_\vep$ as an approximate symmetry of $X_\vep$.

\subsection{Symmetry reduction in nearly-periodic Hamiltonian systems}

In the presence of an exact continuous symmetry $t\to {\Phi_t}$ of a Hamiltonian system $(M,\Omega,H,X_H)$, such that $\CalL_RX_H=0$ and $\iota_{R}\Omega = d\mu_*$, where $R$ is the infinitesimal generator of ${\Phi_t}$, the phase space $(M,\Omega)$ and dynamics $(H,X_H)$ may be reduced using the Marsden-Weinstein-Meyer construction (see \cite[Chapter 4]{AbrahamMarsden1978}). At a high level, this proceduring involves (a) restricting to a level set $\Lambda_\mu=\mu_*^{-1}(\mu)$ of the conserved quantity $\mu_*$ associated with ${\Phi_t}$, (b) identifying points on orbits of ${\Phi_t}$ to form the reduced phase space $\Mred$, and (c) identifying a suitable symplectic form $\Ored$ on $\Mred$. Existence of $\Ored$ on the reduced phase space $\Mred$ is at the heart of this result. On the other hand, the reduced Hamiltonian $\Hred$ on $\Mred$ is simply defined since $H$ is, by assumption, constant along $R$-orbits.
We note that other analytical methods of reducing Hamiltonian systems do exist, including those related to normal-form theory \cite{ChurchillKummerRod1983JournalofDifferentialEquations}, which involves successive near-identity coordinate transformations.

In the setting of an approximate symmetry, to derive a reduced Hamiltonian $\Hred$ the procedure is similar to Marsden-Weinstein-Meyer, but with substantially more effort for higher orders in $\vep$. For nearly-periodic systems one is guaranteed a reduction in dimension of at least 2, as (1) we restrict to a level set of the adiabatic invariant $\mu_\vep$, and (2) we condense phase space by forming an equivalence class of points that lie on the same integral curve of the roto-rate $R_\vep$. However, the reduction in dimension can be greater, as demonstrated below Section \ref{sec:numerics}, Example 4.

To connect theory and practice, it is useful to draw parallels between geometric reduction of Hamiltonian systems, such as the Marsden-Weinstein-Meyer construction, and dynamic reduction involving averages and restrictions at the level of vector fields. In the next section we place the latter (dynamic) reduction technique on firmer ground by providing new theoretical results that ensure first-order averaging at the level of $X_\vep$ preserves Hamiltonian structure. These theoretical results complement existing results on nearly periodic systems in \cite{BurbySquire2020JPlasmaPhys,Burby2022LA-UR-22-265241875767}, and provide motivation for weak-form equation learning from trajectories of nearly-periodic systems. Following is a discussion drawing parallels to classical time-averaging in fast-slow systems. 

\section{Hamiltonian structure of first-order averaging theory}\label{sec:1stOrderHtheory}
Let $X_\vep = \omega_0\,R_0 + \vep\,X_1 + \dots$ be a nearly-periodic Hamiltonian system on the exact presymplectic manifold $(M,\Omega_\vep)$ with Hamiltonian $H_\vep$ and limiting roto-rate $R_0$. Recall that presymplectic means $\Omega_\vep$ is antisymmetric and closed, $d\Omega_\vep =0$, but possibly degenerate. Exact means there is a $1$-form $\vartheta_\vep$ such that $\Omega_\vep = -d\vartheta_\vep$. Hamiltonian means $\iota_{X_\vep}\Omega_\vep = dH_\vep$, which implies the sequence of equations \eqref{leading_order_ham}-\eqref{second_order_ham}. The flow map for $R_0$ will be denoted $\Phi_\theta^0$, where $\theta\in\mathbb{R}\text{ mod }2\pi$ is the time parameter. Given a tensor $\tau$ on $M$ its $R_0$-average will be denoted by
\begin{equation}\label{R0avg}
\langle \tau\rangle = (2\pi)^{-1}\int_0^{2\pi}(\Phi_\theta^0)^*\tau\,d\theta
\end{equation}
that is, the time-average of the pullback of $\tau$ by $\Phi^0_\theta$.

Generic (e.g.\ possibly non-Hamiltonian) first-order averaging theory approximates the vector field $X_\vep$ with the {\it first-order average vector field} 
\[\Xav = X_0 + \vep\,\langle X_1\rangle.\]
The first-order average is automatically $R_0$-invariant, $\CalL_{R_0} \Xav = 0$, meaning the fast phase corresponding to motion along $X_0 = \omega_0\,R_0$ is ignorable after the replacement. We will show that the reduced dynamics defined by $\Xav$ and obtained by ignoring the fast phase is presymplectic Hamiltonian under the assumptions in the previous paragraph. Although the all-orders averaging theory of Kruskal \cite{
Kruskal1962JournalofMathematicalPhysics} is known to be Hamiltonian, the Hamiltonian structure underlying first-order averaging has never been identified in full generality. Since first-order averaging is sufficient for analyzing various important examples, including those considered in this work, a self-contained description here is warranted.

First we recall some useful results related to the adiabatic invariant $\mu_\vep$ from the general theory of nearly-periodic Hamiltonian systems on presymplectic manifolds\cite{BurbySquire2020JPlasmaPhys}. Adiabatic invariance means $\mathcal{L}_{X_\vep}\mu_\vep = 0$ in the sense of formal power series. There are general formulas for the first few terms in the series $\mu_\vep$, the simplest being
\begin{align}
\mu_0 = \iota_{R_0}\langle \vartheta_0\rangle,\label{mu0_formula}
\end{align}
where $\vartheta_0$ denotes the first term in the formal power series expansion for $\vartheta_\vep = \vartheta_0 + \vep\,\vartheta_1 + \dots$. Observe that the exterior derivative of Eq.\,\eqref{mu0_formula} justifies thinking of the leading-order adiabatic invariant as ``energy over frequency,"
\begin{align}
d\mu_0 = \iota_{R_0}d\langle\vartheta_0 \rangle = \omega_0^{-1}\iota_{X_0}\langle \Omega_0\rangle = \omega_0^{-1}\,dH_0,\label{energy_period_relation}
\end{align}
where we have used Eq.\,\eqref{leading_order_ham} to infer $\langle H_0\rangle = H_0$.
The $R_0$-average of $\mu_1$ is also simple:
\begin{align}
\langle \mu_1\rangle = \iota_{R_0}\langle \vartheta_1\rangle.\label{mu1_av_formula}
\end{align}
In general, some number of the first terms in the series $\mu_\vep$ vanish. Let $k$ denote the smallest integer such that $\mu_k$ is not identically $0$. Set $\mu^* = \mu_k$. The condition $\mathcal{L}_{X_\vep}\mu_\vep = 0$ is equivalent to the sequence of equations
\begin{align}
\mathcal{L}_{X_0}\mu^*& = 0\label{leading_order_mu}\\
\mathcal{L}_{X_1}\mu^* + \mathcal{L}_{X_0}\mu_{k+1} & = 0\label{first_order_mu}\\
\vdots\nonumber
\end{align}

Next we prove a simple technical lemma that establishes structural properties of the first-order averaged system related to its presymplectic Hamiltonian formulation.
\begin{lemma}\label{first_order_structure_lemma}
The first non-zero coefficient $\mu^*$ in the adiabatic invariant series is constant along solutions of the first-order averaged system $\Xav$,
\begin{align}
\mathcal{L}_{\Xav}\mu^* = 0.\label{av_ad_inv}
\end{align}
In addition, the first-order averaged system satisfies the Hamilton-like equation
\begin{align}
\iota_{\Xav}\Omega_0=dH_0 + \vep\,d\langle H_1\rangle - \vep\,\omega_0\,d\langle \mu_1\rangle.\label{av_ham_eqn_unreduced}
\end{align}
\end{lemma}
\begin{proof}
The conservation law \eqref{av_ad_inv} follows immediately from \eqref{leading_order_mu} and the $R_0$-average of \eqref{first_order_mu}. (Notice that $\langle \mathcal{L}_{X_0}\mu_1\rangle = \omega_0\,\langle \mathcal{L}_{R_0}\mu_1\rangle = 0$ and $\langle \mu^*\rangle = \mu^*$.) As for \eqref{av_ham_eqn_unreduced}, it follows directly from the $R_0$-average of \eqref{first_order_ham} and the identity
\begin{align*}
\iota_{X_0}\langle \Omega_1\rangle = -\omega_0\langle \iota_{R_0}d\vartheta_1\rangle = \omega_0\,d(\iota_{R_0}\langle\vartheta_1\rangle) = \omega_0\,d\langle \mu_1\rangle,
\end{align*}
which represents an application of Cartan's formula \eqref{cartan} and \eqref{mu1_av_formula}.
\end{proof}

Because $\mu^*$ is automatically conserved by the first-order averaged system we can study $\Xav$ one level set of $\mu^*$ at a time. Let $\Lambda_\mu = (\mu^*)^{-1}\{\mu\}$ be a regular level set for $\mu^*$ with regular value $\mu$. By Lemma \ref{first_order_structure_lemma} $\Xav$ restricts to a vector field $\XavLam$ on $\Lambda_\mu$. It turns out that $\XavLam$ satisfies a genuine Hamilton equation, even though $\Xav$ is not Hamiltonian in general due to the inexact $1$-form $\omega_0\,d\langle \mu_1\rangle$ in Eq.\,\eqref{av_ham_eqn_unreduced}.
\begin{lemma}\label{level_set_ham}
Let $i_\mu:\Lambda_\mu\rightarrow M$ denote the canonical inclusion map. Define the closed $2$-form $\OavLam =i_\mu^*\Omega_0$ and the function $\HavLam = i_\mu^*(H_0 + \vep\,\langle H_1\rangle - \vep\,\omega_0\langle \mu_1\rangle)$. The restricted first-order averaged system $\XavLam$ is $\OavLam$-Hamiltonian with Hamiltonian $\HavLam$,
\begin{align}
\iota_{\XavLam}\OavLam = d\HavLam.\label{res_ham}
\end{align}
\end{lemma}
\begin{proof}
Pulling back Eq.\,\eqref{av_ham_eqn_unreduced} along $i_\mu$ implies
\begin{align*}
\iota_{\XavLam}\OavLam = di_\mu^*(H_0 + \vep\,\langle H_1\rangle) - \vep\,i_\mu^*\bigg(\omega_0\,d\langle \mu_1\rangle\bigg).
\end{align*}
So the proof merely requires showing $i_\mu^*\bigg(\omega_0\,d\langle \mu_1\rangle\bigg) = di_\mu^*(\omega_0\,\langle \mu_1\rangle)$. We consider three cases separately.

\emph{case $1$:} $\mu^* = \mu_0$. In this case the energy-period relation \eqref{energy_period_relation} implies $d\omega_0 \wedge d\mu^* = d\omega_0\wedge d\mu_0= ddH_0 = 0$. Thus, $d\omega_0 = \lambda\,d\mu^*$, for some function $\lambda$. Pulling back along $i_\mu$ then implies $di_\mu^*\omega_0 = 0$, which says that $i_\mu^*\omega_0$ is constant. Since the exterior derivative is $\mathbb{R}$-linear, we therefore find
\begin{align*}
i_\mu^*\bigg(\omega_0\,d\langle \mu_1\rangle\bigg) = (i_\mu^*\omega_0)di_\mu^*\langle \mu_1\rangle = d\bigg((i_\mu^*\omega_0)i_\mu^*\langle \mu_1\rangle\bigg) = di_\mu^*(\omega_0\,\langle \mu_1\rangle),
\end{align*}
as desired.

\emph{case $2$:} $\mu^* = \mu_1$. In this case Eq.\,\eqref{leading_order_mu} implies $\mu^* = \mu_1 = \langle \mu_1\rangle$. Therefore
\begin{align*}
i_\mu^*\bigg(\omega_0\,d\langle \mu_1\rangle\bigg) = (i_\mu^*\omega_0)di_\mu^*\langle \mu_1\rangle = (i_\mu^*\omega_0)di_\mu^*\mu^* = 0,
\end{align*}
as desired.

\emph{case $3$:} $\mu^* = \mu_k$, $k >1$. In this case both $\mu_0$ and $\mu_1$ vanish identically. Therefore
\begin{align*}
i_\mu^*\bigg(\omega_0\,d\langle \mu_1\rangle\bigg) = i_\mu^*\bigg(\omega_0\,d0\bigg)=0,
\end{align*}
as desired.
\end{proof}

We are finally in position to establish our main claim: after ``forgetting" the fast phase in the first-order averaged system, the resulting slow dynamics is presymplectic Hamiltonian. The result follows from Lemma \ref{level_set_ham} and a variant of the Marsden-Weinstein symplectic quotient technique.

\begin{theorem}\label{thm_1stOrder}
Let $\Mred = \Lambda_\mu/U(1)$ denote the space of $R_\mu$-orbits in $\Lambda_\mu$. Assume it is a manifold. Let $\pi:\Lambda_\mu\rightarrow \Mred$ denote the corresponding quotient map. On $\Mred$, there is a unique vector field $\Xred$, a unique closed $2$-form $\Ored$, and a unique function $\Hred$ such that
\begin{align}
T\pi\circ \XavLam = \Xred\circ\pi, \quad \pi^*\Ored = \OavLam,\quad \pi^*\Hred = \HavLam.\label{reduction_formulas}
\end{align}
Moreover, $\Xred$ is $\Ored$-Hamiltonian with Hamiltonian function $\Hred$,
\begin{align}
\iota_{\Xred}\Ored = d\Hred.\label{reduced_ham}
\end{align}
\end{theorem}
\begin{proof}
By construction, the first-order averaged system commutes with $R_0$, $\CalL_{R_0}\Xav = 0$. By Eqs.\,\eqref{leading_order_mu} and \eqref{av_ad_inv} both $R_0$ and $\Xav$ are tangent to $\Lambda_\mu$. Therefore $\CalL_{R_\mu}\XavLam = 0$, which says that the $\XavLam$-flow maps $R_\mu$-orbits into $R_\mu$-orbits. It follows that there is a unique vector field $\Xred$ on orbit space $\Mred$ that is $\pi$-related to $\XavLam$.

In general, the interior product of $R_\mu$ with $\OavLam$ is given by $\iota_{R_\mu}\OavLam = i_\mu^*(\iota_{R_0}\Omega_0) = di_\mu^*\mu_0$. When $\mu^* = \mu_0$, $i_\mu^*\mu_0$ is constant and $\iota_{R_\mu}\OavLam = 0$. When $\mu^* = \mu_k$, $k>0$, $\mu_0  =0$ and, again, $\iota_{R_\mu}\OavLam = 0$. So $\iota_{R_\mu}\OavLam = 0$ in general. It follows that there is a unique $2$-form $\Ored$ on $\Mred$ such that $\pi^*\Ored = \OavLam$. It is straightforward to show that $d\Ored = 0$.

Since $\mathcal{L}_{R_\mu}\HavLam  =0$ there is a unique function $\Hred$ such that $\pi^*\Hred = \HavLam$. Overall, we have just established existence of the desired objects on $\Mred$ that satisfy Eq.\,\eqref{reduction_formulas}. Establishing the Hamilton equation \eqref{reduced_ham} is now simple. Notice that
\begin{align*}
\iota_{\XavLam}\OavLam = \iota_{\XavLam}\pi^*\Ored = \pi^*(\iota_{\Xred}\Ored) = \pi^*d\Hred,
\end{align*}
which says that the $1$-form $\pi^*(\iota_{\Xred}\Ored-d\Hred)$ on $\Lambda_\mu$ vanishes. But because $\pi$ is a surjective submersion $\pi^*$ is injective. It follows that $\iota_{\Xred}\Ored-d\Hred = 0$, as claimed. 
\end{proof}

\noindent From equations \eqref{av_ham_eqn_unreduced}, \eqref{res_ham}, and \eqref{reduced_ham}, we can identify two methods of obtaining the leading-order reduced Hamiltonian system $(\Mred,\Ored,\Hred,\Xred)$.  To clarify terminology, we refer to $(\Mred,\Ored,\Hred,\Xred)$ as the ``leading-order reduced'' system since it is obtained from utilizing the roto-rate $R_\vep$ and adiabatic invariant $\mu_\vep$ each to leading-order (while $\Xav$ is the ``first-order averaged'' vector field since it is obtained from averaging $X_\vep$ to first order, $\Xav= \langle X_0+\vep X_1\rangle$). Higher-order reduced systems can be obtained in a similar manner when given access to higher-order terms in $R_\vep$ and $\mu_\vep$, however the order-by-order Hamiltonian structure is as-of-yet undetermined, and we reserve it for a future work. We will differentiate the two methods of obtaining $(\Mred,\Ored,\Hred,\Xred)$ as {\it geometric reduction} and {\it dynamic reduction}:

\begin{enumerate}[label = (\Roman*)]
\item {\it Geometric reduction}: 
\begin{enumerate}
\item Average $(H_\vep,\Omega_\vep)$ around $R_0$ to get $(\Hav,\Oav)$, keeping terms to first-order in accordance with Hamilton's equations \eqref{leading_order_ham}-\eqref{second_order_ham} 
\item Restrict $(\Hav,\Oav)$ to the level set $\Lambda_\mu = \mu_*^{-1}\{\mu\}$ to obtain $\HavLam = i_\mu^*(\Hav-\vep\omega_0\langle \mu_1\rangle)$ and $\OavLam = i_\mu^*\Omega$, having adjusted the energy $\Hav$ according to the $\CalO(\vep)$ adiabatic invariant $\mu_1$
\item Quotient out $\Lambda_\mu$ by orbits of $\RLam = i_\mu^*R_0$ to get the reduced phase space $\Mred$, reducing $(\HavLam,\OavLam)$ through invariance to $\RLam$ to $(\Hred,\Ored)$
\item Obtain the reduced Hamiltonian vector field $\Xred$ on $\Mred$ using Hamilton's equations applied to $(\Hred,\Ored)$ 
\end{enumerate}
\item {\it Dynamic reduction}:
\begin{enumerate}
\item Obtain $\Xav$ by averaging the first-order vector field $X_0+\vep X_1$ around the flow of $R_0$
\item Notice that for regular values of $\mu_*$, $\Xav$ is tangent to $\Lambda_\mu$, hence $\XavLam$ is obtained by observing $\Xav$ on $\Lambda_\mu$
\item Represent the quotient space $\Mred$ as a section from $\Lambda_\mu$, that is, let each point in $\Mred$ correspond to a unique $\RLam$ orbit lying in $\Lambda_\mu$
\item $\XavLam$ observed on this section equals $\Xred$ and is Hamiltonian according to \eqref{reduced_ham}
\end{enumerate}
\end{enumerate}

\noindent Geometric reduction is more robust and allows for systematic exploration of higher-order reductions in $\vep$. On the other hand dynamic reduction merely involves observing the system on the reduced submanifold $\Mred$, after averaging around the flow of the limiting roto-rate $R_0$. In physical systems, it can be assumed that the system is already close to this manifold, hence some method of averaging observations is all that is needed to observe the Hamiltonian structure of the system to first order. We exploit this property in the current work, whereby the weak form emulates the averaging necessary to reveal Hamiltonian structure in the dynamics. 

\subsection{Connection with classical time-averaging}

One case of interest is when $M=\Rbb^{2N}$ and the dynamics can be partitioned into slow and fast modes, $z = (z_s,z_f)\in \Rbb^{2(N-1)}\times\Rbb^{2}$, with Hamiltonian $H_\vep$ given by 
\[H_\vep(z) = H_0(z_f) + \vep H_1(z_s,z_f).\]
 As in Examples 1 and 2 below, the first-order averaged Hamiltonian then takes the form, for some $\tilde{H}:\Rbb^{2(N-1)}\times\Rbb\to \Rbb$, 
 \[\Hav(z) = H_0(z_f) + \vep \tilde{H}(z_s,\mu_0(z_f))\]
where $\mu_*=\mu_0$ is the leading-order adiabatic invariant. When the symplectic form $\Omega_\vep$ satisfies $\Omega_\vep = \Omega_0$, the first-order averaged dynamics $\Xav$ are already Hamiltonian, given by $\iota_{\Xav} \Omega_0 = d\Hav$. If we further have $\iota_{V_s}\iota_{V_f}\Omega_0=0$ for any pair of vector fields $V_s$, $V_f$ tangent to the slow and fast submanifolds, respectively, then the dynamics $\dot{\tilde{z}} = \Xav(\tilde{z})$ take the form
\begin{equation}\label{slowfastdecouple}
\begin{dcases} \dt{\tilde{z}}_s = \vep \Jbf_s\nabla_{z_s} \tilde{H}(\tilde{z}_s,\mu_0(\tilde{z}_f)) \\
\dt{\tilde{z}}_f = \Jbf_f \Big(\nabla_{z_f}H_0(\tilde{z}_f)\ +\ \vep \partial_2 \tilde{H}(\tilde{z}_s,\mu_0(\tilde{z}_f))\nabla_{z_f} \mu_0(\tilde{z}_f)\Big)\end{dcases}
\end{equation}
with $\partial_2$ denoting differentiation with respect to the second argument and symplectic matrices $\Jbf_s$, $\Jbf_f$ derived from $\Omega_0$ for the slow and fast subsystems. While we have not yet restricted to a level set of $\mu_0$, for any trajectory $\tilde{z}=(\tilde{z_s},\tilde{z}_f)$ of \eqref{slowfastdecouple} it holds that
\[\frac{d}{dt}\mu_0(\tilde{z}(t))\ =\ \nabla_{z_f}\mu_0(\tilde{z}_f)\cdot\Jbf_f \nabla_{z_f}H_0(\tilde{z}_f)\ +\  \left[\nabla_{z_f}\mu_0(\tilde{z}_f)\cdot\Jbf_f \nabla_{z_f} \mu_0(\tilde{z}_f)\right]\left(1+\vep \partial_2 \tilde{H}(\tilde{z}_s,\mu_0(\tilde{z}_f))\right)\ =\ 0\]
where the first term is zero because $\mu_0$ is conserved by the flow of $X_0$, or
\[\nabla_{z_f}\mu_0(\tilde{z}_f)\cdot\Jbf_f \nabla_{z_f}H_0(\tilde{z}_f) = \nabla_{z_f}\mu_0(\tilde{z}_f)\cdot X_0(\tilde{z}) = \CalL_{X_0}\mu_0 (\tilde{z}) = 0,\]
    and the second term is zero because $\Jbf_f$ is symplectic. From the dependence of $\tilde{H}$ on $\mu_0$, conservation of $\mu_0$ along this reduced flow implies that the slow system $\tilde{z}_s$ {\it decouples from the fast system} $\tilde{z}_f$. Reducing phase space according to $M\to\Lambda_\mu\to \Mred$ is then just restriction to the slow variables $z_s$. The reduced Hamiltonian is then given by, up to a constant depending on $\mu$, $\Hred =  \vep\tilde{H}(z_s,\mu)$.

We can now make some connections with classical time-averaging. If we instead represent the dynamics of the original slow variables $z_s$ as the non-autonomous system
\begin{equation}\label{classicalaveraging}
\dt{z}_s = \vep F(z_s,t),
\end{equation}
suppressing explicit dependence on $z_f$ through the second argument of $F$, then the previous procedure of first-order averaging by the limiting roto-rate $R_0$ is similar to classical averaging techniques \cite[Ch.2]{2007}. Specifically, for flow-map $\Phi^0_t$ of $R_0$, it holds that $z_f(t) = (\Phi^0_t z(0))_f+\CalO(\vep)$, and we see that $F$ is $\CalO(\vep)$ close to being $2\pi$-periodic in its second variable. The time-averaged {\it autonomous} system 
\[\dt{\overline{z}} = \vep\overline{F}(\overline{z}), \qquad \overline{F}(z) := \frac{1}{2\pi}\int_0^{2\pi} F(z,t)\,dt\]
then stays $\CalO(\vep)$ close to the slow portion of the averaged system \eqref{slowfastdecouple} (and hence to the original system) on timescales of $\CalO(\vep^{-1})$.

While classical time-averaging is very useful for systems \eqref{classicalaveraging} with $F$ exactly $2\pi$-periodic in its second argument, this is the extent of its utility, and its application to nearly-periodic systems is severely limited. The fast mode $z_f$ is only $2\pi$-periodic when $\vep=0$, such that for any $\vep>0$, the slow and fast systems are coupled, with the dynamics for $z_s$ not falling into the class \eqref{classicalaveraging} with $F$ periodic in $t$. Deriving the correct analytical formulas for reductions of nearly-periodic systems is a more subtle and labor-intensive task (if not intractable), especially if a Hamiltonian structure is to be preserved (as outlined in the previous section), but significant progress has been made in this direction \cite{BurbySquire2020JPlasmaPhys,Burby2022LA-UR-22-265241875767,BHL_2023}. 

The aim of the current work is to complement these analytical techniques with computational methods to learn the correct reduced Hamiltonian systems using weak-form equation learning. Given the intricacies of reducing a nearly periodic system by means of the roto-rate $R_\vep$, it is surprising that the weak form recovers the same reduced dynamics from discrete-time trajectory data. In the next section we describe how WSINDy may be employed in this context. 

\section{WSINDy for Hamiltonian systems}\label{sec:wsindyH}

In this section we describe how WSINDy may be used to identify the Hamiltonian $H$ defining a Hamiltonian system using discrete-time observations of the flow of $X_H$. For the remainder of this article we will work in $M=\Rbb^{2N}$, although many of the concepts that follow have direct extensions to general manifolds. We will also assume that the symplectic form $\Omega$ is known, and leave identification of $\Omega$ and extension to general manifolds to future work. 

We can define a weaker version of \eqref{HdefEuc} (i.e.\ requiring less regularity of the resulting integral curves of $X_H$) by considering a trajectory $z:[0,T]\to \Rbb^{2N}$ satisfying 
\[\dt{z} = X_H(z) :=\Jbf(z) \nabla H(z)\]
where $\Jbf^{-1}(z)$ is the symplectic matrix associated with $\Omega_z$. Using integration by parts in time against a smooth time-dependent vector field $V(z,t)$ satisfying $V(z(0),0) = V(z(T),T) = 0$, we have
\begin{equation}\label{wfHdef}
-\int_0^T \left(\frac{d}{dt}{V}(z(t),t)\right)\cdot z(t)\,dt = \int_0^T V(z(t),t)\cdot\Jbf(z(t)) \nabla H(z(t))\,dt.
\end{equation}
If $V$ has no explicit $z$-dependence, the left-hand side can be evaluated without differentiating $z$, leading to 
\begin{equation}\label{wfHdef_noz}
-\int_0^T \dt{V}(t)\cdot z(t)\,dt = \int_0^T V(t)\cdot\Jbf(z(t)) \nabla H(z(t))\,dt.
\end{equation}
On the other hand, if $V$ depends on $z$, we can use the dynamics $\dt{z} = X_H(z)$ and rearrange terms to get
\[-\int_0^T\partial_tV(z(t),t)\cdot z(t)\,dt = \int_0^T\left(V(z(t),t) + z\cdot \nabla V(z(t),t)\right)\cdot\Jbf(z(t)) \nabla H(z(t))\,dt.\]
In either case, the phase space variables $z$ are not differentiated. This is crucial to accurate identification of $H$ from time series data that is corrupted from measurement noise. We will show here that this weak formulation also serves to filter out intrinsic dynamics that occur on a faster timescale. We restrict ourselves to the simpler case \eqref{wfHdef_noz} in this work and leave full exploration of \eqref{wfHdef} with general test vector fields $V = V(z,t)$ to future research.

To identify $H$ from data, we consider noisy evaluations  $\Zbf = z(\tbf)+\ep$ where $t\to z(t)$ is a trajectory from some Hamiltonian system $(\Rbb^{2N},\Omega,H,X_H)$, $\tbf = (t_0,\dots,t_{m+1})$ are the timepoints, and $\ep$ represents i.i.d.\ mean-zero measurement noise with variance $\sigma^2<\infty$. We approximate $H$ by expanding in terms of a chosen library $\Hbb=(H_1,\dots,H_J)$ of $J$ trial Hamiltonian functions, 
\[\widehat{H} = \sum_{j=1}^J\what_jH_j:=\Hbb\what\]
and we make the assumption that $\what$ is sparse. To solve for $\what$, we first discretize \eqref{wfHdef_noz} using $K$ test vector fields $\Vbb = (V_1,\dots,V_K)$ to arrive at the linear system $(\bbf,\Gbf)\in \Rbb^K\times \Rbb^{K\times J}$ defined by
\[\bbf_k = -\lan \dt{V}_k, \Zbf\ran_{\tbf}, \quad \Gbf_{kj} = \lan V_k, \Jbf(\Zbf)\nabla H_j(\Zbf)\ran_{\tbf}.\] 
Here, $\lan \cdot,\cdot\ran_\tbf$ defines a discrete inner product on $\Rbb^{2N}$-valued functions of time using $\tbf$ as quadrature nodes. For simplicity and previously demonstrated benefits \cite{MessengerBortz2021JComputPhys,MessengerBortz2021MultiscaleModelSimul}, we use the trapezoidal rule throughout, so that
\[\lan V,X\ran_\tbf := \sum_{i=0}^{m} \frac{\Delta t_i}{2}\left({V_i}\cdot X_i+{V_{i+1}}\cdot X_{i+1}\right)\]
where $V_i := V(t_i)$ and $\Delta t_i := t_{i+1}-t_i$. For compactly supported $V$ (or $X$) in time, this reduces to
\[\lan V,X\ran_\tbf = \sum_{i=1}^{m} \left(\frac{\Delta t_i+\Delta t_{i-1}}{2}\right){V_i}\cdot X_i.\]
In this work, we use the convolutional approach as in\cite{MessengerBortz2021JComputPhys}, that is, we fix a reference test function $\phi\in C^\infty_c(\Rbb)$ supported on $[-T_\phi/2,T_\phi/2]$ for some radius $T_\phi$, and we set the test vector fields to
\[V_k(z,t) = \phi(t-t_k)\sum_{j=1}^{2N} \partial_{z_j}, \qquad 1\leq k\leq K\]
for a fixed set of {\it query timepoints} $\CalQ := \{t_k\}_{k=1}^K$. We then solve the sparse regression problem
\begin{equation}\label{eq:SR}
\what = \argmin_\wbf \|\Gbf\wbf-\bbf\|_2^2+\lambda^2\|\wbf\|_0,
\end{equation}
where in this work we use the MSTLS algorithm to solve \eqref{eq:SR} (modified Sequential Thresholding Least Squares \cite{MessengerBortz2021JComputPhys}), which uses sequential thresholding on the combined term $\|\Gbf_j\wbf_j\|_2/\|\bbf\|_2$ and coefficient magnitudes $|\wbf_j|$. In addition, MSTLS performs a grid search for a suitable value of $\lambda$ (see \cite{MessengerBortz2021JComputPhys} for more details and \cite{BruntonProctorKutz2016ProcNatlAcadSci} for the original STLS algorithm). More information on MSTLS is provided in Appendix \ref{app:MSTLS}.
 
Data from multiple trajectories can easily be incorporated to improve the recovery process. When $L$ trajectories $\Zbf = (\Zbf^{(1)},\dots,\Zbf^{(L)})$ are available, with samples $\Zbf^{(\ell)} = z^{(\ell)}(\tbf^{(\ell)})+\vep$, the linear system $(\Gbf,\bbf)$ is formed by vertically concatenating the linear systems from each trajectory, $\Gbf = [(\Gbf^{(1)})^T\ |\ \cdots\ |\ (\Gbf^{(L)})^T ]^T$, $\bbf = [(\bbf^{(1)})^T\ |\ \cdots\ |\ (\bbf^{(L)})^T ]^T$. Note that the test vector fields in this case need not be the same for each trajectory. In the current work however, we focus on demonstrating recovery of a suitable reduced Hamiltonian $\Hred$ using only a single trajectory.

\subsection{WSINDy for Hamiltonian Systems with Approximate Symmetries}\label{sec:WSINDy4HS}

\begin{figure}
\begin{center}
\begin{tabular}{|@{}c@{}c@{}|}
\hline     \includegraphics[trim={50 0 50 20},clip,width=0.38\textwidth]{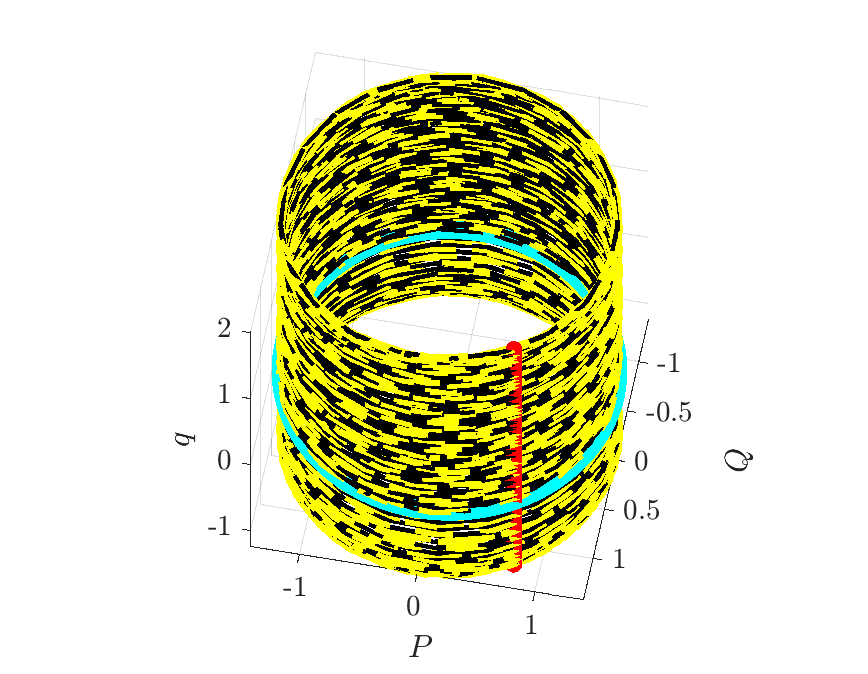} & 
     \includegraphics[trim={0 0 0 0},clip,width=0.5\textwidth]{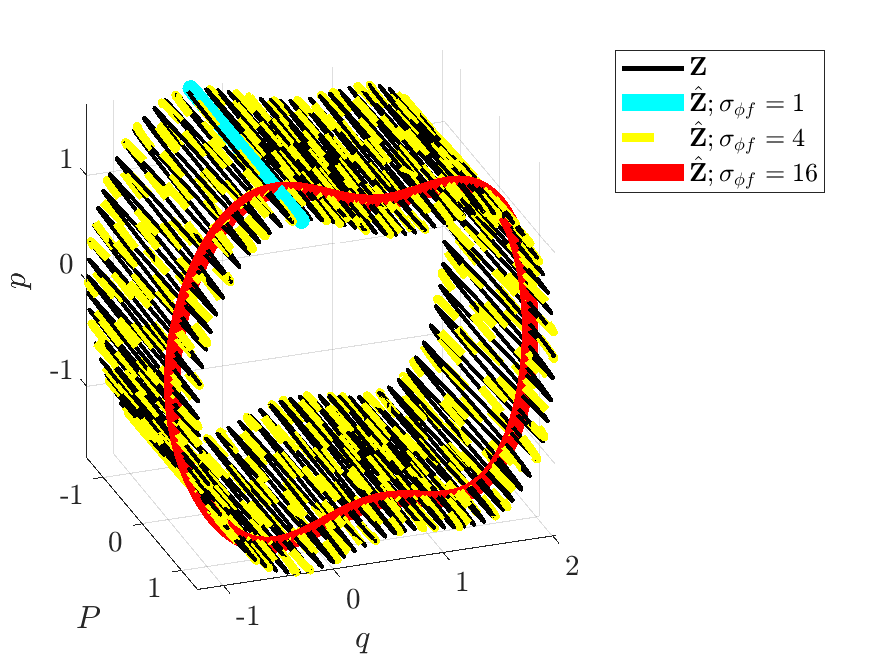} \\ \hline
\end{tabular}
\end{center}
\caption{Trajectories $\widehat{\Zbf}$ (cyan, yellow, red) learned using WSINDy applied to data $\Zbf$ (black) from the nearly-periodic system \eqref{coupledosc_dyn} in Example 1 with $\vep=0.01$, $1\%$ noise, and $z(0)$-index 7 (see Figure \ref{fig:mt}), for $\sigma_{\phi f}\in\{1,4,16\}$ (eq. \eqref{sigmas2}). Left and right: $(Q,P,q)$ and $(P,q,p)$ subdomains of phase space $(Q,P,q,p)\in \Rbb^4$. As $\sigma_{\phi f}$ increases, the learned Hamiltonian transitions from $H_0$ (cyan), to $H_\vep$ (yellow), to $\Hred$ (red), identifying the limiting roto-rate, the full system, and the leading-order reduced dynamics.}
\label{fig:full_sims}
\end{figure}

\begin{figure}
\begin{center}
\begin{tabular}{@{}c@{}c@{}}
     \includegraphics[trim={75 25 80 10},clip,width=0.5\textwidth]{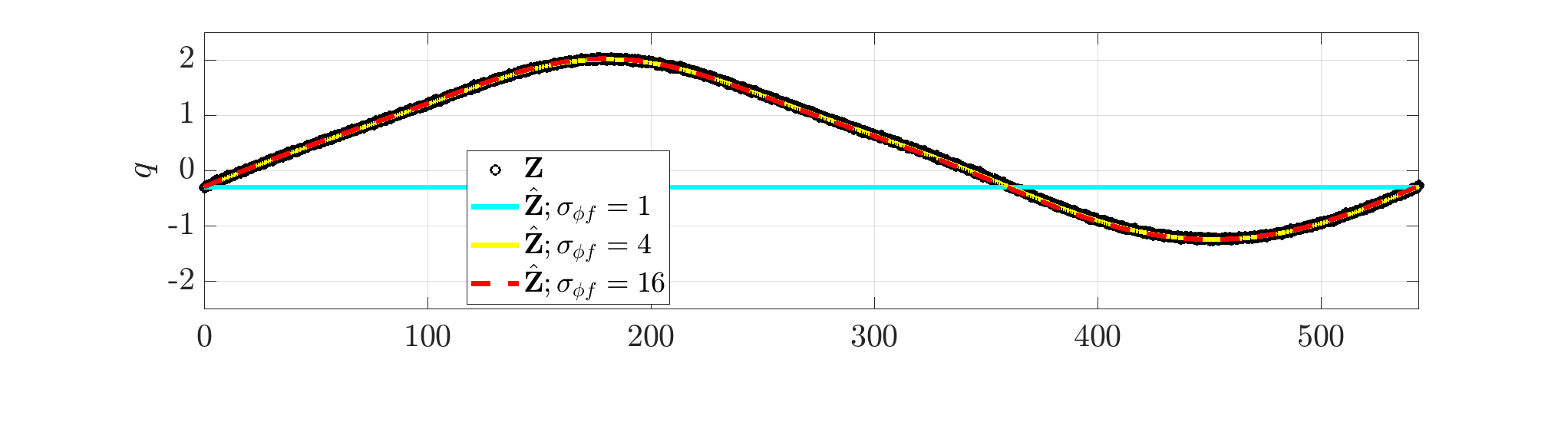}&
     \includegraphics[trim={75 25 80 10},clip,width=0.5\textwidth]{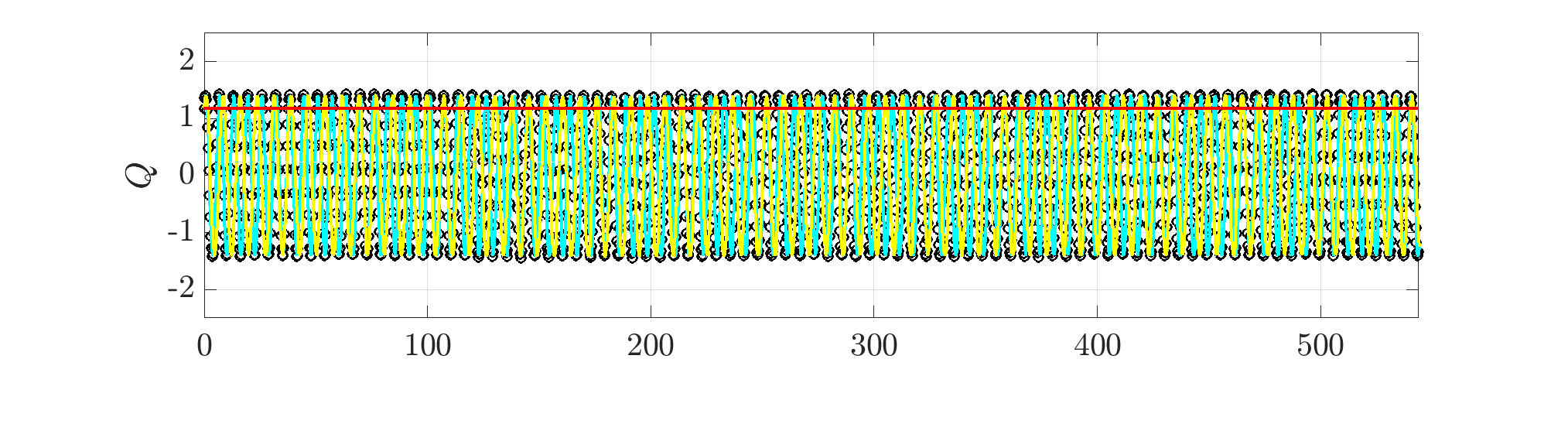}\\
     \includegraphics[trim={75 25 80 10},clip,width=0.5\textwidth]{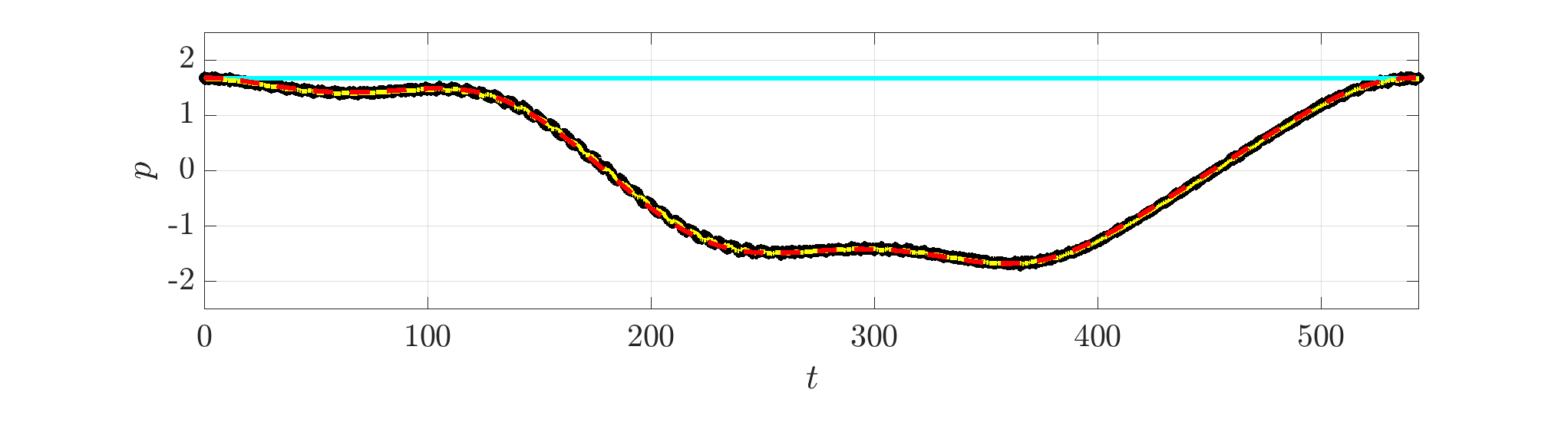}& 
     \includegraphics[trim={75 25 80 10},clip,width=0.5\textwidth]{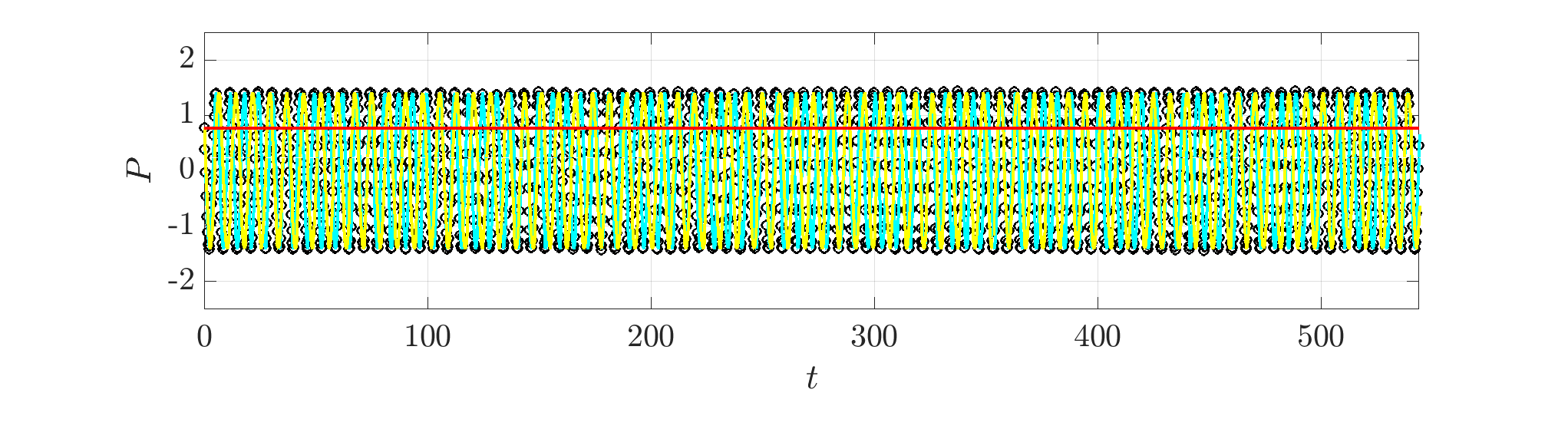} \\ \hline
 \end{tabular}
\end{center}
\caption{Time series plots of the data and learned systems from Figure \ref{fig:full_sims}. The coarse-grained model (red) accurately captures the slow dynamics (left, yellow).}
\label{fig:full_sims_time}
\end{figure}

An intriguing aspect of Hamiltonian systems that exhibit approximate symmetries, from a computational standpoint, is that they typically admit a hierarchy of models which can be used to efficiently express the dynamics in different regimes. These consist of the dynamics of the limiting roto-rate $R_0$, the full Hamiltonian vector field $X_\vep$, and reduced dynamics for each order in $\vep$ obtained from averaging and restricting $X_\vep$ according to the methodology in Section \ref{sec:1stOrderHtheory}. Theorem  \ref{thm_1stOrder} indicates that the leading-order reduced vector-field $\Xred$ is Hamiltonian, and we conjecture that similar conditions exist for higher-order-averaged systems to remain Hamiltonian, justifying a search for Hamiltonian structure in coarse-grained models of $X_\vep$. While the focus of this article is to demonstrate that weak-form methods can be used to identify a coarse-grained Hamiltonian system that agrees with averaging theory, in this Subsection we demonstrate a more general property, that when given access to all state variables, WSINDy allows access to multiple models from the hierarchy associated with nearly-periodic Hamiltonian systems, simply by varying the test function radius $T_\phi$. That is, this Subsection provides justification for a complete analysis of this multi-model inference capability, which will be provided in a future work. 


\begin{figure}
\begin{tabular}{@{}cc@{}}
\includegraphics[trim={0 0 0 0},clip,width=0.5\textwidth]{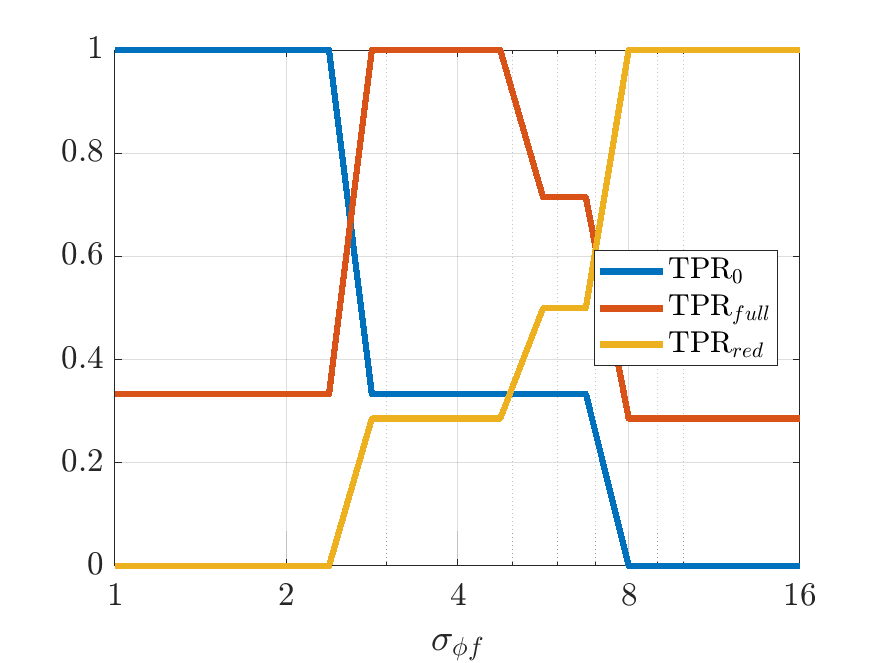} &
\includegraphics[trim={0 0 0 0},clip,width=0.485\textwidth]{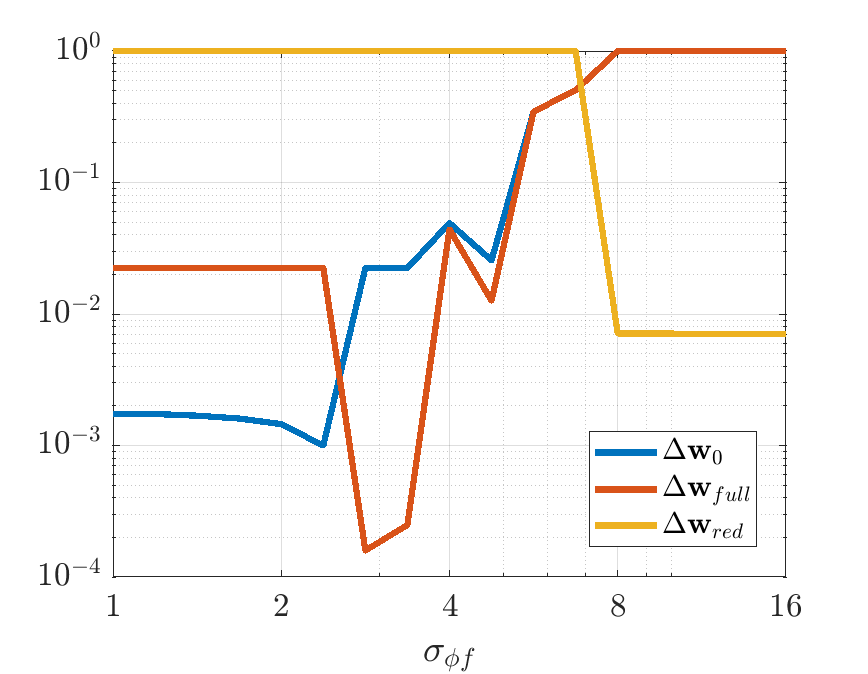}\\ \hline
\end{tabular}
\caption{Transitions in the WSINDy output from $H_0$ (blue), to $H_\vep$ (red), to $\Hred$ (yellow) as $\sigma_{\phi f}$ (eq. \eqref{sigmas2}) is increased using data as in Figure \ref{fig:full_sims}. Results are averaged over 100 instantiations of noise. The left and right plots show the TPR values (eq.\ \eqref{tpr}, a TPR of 1 indicates identification of the true model) and coefficient errors $\Delta \wbf_{(*)}:=\nrm{\what-\wbf_{(*)}}_2/\nrm{\wbf_{(*)}}_2$ for coefficients $\wbf_{(*)}$ defining the Hamiltonians $H_0$, $H_\vep$, and $\Hred$, associated with subscripts ``0'', ``full'', and ``red'' respectively.}
\label{fig:full_results}
\end{figure}

The advantages of weak-form coarse-graining can be clearly observed by examining the performance of WSINDy as a function of the {\it weak-form scale selector} 
\begin{equation}\label{sigmas2}
\sigma_{\phi f} = \frac{T_\phi}{T_f}
\end{equation}
where $T_\phi$ is the support length (in time) of the test function $\phi$ and $T_f$ is the dominant period of the fast scale. In words, $\sigma_{\phi f}$ is the number of fast periods that occur in one integration of the dynamics against $\phi$. Figures \ref{fig:full_sims} and \ref{fig:full_results} depict the performance of WSINDy as $\sigma_{\phi f}$ is varied from 1 to 16 (i.e.\ $T_\phi$ is varied from $T_f$ to $16T_f$) when applied to the full 4-dimensional nearly-periodic Hamiltonian system \eqref{coupledosc_dyn} in Example 1. Here the library $\Hbb$ contains only the 7 terms necessary to represent all three relevant models, that is, the leading order dynamics given by $H_0 = \frac{1}{2}(Q^2+P^2)$, the full dynamics given by $H_\vep$ (eq. \eqref{eq:Hvepexp1}), and the reduced Hamiltonian $\Hred$ (eq. \eqref{eq:Hcalexp1}). Figure \ref{fig:full_results} demonstrates that simply by varying the support of the test function, we gain access to all three models. The added 1\% noise affects the transition between models (Figure \ref{fig:full_results}, left; with lower noise requiring larger $\sigma_{\phi f}$ to obtain the reduced model) as well as the recovered coefficient accuracy (Figure \ref{fig:full_results}, right), which is less that $1\%$ for each model in its range of validity (i.e.\ when $\text{TPR}=1$, see eq. \eqref{tpr}). Figures \ref{fig:full_sims} and \ref{fig:full_sims_time} and provide visualizations of the learned dynamics $\widehat{\Zbf}$ in phase space and as time series. It is clear from the plots over $(P,q,p)$ and $(Q,P,q)$ space (Figure \ref{fig:full_sims}) that the dynamics can be approximately described as two commutative flows given by the coarse-grained (red) and roto-rate (cyan) dynamics. Meanwhile, the time series plots of the slow variables (\ref{fig:full_sims_time}, left) indicate that the coarse-grained dynamics (red) very accurately match the full model (yellow). Hence, from a single dataset, the weak form allows access to three models that are fundamental to understanding a nearly-periodic Hamiltonian system.

For the remainder of the article, we focus solely on the ability of WSINDy to identify a suitable coarse-grained Hamiltonian model in the more realistic setting of having access only to the slow degrees of freedom, as the perturbation level $\vep$, the noise level $\sigma_{NR}$, and the region of phase space are varied.

\section{Numerical Experiments}\label{sec:numerics}



To exhibit the efficacy of weak-form coarse-graining in a variety of contexts, we examine the following nearly periodic Hamiltonian systems. In each case we apply the geometric reduction procedure outlined in Section \ref{sec:1stOrderHtheory} to derive the reduced Hamiltonian system to be discovered. Note that in Examples 1 and 2 the coefficients of the reduced system are given in terms of special functions and integrals which must be numerically evaluated.

\subsection{Example 1: Nonlinearly Coupled Oscillators}

\begin{figure}
\begin{tabular}{@{}c@{}c@{}}
		\includegraphics[trim={90 0 90 10},clip,width=0.5\textwidth]{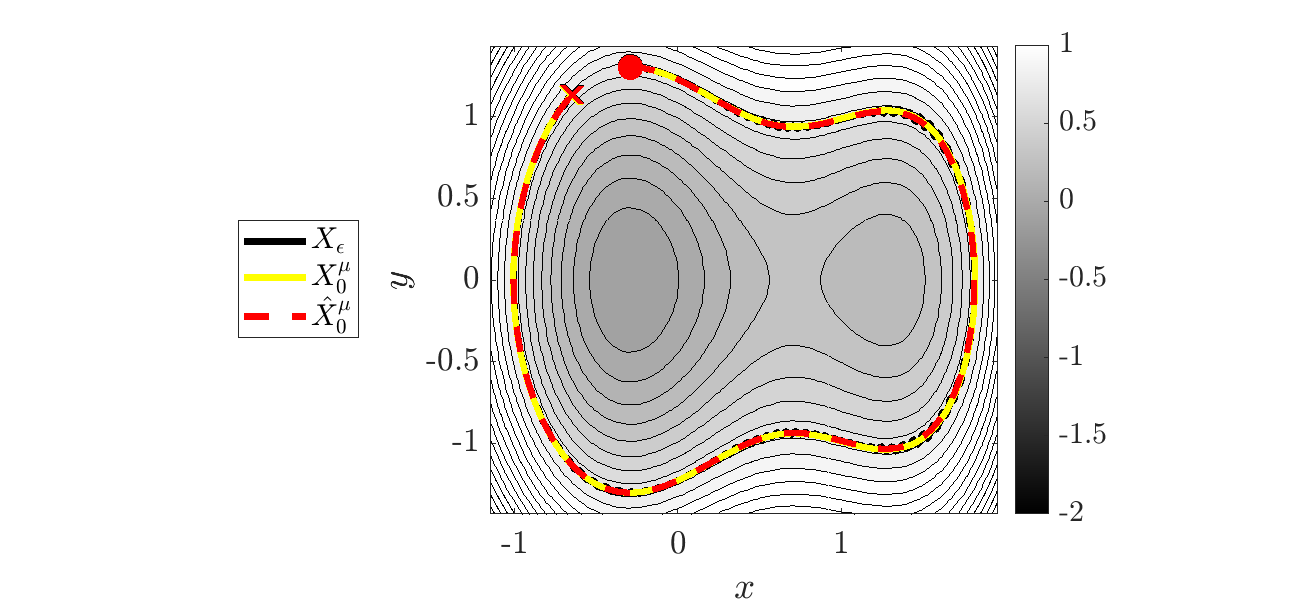} &
	    \includegraphics[trim={90 0 90 10},clip,width=0.5\textwidth]{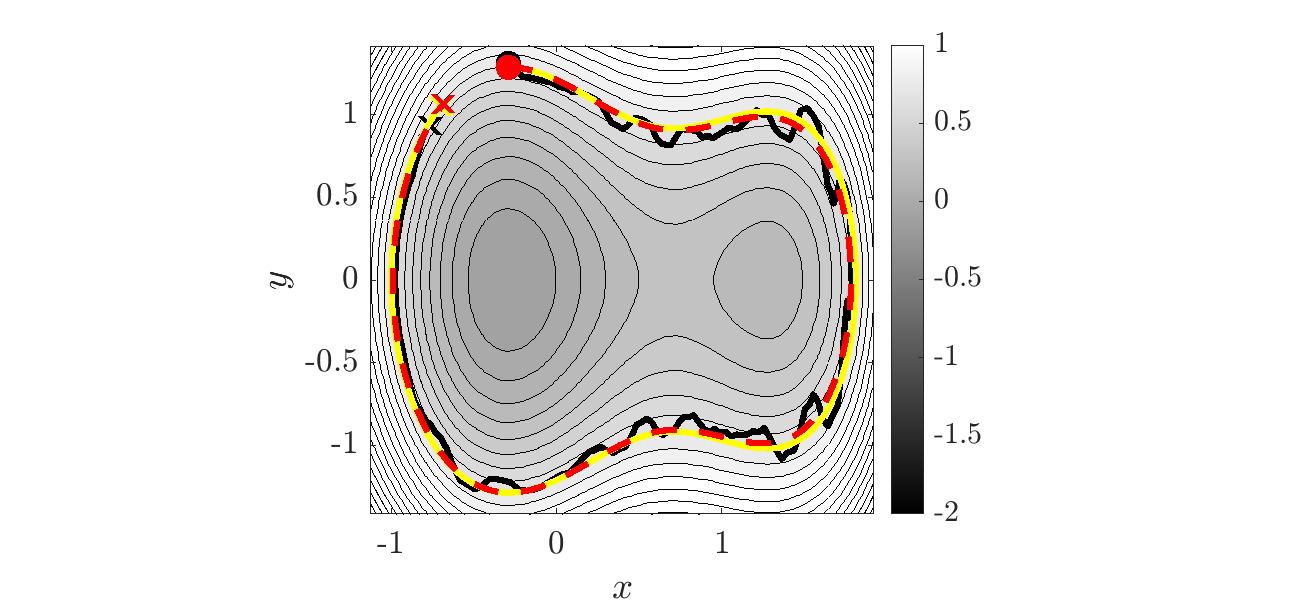} \\ \hline
\end{tabular}
\caption{Dynamics of Example 1 in reduced phase space, with true dynamics $X_\vep$ in black, leading-order reduced dynamics $\Xred$ in yellow, and learned reduced dynamics $\Xredhat$ in red (initial and final conditions are marked with dots and x's), overlaying contours of the learned reduced Hamiltonian $\Hredhat$. Left: $\vep= 0.01$ dynamics, right: $\vep=0.05$. At $\vep=0.05$ the nonlinear coupling between the fast and slow oscillators results in irregular perturbations to the slow variables.}
\label{fig:exp1_sim}
\end{figure}

For our first system we examine two coupled oscillators with different timescales and a nonlinear coupling potential. The $\vep$-dependent Hamiltonian and sympletic form on $\Rbb^4$ are given by
\begin{equation}\label{eq:Hvepexp1}
    H_\vep (Q,P,q,p) = \frac{1}{2} (Q^2 +P^2 ) + \frac{1}{2} \vep (q^2 + p^2 )  + \vep V(Q,q), \quad \Omega = dQ\wedge dP+dq\wedge dp 
\end{equation}
with coupling potential given by
\begin{equation}
	V(Q,q) = Qq \sin{(2Q + 2q)}.
\end{equation}
Note that $\Omega_\vep = \Omega_0=\Omega$. The variables $(Q,P)$ denote the position and momentum of a fast oscillator with dynamics on $\CalO(1)$ timescales, while $(q,p)$ denotes the slow variables evolving on a timescale of $\CalO(\vep)$. The equations of motion for this two-oscillator system are
\begin{equation} \label{coupledosc_dyn}	\begin{cases}
		\dt{Q} = P, \qquad & \dt{P} = -Q - \vep \partial_{Q} V(Q,q) = -Q - \vep\,\bigg(q\,\sin{(2Q + 2q)} + 2\,Q\,q\,\cos(2Q + 2q)\bigg) \\ \dt{q} = \vep p, \qquad & \dt{p} = -\vep q - \vep \partial_{q} V(Q,q) = -\vep\,q - \vep\bigg(Q\, \sin{(2Q + 2q)} + 2\,Q\,q\,\cos(2Q+2q)\bigg).
\end{cases} \end{equation}
As $\vep$ increases, the fast variables $(Q,P)$ impart large perturbations to the oscillating motion of the slow variables $(q,p)$ by means of  $V$ (see Figure \ref{fig:exp1_sim}, right). At $\vep = 0$ the two oscillators decouple and the dynamics are governed by the limiting roto-rate
\[R_0 = P\partial_Q-Q\partial_P,\]
which we identify as the Hamiltonian flow on $(\Rbb^4,\Omega)$ of the leading-order adiabatic invariant
\begin{align*}
\mu_0(Q,P,q,p) = \frac{1}{2}(Q^2 + P^2).
\end{align*}
In this case the limiting frequency function $\omega_0$ such that $X_0=\omega_0R_0$ is unity. 
Hence, to leading order $(Q,P)$ exhibits periodic clockwise circular rotation while $(q,p)$ remains fixed. The flow-map $\Phi_t^0$ for $R_0$ is thus given by 
\begin{equation}
\begin{pmatrix} Q\\ P\end{pmatrix} \mapsto \CalR(t)\begin{pmatrix} Q\\ P\end{pmatrix}, \quad \begin{pmatrix} q\\ p\end{pmatrix} \mapsto \begin{pmatrix} q\\ p\end{pmatrix};\qquad 
    \qquad \CalR(t) = \begin{pmatrix}
		\cos(t) & \sin(t) \\ -\sin(t) & \cos(t)\end{pmatrix}.
\end{equation}
Each level set of $\mu_0$ is a single orbit of $R_0$ crossed with the slow variables $(q,p)\in \Rbb^2$, hence the quotient procedure $M\to \Lambda_\mu\to\Mred$ produces $\Mred =(q,p)\in \Rbb^2$. Thus, the leading-order reduced Hamiltonian $\Hred(q,p)$, is obtained from averaging $H_\vep$ around $\Phi_t^0$ starting from an arbitrary point $(Q,P,q,p)$ satisfying $\mu = \mu_0(Q,P,q,p)$:
\begin{align}
\Hred(q,p) &= \frac{1}{2\pi}\int_0^{2\pi} H_\vep(\Phi_t^0(Q,P,q,p))\,dt,\\
& =\mu + \vep\, \frac{1}{2}(q^2 + p^2) + \vep\,\frac{1}{2\pi} \int_0^{2\pi}{V(Q \cos{t} + P \sin{t} , q) dt}   \\ 
&= \mu + \vep\, \frac{1}{2}(q^2 + p^2)  +  \vep\,\frac{q}{2\pi}\int_0^{2\pi}{ \left( Q \cos{t} + P \sin{t} \right)  \sin{\left(2\left( Q \cos{t} + P \sin{t} \right)+ 2q \right)} dt}  \\
\label{eq:Hcalexp1}&= \mu + \vep\, \frac{1}{2}(q^2 + p^2)  +  \vep\sqrt{ 2\mu}   \ \mathcal{J}_1\left(2 \sqrt{ 2\mu} \right)\,Q \cos(2 Q)   ,
\end{align}
where $\mathcal{J}_1(z)$ denotes the Bessel function of the first kind.
The leading-order reduced equations of motion for the slow variables $(q,p)$ are therefore given by
\begin{align}
\dt{q} & = \partial_{p}\Hred =\vep\,p \label{reduced_q2}\\
\dt{p} & = -\partial_{q}\Hred = -\vep\,q - \vep\,\sqrt{ 2\mu}   \ \mathcal{J}_1\left(2 \sqrt{ 2\mu} \right)\,\cos(2q) + 2\,\vep\,\sqrt{ 2\mu}   \ \mathcal{J}_1\left(2 \sqrt{ 2\mu} \right)\,q\,\sin(2q).\label{reduced_p2}
\end{align}

\subsection{Example 2: H\'enon-Heiles Embedded Pendulum}

\begin{figure}
\begin{tabular}{@{}c@{}c@{}}
		\includegraphics[trim={90 0 60 10},clip,width=0.5\textwidth]{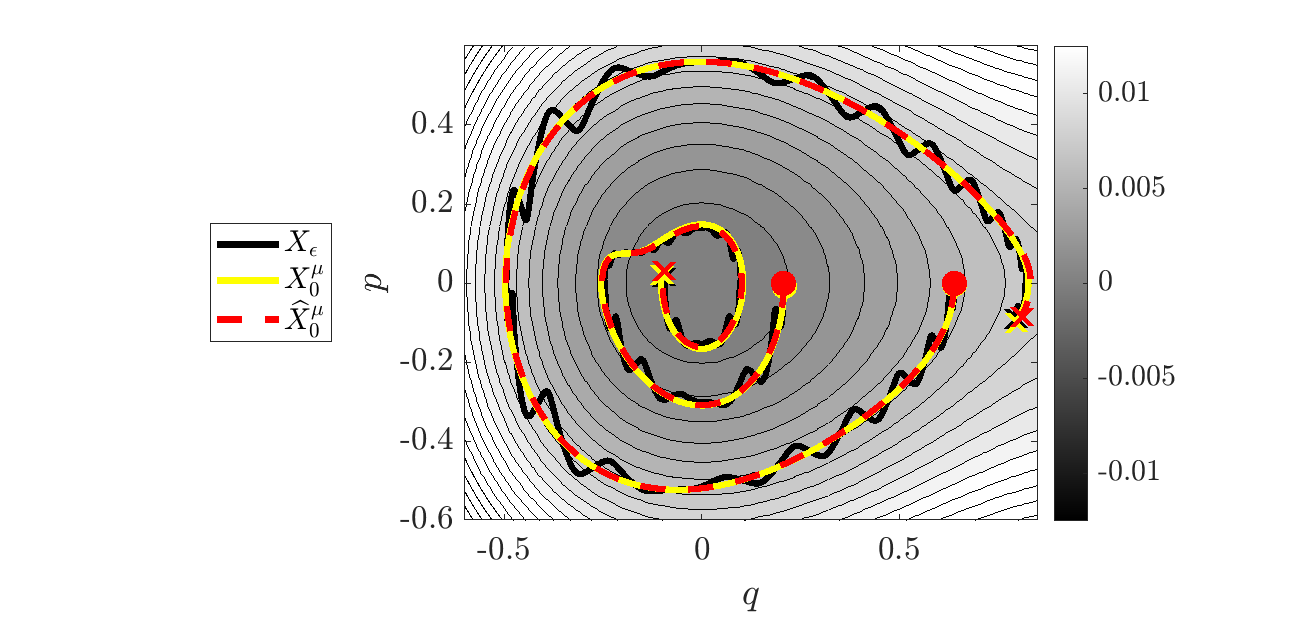} &
	    \includegraphics[trim={90 0 60 10},clip,width=0.5\textwidth]{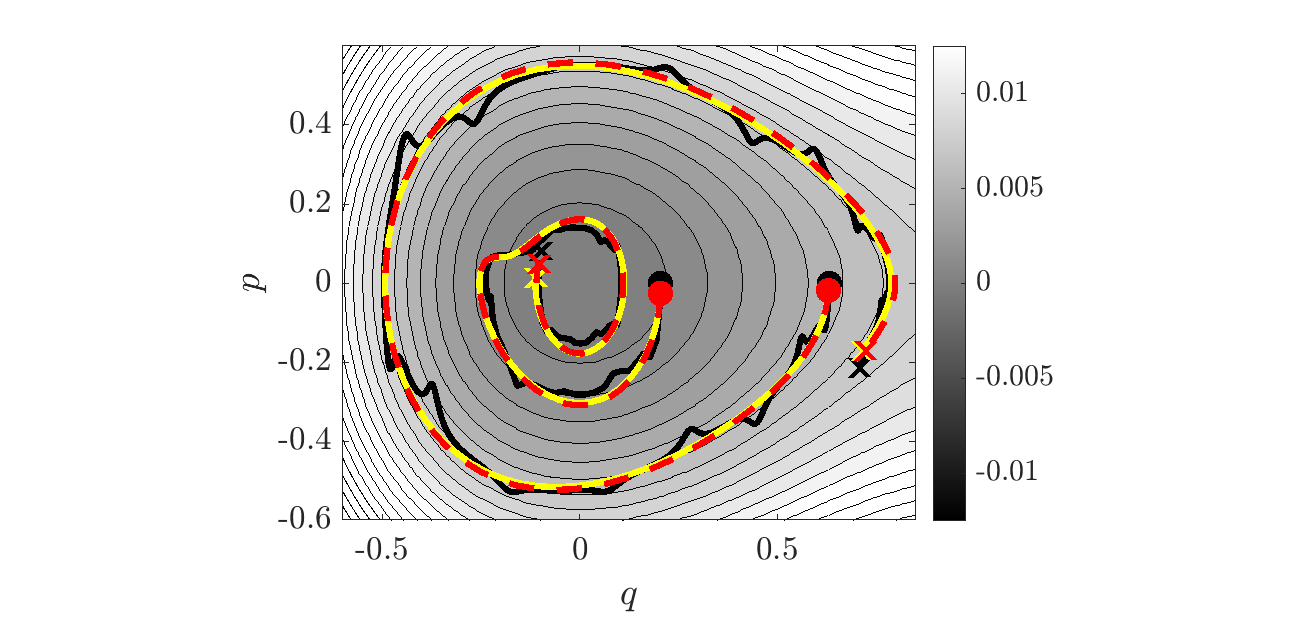} \\ \hline
\end{tabular}
\caption{Reduced dynamics of Example 2 with $\vep= 0.05$ and different initializations of the embedded pendulum. Left: $(Q(0),P(0)) = (\frac{1}{2}\pi,0)$, with pendulum dynamics close to linear. Right: $(Q(0),P(0)) = (\frac{31}{32}\pi,0)$, with the pendulum initialized close to the unstable fixed, removing the strict scale separation (see Figure \ref{fig:exp2_fft}). For each, $(q_1,p_1)$ and $(q_2,p_2)$ are plotted as separate trajectories overlaying $\Hredhat(0,0,q_2,p_2)$ (coloring conventions are the same as in Figure \ref{fig:exp1_sim}). }
\label{fig:exp2_sim}
\end{figure}

Let $z = (Q,P,q_1,p_1,q_2,p_2)$ denote the variables of a 3-degree of freedom system with Hamiltonian
\begin{align*}
H_\vep(z) & = H_0(z) + \vep\,H_1(z)\\
& = \frac{1}{2}P^2 + \alpha^2\,(1-\cos\,Q) + \vep\left(\frac{1}{2}p_1^2 + \frac{1}{2}\nu_1(Q)q_1^2 + \frac{1}{2}p_2^2 + \frac{1}{2}\nu_2(Q)\,q_2^2 + V(q_1,q_2,Q)\right),
\end{align*}
and canonical symplectic form $\Omega = dQ\wedge dP + dq_1\wedge dp_1 + dq_2 \wedge dp_2$, again independent of $\vep$. Here $\alpha^2$ is a positive constant and $\nu_1:\mathbb{R}\rightarrow\mathbb{R}$, $\nu_2:\mathbb{R}\rightarrow\mathbb{R}$, $V:\mathbb{R}^3\rightarrow\mathbb{R}$ are each fixed smooth functions. The corresponding equations of motion are
\begin{align}
\dt{P} &= -\alpha^2\,\sin\,Q - \vep\,\bigg(\frac{1}{2}\frac{d\nu_1}{dQ}\,q_1^2 + \frac{1}{2}\frac{d\nu_2}{dQ}\,q_2^2 + \partial_QV\bigg)\label{ae_micro_1}\\
\dt{Q} & = P\\
\dt{p}_1 & = - \vep\bigg(\nu_1(Q)\,q_1  +\partial_{q_1}V\bigg)\\
\dt{q}_1 & = \vep\,p_1\\
\dt{p}_2 & = -\vep\bigg(\nu_2(Q)\,q_2 + \partial_{q_2}V\bigg)\\
\dt{q}_2 & = \vep\,p_2.\label{ae_micro_6}
\end{align}
On the subset of phase space $
S_0 = \left\{z\ :\ \frac{1}{2}P^2 < \alpha^2(1+\cos\,Q)\right\}$, equations  \eqref{ae_micro_1}-\eqref{ae_micro_6} define a fast-slow system. The limiting angular frequency function (see Definition \ref{def:nphs}) is
\begin{align}
\omega_0(z) =\frac{4}{\alpha}\int_0^1\frac{ds}{\sqrt{1 - \frac{H_0(z)}{2\alpha^2}s^2}\sqrt{1-s^2}}= \frac{4}{\alpha}K\left({\frac{H_0(z)}{2\alpha^2}}\right)
\end{align}
where $K$ denotes the complete elliptic integral of the first kind (defined in e.g.\ \cite{AbramowitzStegun1972}; note that within $S_0$ the argument for $K$ is in the interval $[0,1]$). The limiting roto-rate is
\begin{align*}
R_0  = -\frac{\alpha^2}{\omega_0}\,\sin Q\,\partial_{P} + \frac{P}{\omega_0}\,\partial_Q.
\end{align*}
Let $\Phi_\theta$ denote the time-$\theta$ flow map for $R_0$. The leading-order adiabatic invariant is given by
\begin{align*}
\mu_0(z) &= \bigg(\frac{1}{2\pi}\int_0^{2\pi}\Phi_\theta^*(P\,dQ + p_1\,dq_2 + p_2\,dq_2)\,d\theta\bigg)_z(R_0(z))\\
& = \frac{1}{2\pi}\int_0^{2\pi}(P\,dQ + p_1\,dq_2 + p_2\,dq_2)_{\Phi_\theta(z)}(T_z\Phi_\theta R_0(z))\,d\theta\\
& = \frac{1}{2\pi}\int_0^{2\pi} \bigg(\iota_{R_0}(P\,dQ + p_1\,dq_2 + p_2\,dq_2)\bigg)(\Phi_\theta(z))\,d\theta\\
& = \frac{1}{2\pi}\int_{\gamma(z)}P\,dQ.
\end{align*}
where $\gamma(z)$ is the parameterized curve $\gamma(z)(\theta) = \Phi_\theta(z)$. This integral can also be expressed in terms of elliptic integrals as
\begin{align*}
\mu_0(z) = 16\,\alpha\,E\left({\frac{H_0(z)}{2\alpha^2}}\right) + \bigg(8\frac{H_0(z)}{\alpha} - 16\,\alpha\bigg)K\left({\frac{H_0(z)}{2\alpha^2}}\right),
\end{align*}
where $E$ denotes the complete elliptic integral of the second kind. 

Like Example 1, $\mu_0$ is only a function of the fast variables $(Q,P)$. For $Q$ mod 2$\pi$, level sets of $\mu_0$ contain a single $R_0$ orbit, so again the quotient procedure $M\to \Lambda_\mu\to\Mred$ simply eliminates the fast variables $(Q,P)$. Let $(Q_\theta,P_\theta)$ be defined by $\Phi_\theta(z) = (Q_\theta,q_1,q_2,P_\theta,p_1,p_2)$. The first-order averaged Hamiltonian $\Hav$ and reduced Hamiltonian $\Hred$ are then given by
\begin{align*}
\Hav(z) = H_0(z)+\vep\,\bigg(\frac{1}{2}p_1^2 + \frac{1}{2}\overline{\nu}_1(\mu)\,q_1^2 + \frac{1}{2}p_2^2 + \frac{1}{2}\overline{\nu}_2(\mu)\,q_2^2 +\overline{V}(q_1,q_2,\mu) \bigg) = H_0(z) +\vep \Hred(q_1,p_1,q_2,p_2) 
\end{align*}
where
\begin{equation}
\overline{\nu}_1(\mu) = \frac{1}{2\pi}\int_0^{2\pi} \nu_1(Q_\theta)\,d\theta,\quad 
\overline{\nu}_2(\mu) = \frac{1}{2\pi}\int_0^{2\pi} \nu_2(Q_\theta)\,d\theta,\quad 
\overline{V}(q_1,q_2,\mu) = \frac{1}{2\pi}\int_0^{2\pi}V(q_1,q_2,Q_\theta)\,d\theta,
\end{equation}
with the integrals taken along any trajectory $(Q_\theta,P_\theta)$ of $R_0$ with $\mu_0(Q_0,P_0) = \mu$. 

We choose the following parametrizations, which lead to the reduced dynamics reproducing the famous H\'enon-Heiles system\cite{HenonHeiles1964AstronJ} proposed as a simple model for the motion of a star within a galaxy: 
\[\nu_1(Q) = \nu_2(Q)=\nu_3(Q)=1 + 2\sin(Q),\quad V(Q,q_1,q_2) = \nu_3(Q)\left(q_1^2\,q_2 -\frac{1}{3} q_2^3\right).\]
Unlike the previous example, here the nonlinear pendulum dynamics lead to anisotropic oscillation frequencies depending on the initial conditions. Correspondingly, the identified reduced system is highly-dependent on which integral curve of the roto-rate the initial conditions lie. For $(Q(0),P(0))$ close to the unstable fixed-point $Q(0) = \pi$, the fast dynamics become increasingly nonlinear and exhibit a longer period. This leads to a blending of the fast and slow scales, as visualized in Figure \ref{fig:exp2_fft} in Fourier space. The right plots show the power spectrum of the trajectories at $Q(0) = \pi/2$ (top) and $Q(0)=31\pi/32$ (bottom), where the latter clearly has much less of a separation of scales. This example demonstrates that WSINDy is still able to identify the leading-order reduced dynamics from systems with a complex microstructure.

\begin{figure}
\begin{center}
\begin{tabular}{@{}c@{}c@{}}
    \includegraphics[trim={0 0 0 0},clip,width=0.4\textwidth]{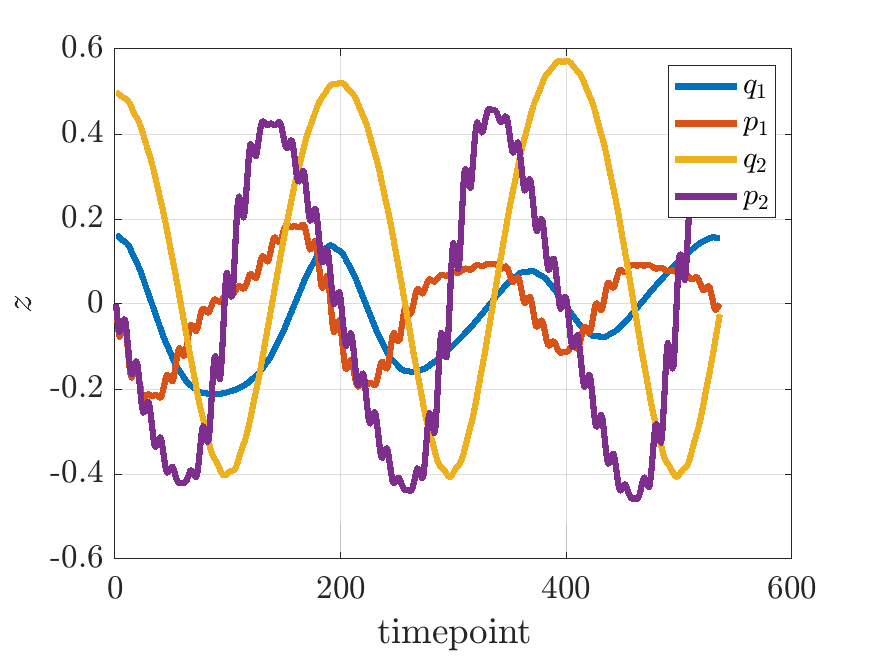} &
		\includegraphics[trim={0 0 0 0},clip,width=0.4\textwidth]{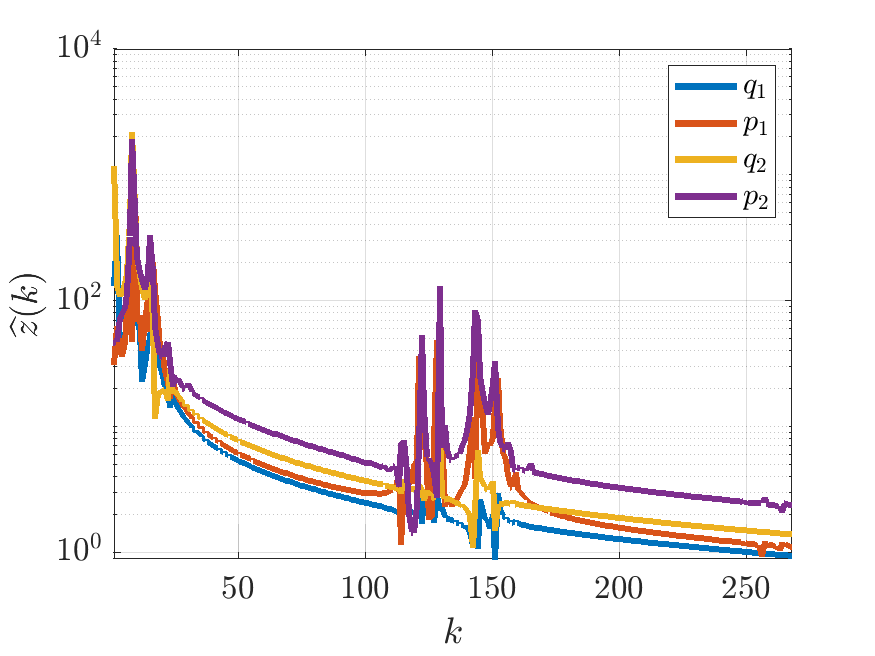} \\
		\includegraphics[trim={0 0 0 0},clip,width=0.4\textwidth]{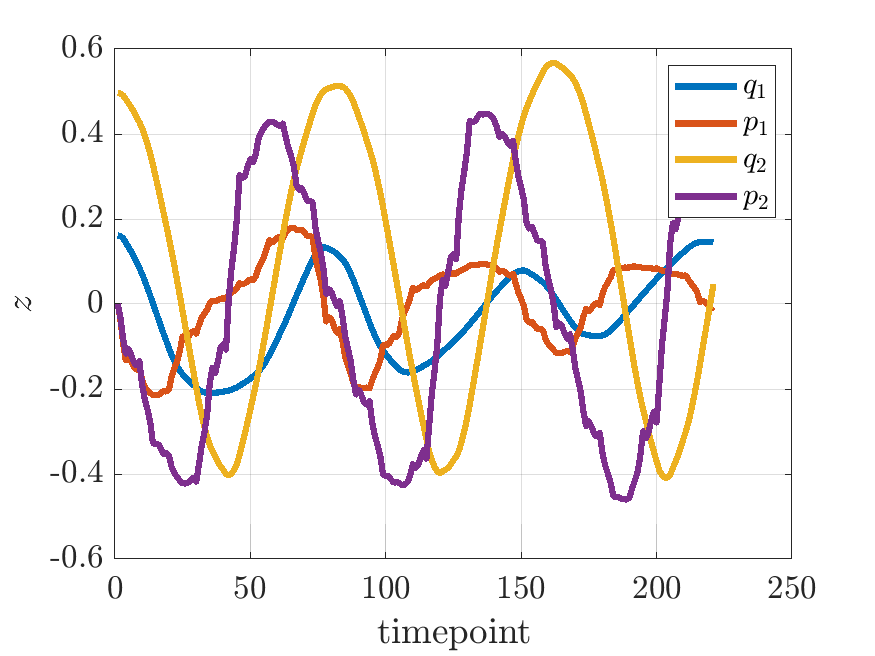} &
		\includegraphics[trim={0 0 0 0},clip,width=0.4\textwidth]{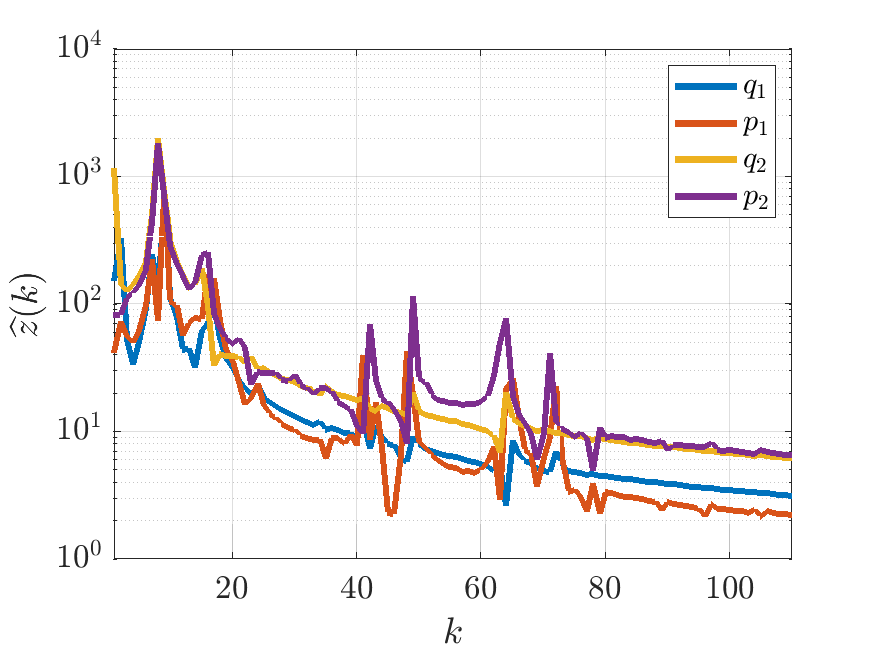} \\ \hline
\end{tabular}
\end{center}
\caption{Time series of the slow variables for Example 2 with $\vep=0.05$ (left) along with power spectra (right). The top plots show $Q(0) = \pi/2$, and an observable separation of scales (top right). The separation of scales breaks down at $Q(0) = 31\pi/32$ (bottom row) as the embedded pendulum dynamics approach the separatrix, as observed in the power spectrum.}
\label{fig:exp2_fft}
\end{figure}

\subsection{Example 3: Charged particle motion}

\begin{figure}
\begin{tabular}{@{}c@{}c@{}}
		\includegraphics[trim={25 0 35 10},clip,width=0.485\textwidth]{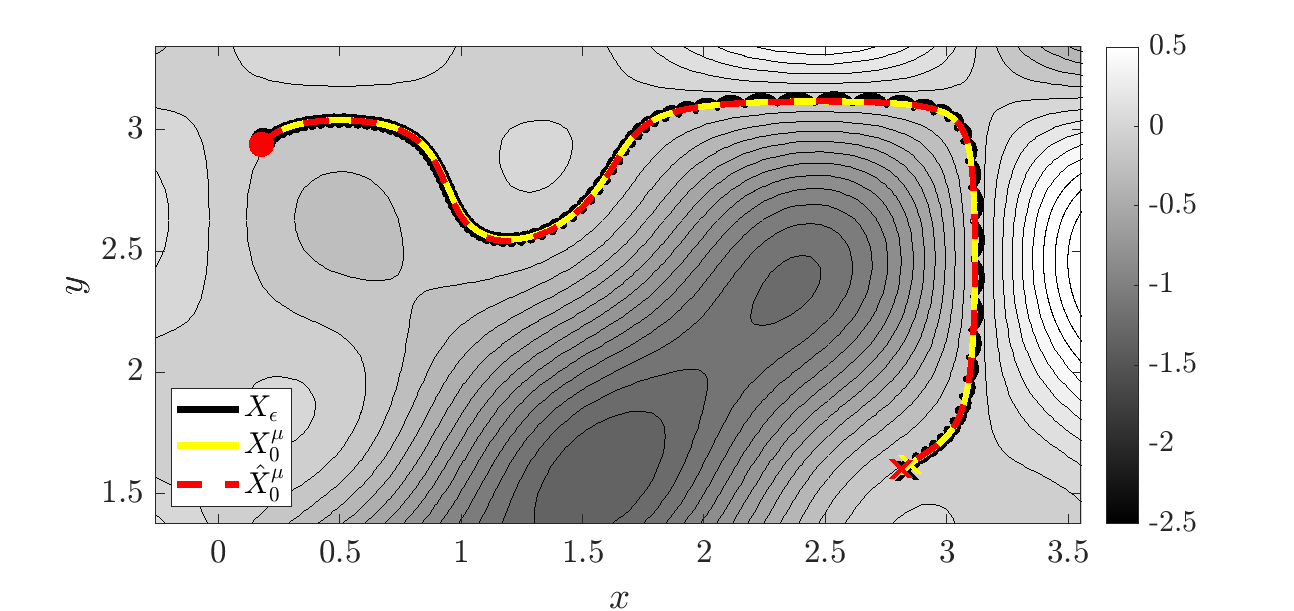} &
	    \includegraphics[trim={25 0 35 5},clip,width=0.5\textwidth]{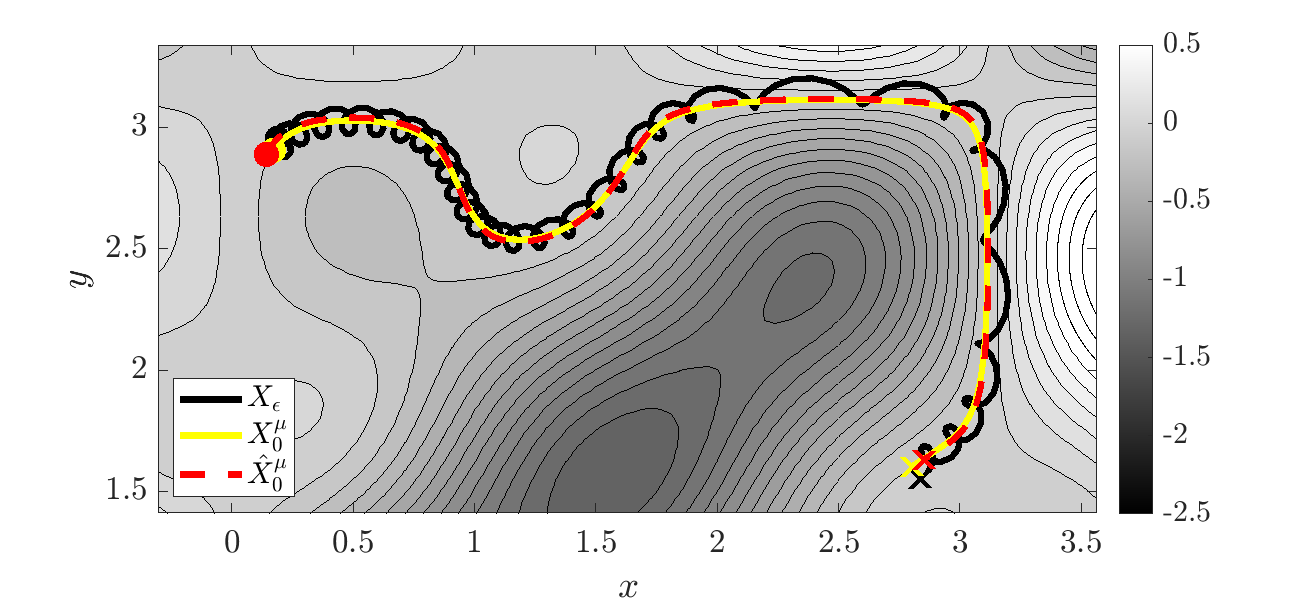} \\ \hline
\end{tabular}
\caption{Dynamics of Example 3 in reduced phase space $(x,y)\in \Rbb^2$. Left: $\vep= 0.01$ dynamics, right: $\vep=0.03$. On the right one can observe irregular perturbations as the unreduced dynamics (black curve) jump between contours of the electric potential energy $\varphi$ shown in gray. (Coloring conventions are the same as in Figure \ref{fig:exp1_sim}). }
\label{fig:exp3_sim}
\end{figure}

Consider the motion of a charged particle with position $(x,y,z)$ and velocity $(v_x,v_y,v_z)$ in a time-independent electromagnetic field of the form $\bm{B}(x,y,z) = \nabla \times \Abf(x,y,z)$, $\bm{E}(x,y,z) = -\nabla\varphi(x,y,z)$. If the field potentials are independent of $z$, that is, $\Abf = (A_x(x,y),A_y(x,y),0)$ and $\varphi = \varphi(x,y)$, then the $z$-component of the particle momentum is a constant of motion. Therefore the Lorentz force equation describing the particle's motion may be written, for $B(x,y) = \partial_xA_y-\partial_yA_x$, as
\begin{align*}
\dt{v}_x &= -\partial_x\varphi(x,y) +v_y\,B(x,y)\\
\dt{v}_y & = -\partial_y\varphi(x,y) - v_x\,B(x,y)\\
\dt{x} & = \vep\,v_x\\
\dt{y} & = \vep\,v_y.
\end{align*}
This is an $\vep$-dependent Hamiltonian system on $\mathbb{R}^4\ni (x,y,v_x,v_y)$ if the Hamiltonian and symplectic form are given by
\begin{align*}
&H_\vep(x,y,v_x,v_y)     = \vep^2\,\frac{1}{2}(v_x^2 + v_y^2  ) + \vep\,\varphi(x,y) \\
&\Omega_\vep = -B(x,y)\,dx\wedge dy + \vep\,(dx\wedge dv_x + dy\wedge dv_y).
\end{align*}
Letting $q = (x,y)$, $v = (v_x,v_y)$, the flow map of the roto-rate at $\vep=0$ is given by
\begin{equation}
\Phi_t^0(q,v) =\Big( q,\ D(q) + \CalR(t)(v-D(q))\Big), \qquad D(q) := (-\partial_y\varphi/ B ,\ \partial_x\varphi / B)
\end{equation}
with $\CalR$ the rotation matrix from Example 1. Using the formulas in \cite{BurbySquire2020JPlasmaPhys}, the adiabatic invariant is degenerate to first order, with $\mu_* = \mu_2$ (see \eqref{av_ad_inv}) given by 
\begin{align*}
\mu_2(q,v) &= \frac{1}{2}\frac{|v - D(q)|^2}{B(q)}.
\end{align*}
The first-order averaged system $\Xav$ is automatically presymplectic Hamiltonian according to Lemma \ref{level_set_ham}, with Hamilton's equations \eqref{av_ham_eqn_unreduced} for $\Hav=\vep\,\varphi(x,y)$ and $\Omega_0=- B(x,y)\,dx\wedge dy$ producing the first-order averaged dynamics 
\begin{equation}
\dt{x} = -\vep\,\frac{\partial_y \varphi(x,y)}{B(x,y)}, \qquad
\dt{y}  = \vep\,\frac{\partial_x \varphi(x,y)}{B(x,y)}.
\end{equation}
Restriction to a level set of $\mu_2$ merely restricts $v$-space, which is already absent from dynamics, and the flow of the roto-rate fixes $q$-variables, so we have $\Hred = \Hav$, $\Ored = \Oav$. 


A simple test case for this problem is a constant magnetic field. This corresponds physically to zero current density, as this reduces Amp\'ere's law to $\nabla \times \pmb{B} = 0$, under which the assumption $\pmb{B} = B(x,y)\hat{z}$ leads to $B = const$. We set $B=1$ throughout, which for example corresponds to the vector potential $\Abf=(\frac{1}{2}x,-\frac{1}{2}y,0)$. We prescribe a sinusoidal ambient electric field potential 
\begin{equation}\label{trigE}
\varphi(x,y) = \sum_{j=1}^3 \frac{1}{j} \sin(jx)\sin(jy)
\end{equation}
which, given Gauss's law $\Delta \varphi = \vep_0^{-1}\rho$ relating the electric potential $\varphi$ to the charge density $\rho$, loosely corresponds to a lattice of bound charges.

As the charged particle approximately travels along contours of $\phi$, the velocity variables $(v_x,v_y)$ exhibit near-periodic fast-scale oscillations. For larger $\vep$, the motion of $(v_x,v_y)$ becomes modulated by the steepness of $\nabla\phi$, which leads to irregular perturbations to the position variables $(x,y)$. This can be seen in Figure \ref{fig:exp3_sim} (right) as the oscillations in the black curve become elongated.

\subsection{Example 4: Coupled charged particle motion}

\begin{figure}
\begin{tabular}{@{}c@{}c@{}}
		\includegraphics[trim={30 0 30 10},clip,width=0.5\textwidth]{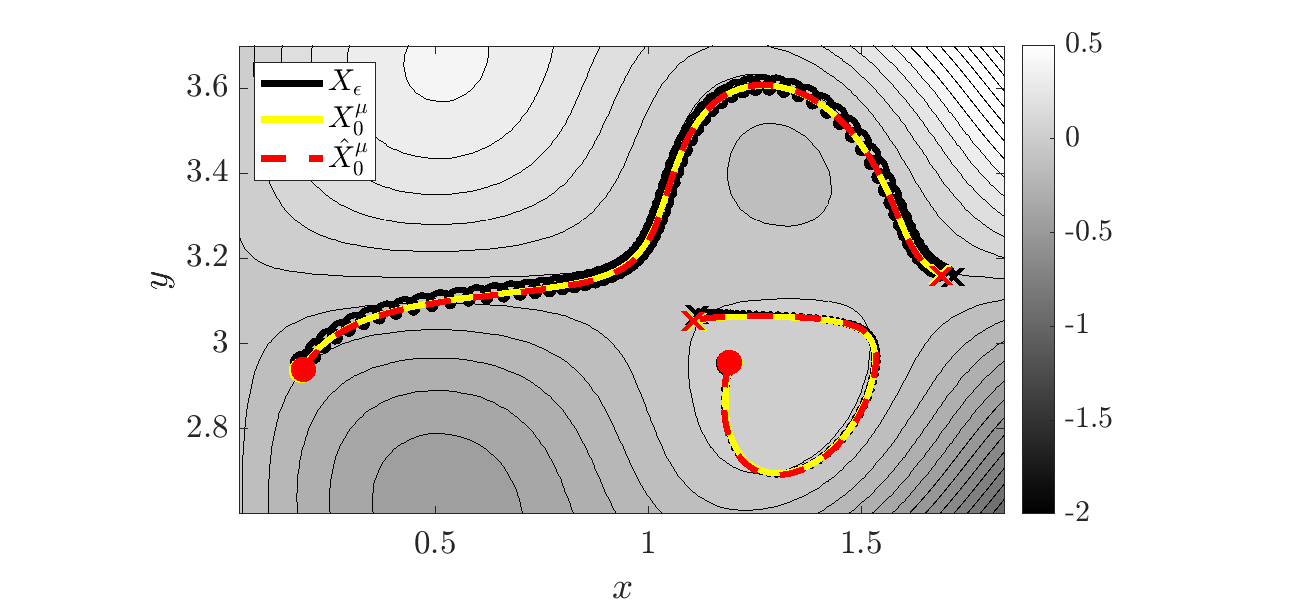} &
	    \includegraphics[trim={30 0 55 10},clip,width=0.5\textwidth]{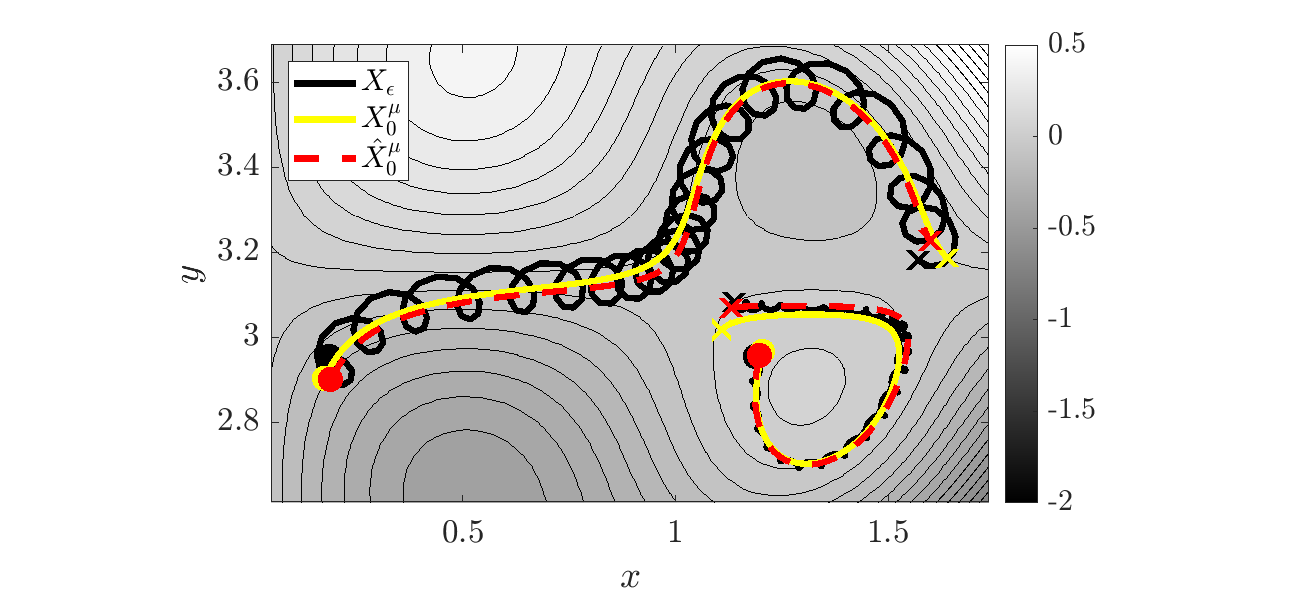} \\ \hline
\end{tabular}
\caption{Dynamics of Example 4 in reduced phase space overlaying contours of the learned background electric field . Left: $\vep= 0.01$ dynamics, right: $\vep=0.03$. Particle two, initialized at $(x_2,y_2) = (1.2,3.0)$, has zero initial velocity and hence smaller perturbations as a result of the Lorentz force.}
\label{fig:exp3_sim}
\end{figure}

The previous example outlines a method of reducing the dimensionality of certain systems by a factor of two. We extend it in our last example to a system of two particles $\{(x_1,y_1,v_{x_1},v_{y_1})$, $(x_2,y_2,v_{x_2},v_{y_2})\}\in \Rbb^{4\times 4}\cong \Rbb^8$ interacting under a pair-wise potential $K$ as well as the same background fields $(\pmb{E},\pmb{B})$ in the previous example. The Hamiltonian describing this motion is 
\begin{align*}
H_\vep      &= \vep^2\,\frac{1}{2}(v_{x_1}^2 + v_{y_1}^2 +v_{x_2}^2 + v_{y_2}^2 ) + \vep\,(\varphi(x_1,y_1) + \varphi(x_2,y_2) + K(x_1,y_1,x_2,y_2)).
\end{align*}
with symplectic form
\begin{align*}
\Omega_\vep &= -\,dx_1\wedge dy_1 -\,dx_2\wedge dy_2 + \vep\,(dx_1\wedge dv_{x_1} + dy_1\wedge dv_{y_1}+dx_2\wedge dv_{x_2} + dy_2\wedge dv_{y_2}).
\end{align*}
If $K$ is the Coulomb potential in $2D$, $K_c(x_1,y_1,x_2,y_2) = Q\log\left(\sqrt{(x_1-x_2)^2+(y_1-y_2)^2}\right)$, the system may be used to classically describe a pair of charged particles ($Q$ corresponding to the product of two charges). To regularize the system we use a crude approximation to $K_c$ given by the first-order Taylor expansion about unit separation
\[K(x_1,y_1,x_2,y_2)= -\frac{Q}{2}\big(1 - (x_1-x_2)^2-(y_1-y_2)^2\big),\]
which provides a reasonable approximation for interparticle distances near unity. We use a repulsive version of this interaction potential with $Q = -(2\pi)^{-1}$, under which the dynamics become
\begin{align*}
\dt{v}_{x_1} &= -\partial_{x_1}\varphi(x_1,y_1)-\partial_{x_1}K(x_1,y_1,x_2,y_2) +v_{y_1} = -\partial_{x_1}\varphi(x_1,y_1)+\frac{1}{2\pi}(x_1-x_2) +v_{y_1}\\
\dt{v}_{y_1} & = -\partial_{y_1}\varphi(x_1,y_1)-\partial_{y_1}K(x_1,y_1,x_2,y_2) - v_{x_1} = -\partial_{y_1}\varphi(x_1,y_1)+\frac{1}{2\pi}(y_1-y_2) - v_{x_1}\\
\dt{x_1} & = \vep\,v_{x_1}\\
\dt{y_1} & = \vep\,v_{y_1},\\
\dt{v}_{x_2} &= -\partial_{x_2}\varphi(x_2,y_2)-\partial_{x_2}K(x_1,y_1,x_2,y_2) +v_{y_2}= -\partial_{x_2}\varphi(x_2,y_2)+\frac{1}{2\pi}(x_2-x_1) +v_{y_2}\\
\dt{v}_{y_2} & = -\partial_{y_2}\varphi(x_2,y_2)-\partial_{y_2}K(x_1,y_1,x_2,y_2) - v_{x_2}= -\partial_{y_2}\varphi(x_2,y_2)+\frac{1}{2\pi}(y_2-y_1) - v_{x_2}\\
\dt{x_2} & = \vep\,v_{x_2}\\
\dt{y_2} & = \vep\,v_{y_2},
\end{align*}
where $\varphi$ is given by \eqref{trigE}. This leads to the leading order roto-rate:
\[ R_0 = (-\partial_{x_1}\varphi(x_1,y_1)+\frac{1}{2\pi}(x_1-x_2) +v_{y_1})\partial_{v_{x_1}} +(-\partial_{y_1}\varphi(x_1,y_1)+\frac{1}{2\pi}(y_1-y_2) - v_{x_1})\partial_{v_{y_1}} \]
\[+\ (-\partial_{x_2}\varphi(x_2,y_2)+\frac{1}{2\pi}(x_2-x_1) +v_{y_2})\partial_{v_{x_2}} +(-\partial_{y_2}\varphi(x_2,y_2)+\frac{1}{2\pi}(y_2-y_1) - v_{x_2})\partial_{v_{y_2}}.\]
Carrying out the same procedures as before (note that presymplectic Hamiltonian dynamics also appear), we arrive at the reduced Hamiltonian
\[\Hred(x_1,y_1,x_2,y_2) = \vep \left(\varphi(x_1,y_1)+\varphi(x_2,y_2) +K(x_1,y_1,x_2,y_2)\right)\]
which under the reduced symplectic form $\Ored = -dx_1\wedge dy_1-dx_2\wedge dy_2$ gives us the reduced Hamiltonian dynamics in $\Rbb^4$,
\begin{align*}
\dt{x_1} & = \vep\,\left( -\partial_{y_1}\varphi(x_1,y_1)+\frac{1}{2\pi}(y_1-y_2)\right)\\
\dt{y_1} & = \vep\,\left(\partial_{x_1}\varphi(x_1,y_1)-\frac{1}{2\pi}(x_1-x_2)\right),\\
\dt{x_2} & = \vep\,\left(-\partial_{y_2}\varphi(x_2,y_2)+\frac{1}{2\pi}(y_2-y_1) \right)\\
\dt{y_2} & = \vep\,\left(\partial_{x_2}\varphi(x_2,y_2)-\frac{1}{2\pi}(x_2-x_1)\right).
\end{align*}

\subsection{Dynamics and sampling regimes}\label{sec:regimes}



Since we are interested in identifying coarse-grained models, we focus on observing only the slow variables and attempt to recover $\Hred$. It is clear from Figure \ref{fig:full_results} that this depends on $\sigma_{\phi f}$ being sufficiently large. In Figure \ref{fig:mt} we will see that this is ensured by choosing $T_\phi$ according to the method described in Appendix \ref{app:cornerpt}, a modification of the procedure in \cite{MessengerBortz2021JComputPhys}. 

To exemplify a range of conditions under which WSINDy identifies the correct coarse-grained dynamics, we quantify the accuracy between the learned  ($\Hredhat$) and true ($\Hred$) reduced Hamiltonian under the following settings:

\begin{enumerate}[label=(\roman*)]
\item {\bf Regions of phase space:} We sample trajectories with initial conditions approaching an elliptic fixed point of $\Hred$, near which the $\vep=0$ dynamics dominate and $\Hred$ fails to be recovered, to probe the near-periodicity of the system. (For Example 4, one particle is placed near an elliptic fixed point and the other a distance one away with zero velocity). 

\item {\bf Perturbative regime:}
We examine two perturbative regimes, labeled the {\it mild} and {\it extreme} regimes. For Examples 1,3, and 4 this is defined by two different values of $\vep$, with $\vep\in \{0.01,0.05\}$ for Example 1 and $\vep\in \{0.01,0.03\}$ for Examples 3 and 4. For Example 2 we fix $\vep$ at the reasonably large value $\vep=0.05$ and define the mild and extreme regimes as the two different sets of initial conditions for the nonlinear pendulum, $(Q(0),P(0)) \in\{(\frac{\pi}{2},0),(\frac{31\pi}{32},0)\}$, the latter driving the system very close to the separatrix at $Q=\pi$, whereby the scale separation disappears (see Figure \ref{fig:exp2_fft}). 
\item {\bf Data sampling regime:}
We fix the timestep $\Delta t$ to be 10 points per fast cycle and the total time window $T$ to be 4 cycles of the slow system, or with $T_s,T_f$ denoting dominant slow and fast periods,
\begin{equation}
T = 4T_s, \qquad \Delta t = T_f / 10.
\end{equation}
From a practical perspective, this is realistic because one cannot expect to sample the fast system at a high resolution; however, this level of resolution is still adversarial because the sampling rate is high enough for WSINDy to resolve the fast system (as evidenced by Figures \ref{fig:full_sims}-\ref{fig:full_results}). With coarser sampling one might expect to more easily recover the reduced dynamics  via aliasing, a property leveraged in \cite{BramburgerDylewskyKutz2020PhysRevE}, although in general, larger $\vep$ leads to mixing of the slow and fast scales over longer time windows. This is why the test function support in Example 2 cannot be taken too large (see Figure \ref{fig:mt}, row 2), hence we conjecture that in general the aliasing approach will lead to inaccurate reduced dynamics. A total time window of 4 cycles ($T= 4T_s$) is used to provide enough data to identify a dominate slow scale (i.e.\ by sampling for longer than a single slow cycle). We also found that $T= 4T_s$ was necessary to identify the dynamics in Example 2 at $Q(0) = 31\pi/32$, as in this case the slow-fast period ratio is $T_{sf} := T_s/T_f\approx 7$, leading to only $278$ timepoints at the resolution $\Delta t = T_f/10$, which did not provide adequate data for robust recovery in the high perturbative regime from time windows $T\leq 3T_s$.
\item {\bf Extrinsic noise:} Realistic measurement settings require consideration of extrinsic (measurement) noise. We define the noise level of the data by 
\begin{equation}\label{sigmas1}
\sigma_{NR} = \frac{\nrm{\Zbf^\star-\Zbf}_2}{\nrm{\Zbf^\star}_2}
\end{equation}
where $\Zbf^\star$ represents the exact (noise-free) data and $\Zbf$ the observed (possibly noisy) data. We consider additive mean-zero i.i.d. Gaussian noise and noise levels up to $\sigma_{NR} = 0.1$ (or 10\% noise), which together with the intrinsic perturbations imparted by the fast scale represents a significant level of corruption. 
\end{enumerate}

\subsection{Hyperparameters}\label{sec:hyperparams}

The WSINDy hyperparameters consist mainly of the library $\Hbb$ and test functions $\Vbb$. The main purpose of this article is not to show robustness to sheer library size, but instead to show that under the weak-form transformations
\[\Zbf \to \lan \dt{V},\Zbf\ran_\tbf, \quad X_{H}(\Zbf) \to \lan V, X_{H}(\Zbf)\ran_\tbf, \quad V\in \Vbb, H\in \Hbb\]
the dynamics agree well with $\Hred$. For this reason, we restrict the library $\Hbb$ to $40-70$ possible terms which include a representation of $\Hred$. 
In Examples 1-4, the leading-order reduced Hamiltonian $\Hred$ can be represented with trigonometric functions, polynomials and products thereof. Define the monomial library of degree $n$ in $2N$ variables by
\begin{equation}\label{pollib}
P^{(n)}_{2N} = \left\{z\to \prod_{i=1}^{2N}z_i^{j_i}\ :\ 1\leq \sum_{i=1}^{2N} j_i\leq n, j_i\in \Nbb\cup \{0\}\right\},
\end{equation}
and the partial cosine library with base frequency $f_0$ up to maximum frequency $nf_0$ as (discarding redundancies)
\begin{equation}\label{coslib}
C^{(n,f_0)}_{2N}= \left\{(q,p)\to \cos(jq_i+kp_i)\pm\cos(jq_i-kp_i)\ :\ (j,k)\in\{0,f_0,\dots,nf_0\}^2, i\in\{1,\dots,N\}\right\}.
\end{equation}
Also define the product linear-trig library with trigonometric terms of frequencies $f_0$ to $nf_0$ by
\begin{equation}\label{prodlib}
LT^{(n,f_0)}_{2N} = \{z\to z_ig(kz_j)\ :\ (i,j)\in \{1,\dots,2N\}^2, g\in \{\cos,\sin\}, k\in \{0,f_0,\dots,nf_0\}\}.
\end{equation}
The libraries used in each example are combinations of $P^{(n)}_{2N}$, $C^{(n,f_0)}_{2N}$, and $LT^{(n,f_0)}_{2N}$, given in Table \ref{libs}.

\begin{table}
\begin{center}
\begin{tabular}{llccl}
Example & $\Hbb$ & $\#\{\Hbb\}$ & $\Hred$ sparsity & $Q$ \\ \hline
1. Coupled Oscillators & $P^{(2)}_2\cup C^{(3,2)}_{2}\cup LT^{(3,2)}_2$  & 41 & 3 & $\Hred$\\
2. H\'enon-Heiles Pendulum & $P^{(4)}_4$ & 69 & 6 & $\Hred(q_1,q_2,0,0)$\\
3. Charged particle & $C^{(4,1)}_{2}$ & 40 & 3 & $\Hred$\\
4. Coupled charged particles & $P^{(2)}_4\cup C^{(3,1)}_{4}$ & 62 & 12 & $\Hred(x_1,y_1,x_1,y_1)$
\end{tabular}
\end{center}
\caption{Algorithmic specifications for examples, including the Hamiltonian libraries $\Hbb$ employed in sparse regression, the total number of terms in $\Hbb$, the number of terms required to represent $\Hred$ from $\Hbb$, and the qauntity of interest $Q$ used to measure agreement with $\Hred$ (see eq. \eqref{delH}). The components $P^{(n)}_{2N}, C^{(n,f_0)}_{2N},LT^{(n,f_0)}_{2N}$ of $\Hbb$ are defined in \eqref{pollib}-\eqref{prodlib}. By $\#\{\Hbb\}$ we denote the cardinality of $\Hbb$.} 
\label{libs}
\end{table}

For the set of test vector fields $\Vbb$, we take the simple convolutional approach as in\cite{MessengerBortz2021JComputPhys}. That is, we fix a reference test function 
\[\phi(t) = \exp\left(\frac{9}{(2t/T_\phi)^2-1}\right)\ind{[-T_\phi/2,T_\phi/2]}(t),\]
which is $C^\infty_c(\Rbb)$ and supported on $[-T_\phi/2,T_\phi/2]$ ($\ind{S}$ denotes the indicator function on the set $S$). We then set the test vector fields to 
\[V_k(z,t) = \phi(t-t_k)\sum_{j=1}^{2N} \partial_{z_j}, \qquad 1\leq k\leq K\]
for a fixed set of {\it query timepoints} $\CalQ := \{t_k\}_{k=1}^K$. The free parameters are $\CalQ$ and the support width $T_\phi$. We choose $T_\phi$ by first finding a cornerpoint $k^*$ in the Fourier spectrum of the data, according to the method in Appendix \ref{app:cornerpt}, and then assigning $k^*$ to be 2 standard deviations into the tail of the power spectrum $|\CalF[\phi(\tbf)]|$, where $\CalF$ is the discrete Fourier transform and $|\CalF[\phi(\tbf)]|$ is interpreted as a probability distribution over Fourier modes. We let $\CalQ$ be equally spaced and covering the given time grid such that $T_\phi/(t_{k+1}-t_k) = 12$ for $1\leq k\leq K$.

\subsection{Performance metrics}\label{sec:metrics}

We measure accuracy of the recovered Hamiltonian $\Hredhat$ with respect to $\Hred$ using three metrics. While all should be considered together to assess performance, each is independently useful and may suffice for specific applications. 

We measure the model selection accuracy using the {\it true positivity ratio} $\text{TPR}(\what)$, defined as 
\begin{equation}\label{tpr}
\text{TPR}(\what) = \frac{TP(\what)}{TP(\what)+FP(\what)+FN(\what)}
\end{equation}
where $TP$ is the number nonzero entries in $\what$ that appear in the correct weight vector $\wstar$, $FP$ is the number of nonzero entries in $\what$ that are zero in $\wstar$, and $FN$ is the number of terms that are zero in $\what$ but nonzero in $\wstar$. In this way, recovering $S$ terms correctly and no false terms leads to $\text{TPR} = S/S^\star$ where $S^\star = \nrm{\wstar}_0$ is the number of true correct terms.

To assess pointwise agreement with $\Hred$, we measure the relative error of a related quantity of interest $Q(z)$, 
\begin{equation}\label{delH}
\Delta H(\what) = \nrm{Q -\widehat{Q}}_{\ell_2(D(\Zbf))}/\nrm{Q}_{\ell_2(D(\Zbf))}
\end{equation} 
where $\widehat{Q}$ is the learned quantity of interest, $D(\Zbf)$ is an equally-spaced computational grid covering the smallest rectangle in phase space containing the observed data $\Zbf$, and $\nrm{Q}_{\ell_2(D(\Zbf))} = \sqrt{\sum_{z\in D(\Zbf)} Q(z)^2}$. The quantities $Q$ used for each example are listed in Table \ref{libs}, with $Q = \Hred$ for Examples 1 and 3. For Example 2 we use the zero-momentum section, $Q(z) = \Hred(q_1,q_2,0,0)$, which captures agreement with the potential field $\overline{V}$. In Example 4 we use agreement with the $(x_1,y_1) = (x_2,y_2)$ section, $Q(z) =\Hred(x_1,y_1,x_1,y_1)$, which measures agreement with the background electric field $\varphi(x,y)$.

Lastly, we measure agreement with forward simulations of the Hamiltonian dynamics given by $\Hred$ and $\Hredhat$. This is a more subtle performance metric because the initial conditions for the reduced system are not simply the restriction of the initial conditions to the slow variables. In addition, the full dynamics may be chaotic, which may also lead to chaos in the reduced system if the reduced system has two or more degrees of freedom (e.g.\ Examples 2 and 4). To find accurate initial conditions we be do a minimal grid as described in Appendix \ref{app:IC}. To deal with chaos, we only measure agreement up to the first full period of the observed variables, $T_s$, defined as the period of the dominant Fourier mode in the dynamics. Letting $\Zbf$, $\Zbf^\mu$, and $\widehat{\Zbf}^\mu$ denote the observed data, the forward simulation from the true reduced system $\Hred$, and the forward simulation from $\Hredhat$, respectively, we define the following forward simulation errors:
\begin{align}
\label{delZ}\Delta \Zbf &= \nrm{\textsf{vec}(\Zbf - \widehat{\Zbf}^\mu)}_2/\nrm{\textsf{vec}(\Zbf)}_2\\
\label{delZmu}\Delta \Zbf^\mu &= \nrm{\textsf{vec}(\Zbf^\mu - \widehat{\Zbf}^\mu)}_2/\nrm{\textsf{vec}(\Zbf^\mu)}_2\\
\label{delZstar}\Delta \Zbf^\star &= \nrm{\textsf{vec}(\Zbf - \Zbf^\mu)}_2/\nrm{\textsf{vec}(\Zbf)}_2.
\end{align}
The value $\Delta \Zbf$ is not expected to be small for larger $\vep$, but assesses whether the trajectory exhibits a phase error from the full system, wheres $\Delta \Zbf^\mu$ assesses agreement with the analytical reduced system. The value $\Delta \Zbf^\star$ is a reference measure and serves as an approximate lower bound for $\Delta \Zbf$. 

\subsection{Dependence on test function radius}\label{sec:Tfsupport}

As mentioned previously, the ratio between the test function support width $T_\phi$ and the fast scale $T_f$, $\sigma_{\phi f} = T_\phi/T_f$, plays a role in the accuracy of $\Hredhat$ with respect to $\Hred$. This relationship is graphically represented in Figure \ref{fig:mt}, where we present the TPR, $\Delta H$, and $\Delta \Zbf^\mu$ values (equations \eqref{tpr}-\eqref{delZstar}) for each example over a range of $\sigma_{\phi f}$, as well as the $\sigma_{\phi f}$ values resulting from $T_\phi$ as computed using the method in this paper (black curves), see Appendix \ref{app:cornerpt}. In Figure \ref{fig:mt} we display only the results for the extreme perturbative regime ($\vep=0.05$ for Examples 1 \& 2 and $\vep=0.03$ for Examples 3 \& 4) as the method performs very well for a wide range of $\sigma_{\phi f}$ in the milder perturbative regime. 

Several general trends can be observed when $\sigma_{\phi f}$ is varied. From the perspective of the TPR score, it can been seen that the optimal $T_\phi$ is usually associated with some $\sigma_{\phi f}\in [5,30]$, and the exact optimal $\sigma_{\phi f}$ is highly dependent on the trajectory. Comparing the first and second columns, we see that TPR=1 always correlates with lower $\Delta H$, yet there are instances where the Hamiltonian is captured accurately (low $\Delta H$) without correct identification of $\Hred$ (e.g.\ Example 3, $z(0)$ index 10). On the other hand, an accurate forward solve does not always correlate with TPR, so from the perspective of $\Delta \Zbf^\mu$ we get a different optimal $\sigma_{\phi f}$. Often larger $\sigma_{\phi f}$ will yield lower forward simulation errors $\Delta \Zbf^\mu$ despite not capturing the correct form of $\Hred$ (i.e.\ TPR<1). This is most readily observed from Examples 3 and 4 (right column), as well as Example 1 for $z(0)$ indices $6-10$. This serves to highlight the importance of examining multiple performance metrics, as the down-stream task (forward simulations, scientific inference, etc.) should determine which metric is weighted most highly.

The black curve indicates values of $\sigma_{\phi f}$ (hence $T_\phi$) resulting from default settings of the method as presented here. In all cases except $z(0)$ indices $\{8,9,10,12\}$ of Example 1 and $z(0)$ index 10 in Example 3, the identified $T_\phi$ yields the correct model terms, which automatically grants excellent agreement with $\Hred$ (column 2). In these five cases with $\text{TPR}<1$, the Hamiltonian is still very accurate, with $\Delta H\approx 0.02$. For  Example 2, the black curve successfully lies in the region of admissible $\sigma_{\phi f}$ values, which is narrow due to mixing of slow and fast scales which causes lower $\Delta H$ and $\Delta \Zbf^\mu$ at larger $T_\phi$. Results are similar for Example 3. In Example 4 the correct model is found for nearly all $\sigma_{\phi f}$ and $z(0)$, although larger $T_\phi$ may increase accuracy.

\begin{figure}
\begin{center}
\begin{tabular}{m{0.1\textwidth}@{}m{0.3\textwidth}@{}m{0.3\textwidth}m{0.3\textwidth}}
& \begin{center}{\fbox{TPR}}\end{center} & \begin{center}{\fbox{$\Delta H$}}\end{center} & \begin{center}{\fbox{$\Delta \Zbf^\mu$}}\end{center} \\
\fbox{$\substack{\text{Example 1} \\ \vep=0.05}$}
&\hspace{1mm}\includegraphics[trim={0 0 0 0},clip,width=0.3\textwidth]{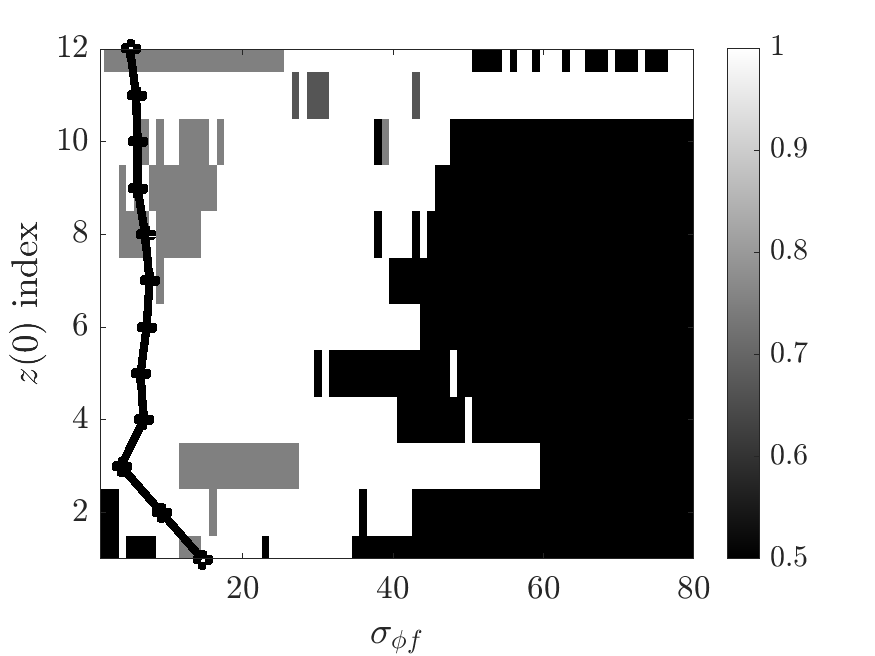} &\hspace{1mm}
		\includegraphics[trim={0 0 0 0},clip,width=0.3\textwidth]{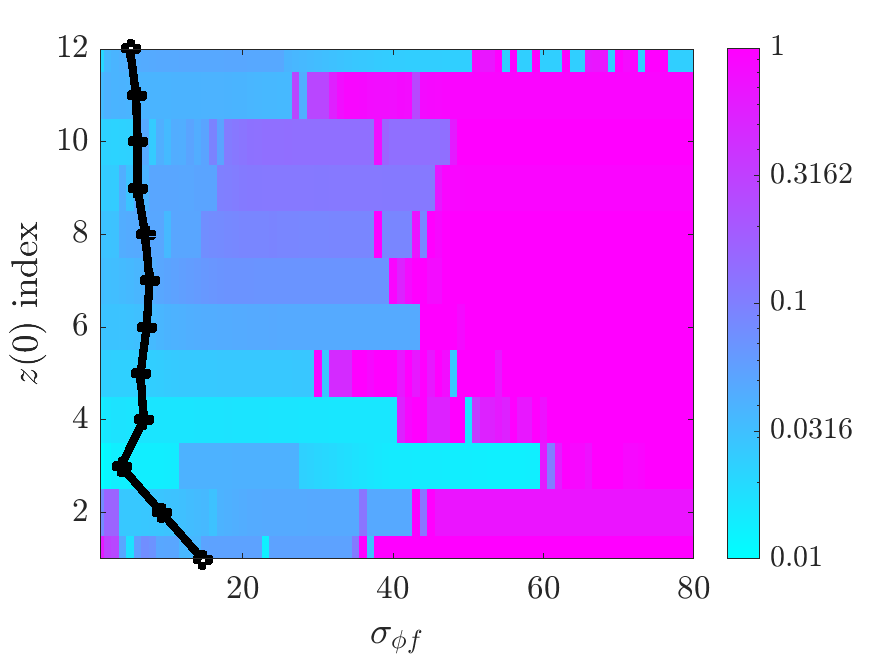} &
		\includegraphics[trim={0 0 0 0},clip,width=0.3\textwidth]{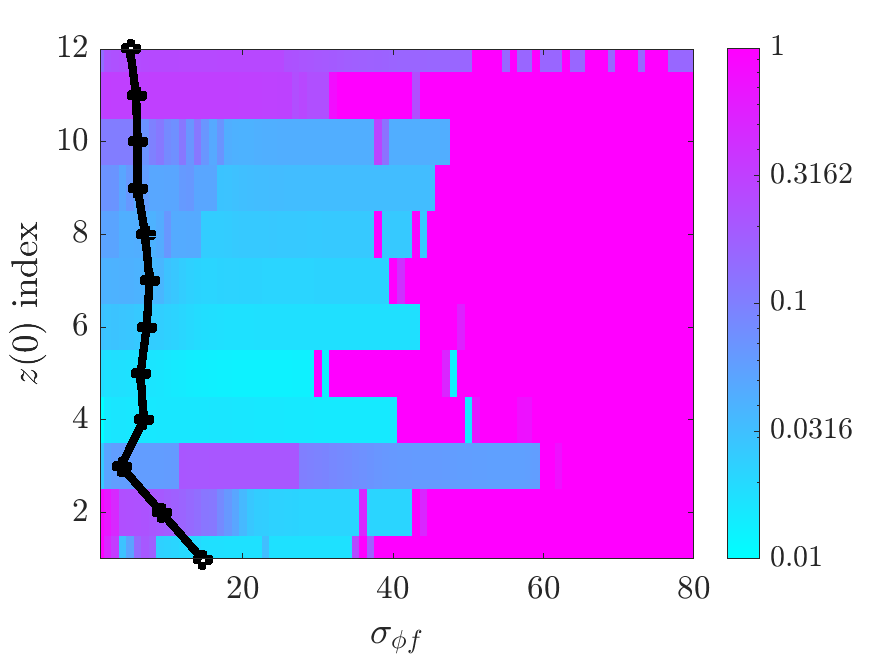} \\
\fbox{$\substack{\text{Example 2} \\ \vep=0.05 \\ Q(0)=\frac{31\pi}{32}}$}
&\includegraphics[trim={0 0 0 0},clip,width=0.3\textwidth]{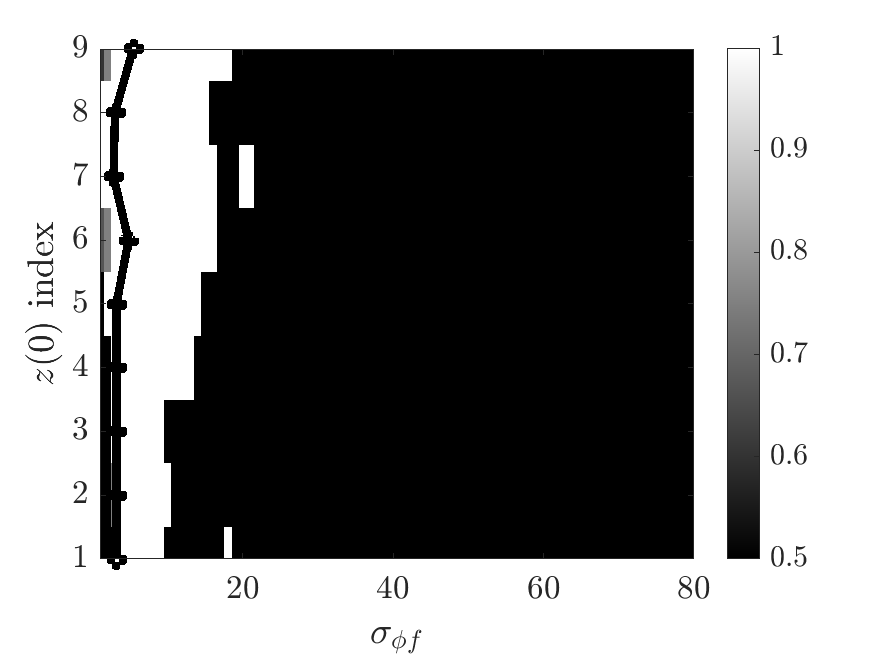} &
		\includegraphics[trim={0 0 0 0},clip,width=0.3\textwidth]{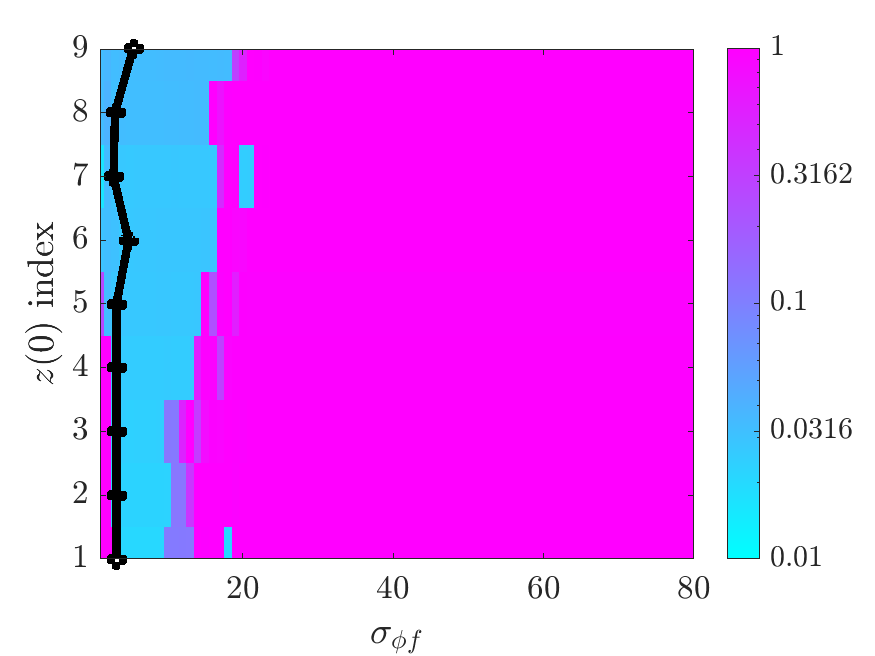} &
		\includegraphics[trim={0 0 0 0},clip,width=0.3\textwidth]{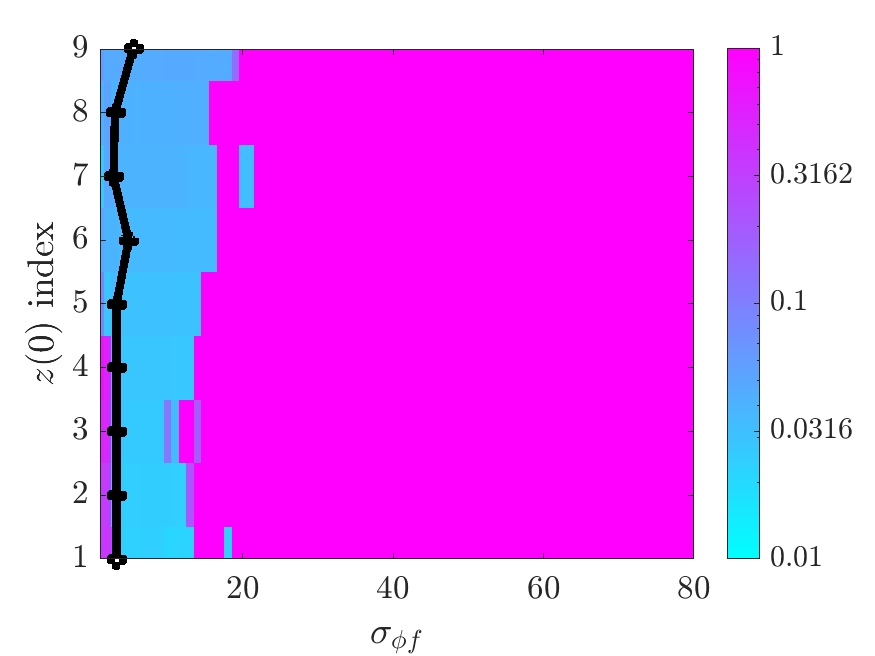} \\
\fbox{$\substack{\text{Example 3} \\ \vep=0.03}$} &
\includegraphics[trim={0 0 0 0},clip,width=0.3\textwidth]{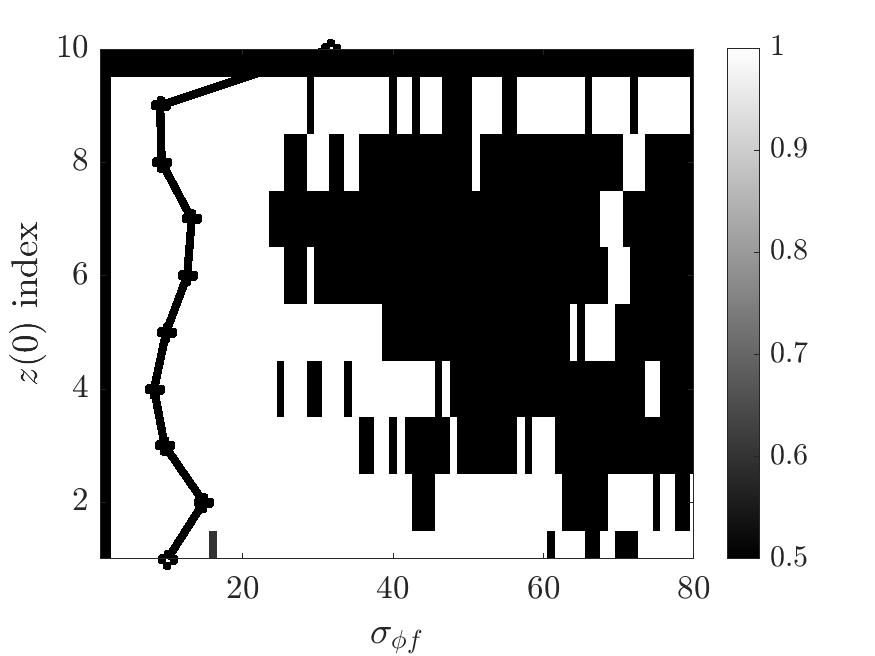} &\hspace{1mm}
		\includegraphics[trim={0 0 0 0},clip,width=0.3\textwidth]{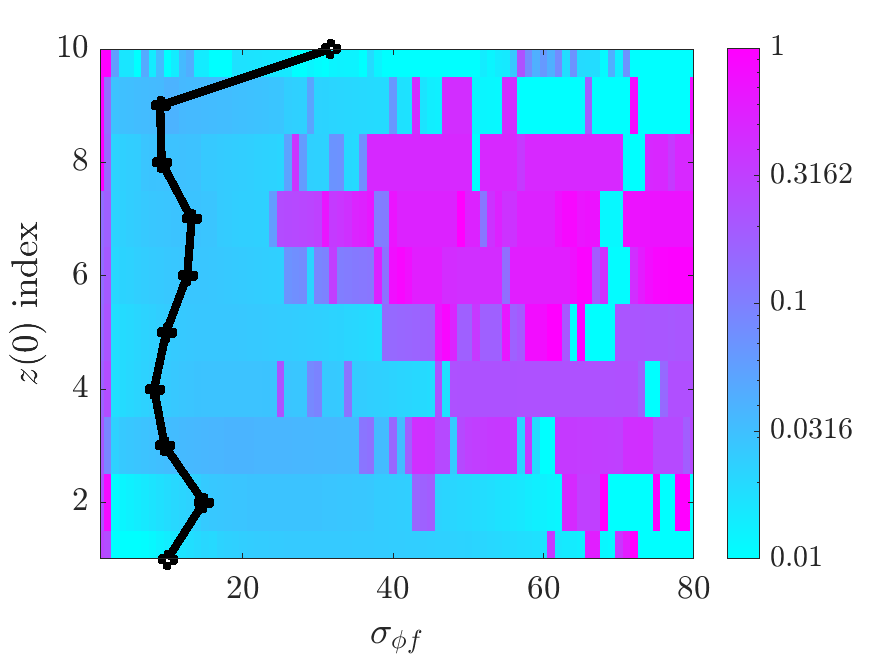} &
		\includegraphics[trim={0 0 0 0},clip,width=0.3\textwidth]{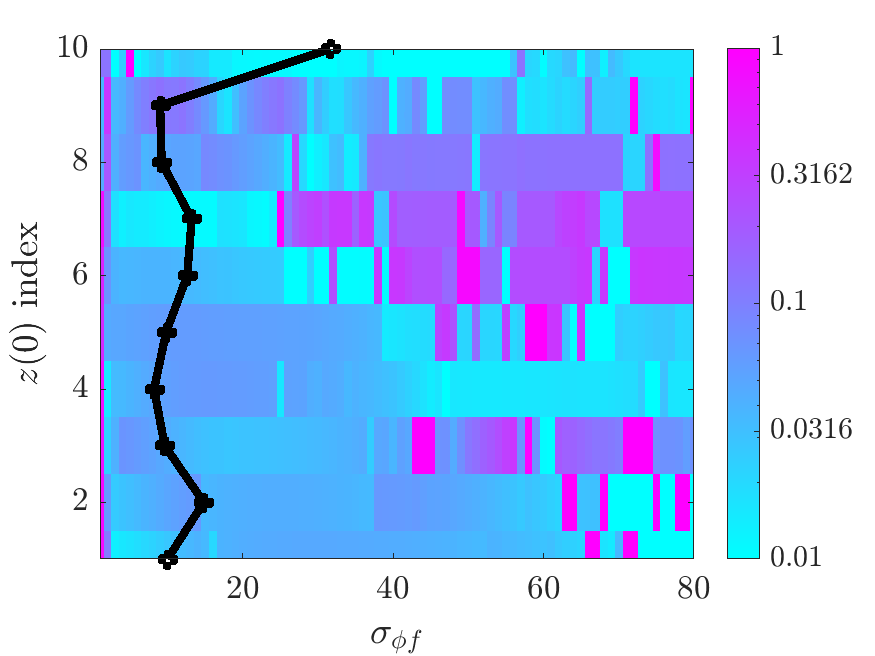} \\
\fbox{$\substack{\text{Example 4} \\ \vep=0.03}$}
&\includegraphics[trim={0 0 0 0},clip,width=0.3\textwidth]{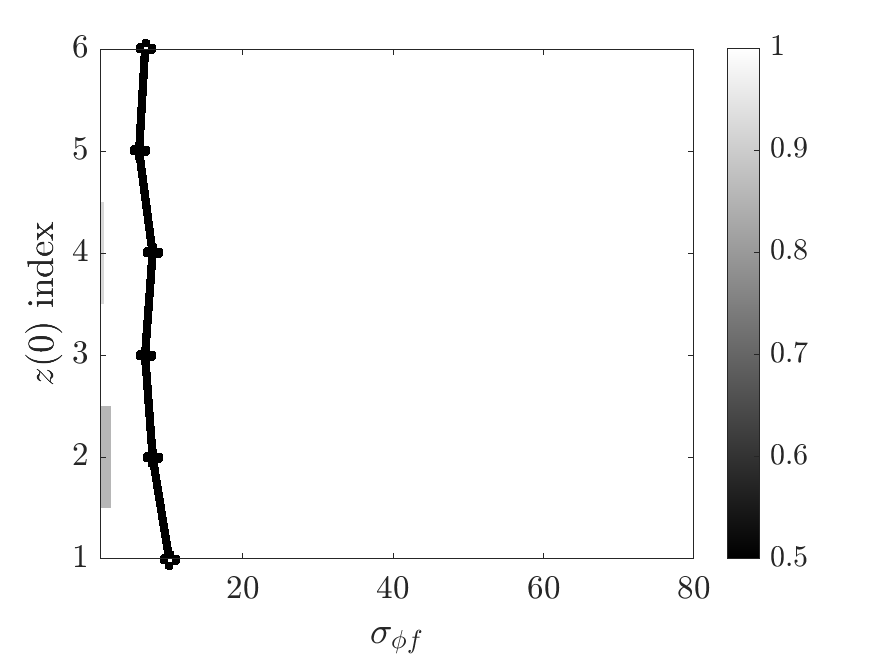} &
		\includegraphics[trim={0 0 0 0},clip,width=0.3\textwidth]{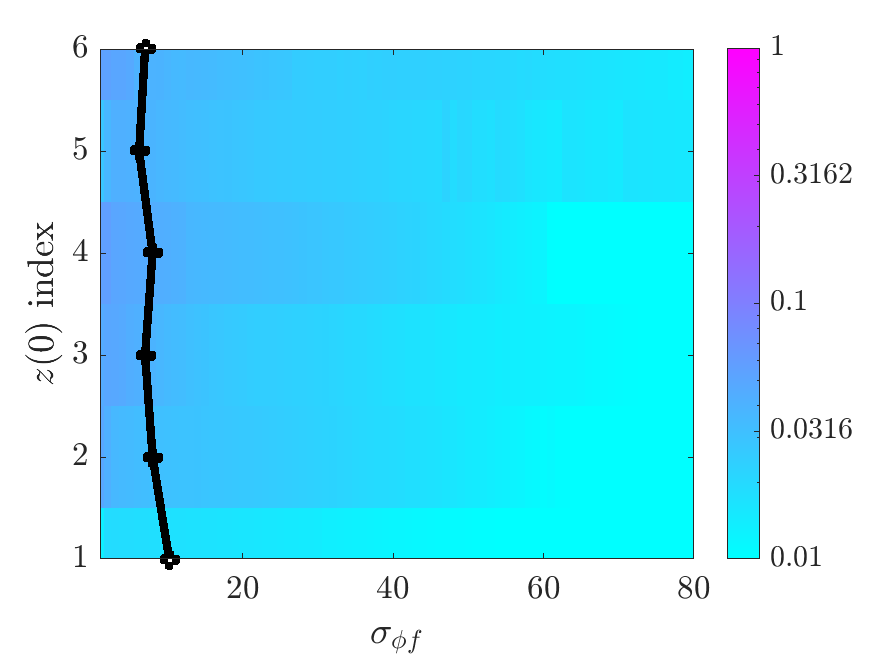} &
		\includegraphics[trim={0 0 0 0},clip,width=0.3\textwidth]{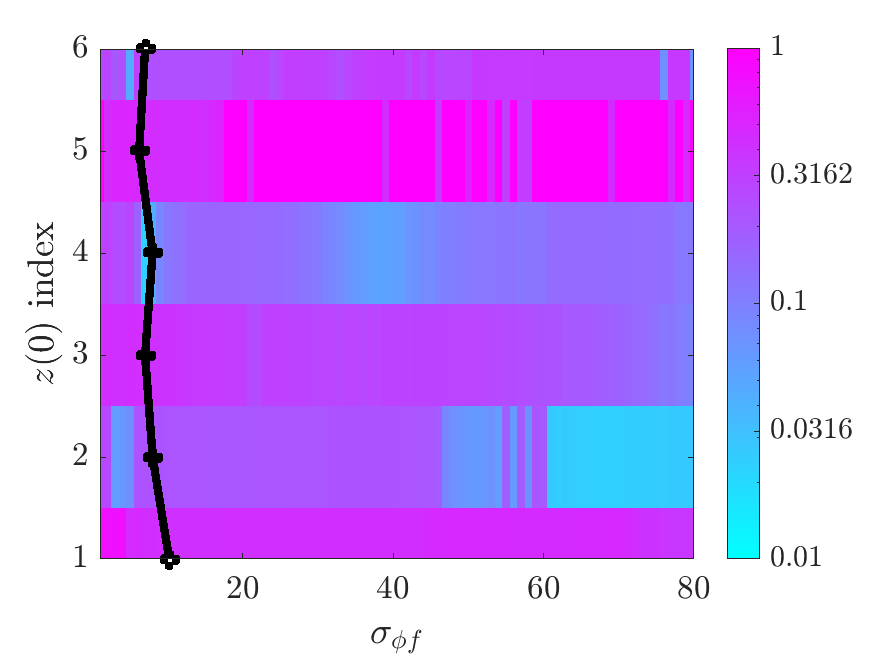} \\ \hline
\end{tabular}
\end{center}
\caption{Values of the TPR, $\Delta H$, and $\Delta \Zbf^\mu$ (see eqs. \eqref{tpr}-\eqref{delZstar}) statistics in the extreme perturbative regime as a function of the ratio $\sigma_{\phi f} = T_\phi/T_f$ (see eq. \eqref{sigmas2}) and the learning trajectory (indicated on the $y$-axes as the $z(0)$ index).}
\label{fig:mt}
\end{figure}

\subsection{Results: zero extrinsic noise ($\sigma_{NR}=0$)}

In this section we demonstrate that the reduced Hamiltonian $\Hred$ is sufficiently recovered over a wide range of initial conditions and perturbative regimes using a single trajectory. This implies that often only a small sample in phase space is needed to recovery the entire reduced Hamiltonian, as opposed to neural-network based approaches, which are often trained on $\CalO(10^4)$ input-output pairs with substantially longer computation times. 

In the top three rows of Figures \ref{fig:tpr}-\ref{fig:exp4_st} we plot the learned trajectories on a black-to-red scale overlaying the training data in cyan. The color of each learned trajectory indicates the value of the given statistic (TPR, $\Delta H$, $\Delta \Zbf^\mu$, $\Delta \Zbf$). As a reference, the bottom rows of Figures \ref{fig:exp1_st}-\ref{fig:exp4_st} plot simulations from $\Hred$, colored according to $\Delta \Zbf^\star$, for comparison with $\Delta \Zbf$.

\subsubsection{Model identification (TPR)}

\begin{figure}
\begin{center}
\begin{tabular}{m{0.2\textwidth}@{}m{0.35\textwidth}@{}m{0.35\textwidth}}
\fbox{Example 1}
$\vep\in\{0.01,0.05\}$ &\hspace{1mm}\includegraphics[trim={20 0 30 10},clip,width=0.32\textwidth]{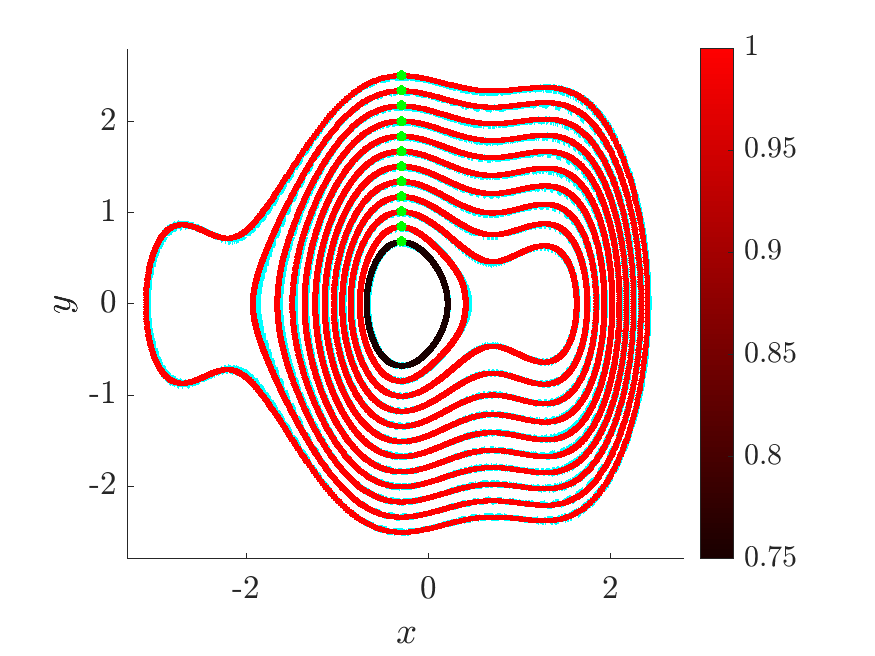} &\hspace{1mm}
		\includegraphics[trim={20 0 30 10},clip,width=0.32\textwidth]{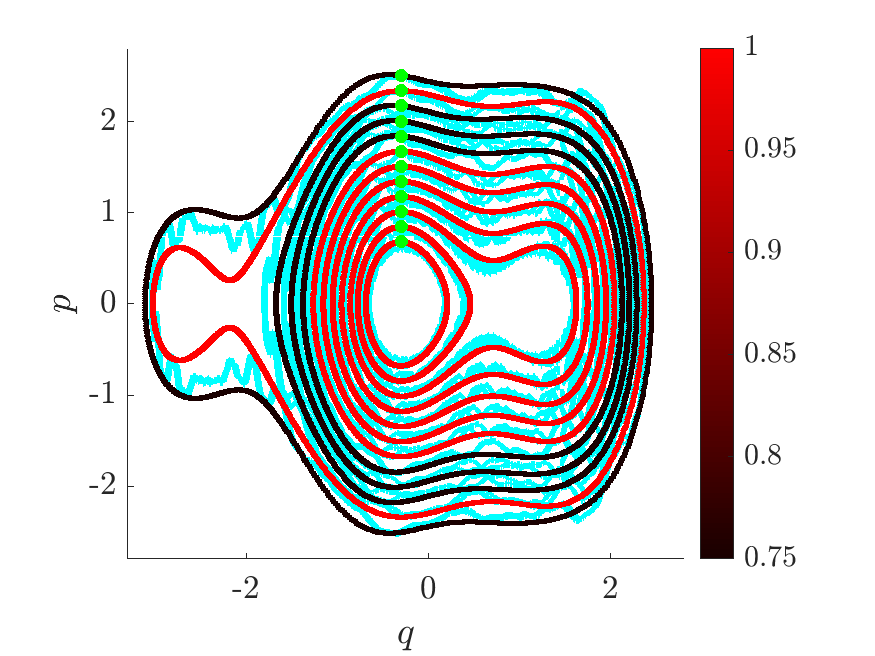} \\
\fbox{Example 2}\textcolor{white}{.................}
$\vep=0.05$, $Q(0)\in\{\frac{\pi}{2},\frac{31\pi}{32}\}$ &\includegraphics[trim={0 0 15 10},clip,width=0.33\textwidth]{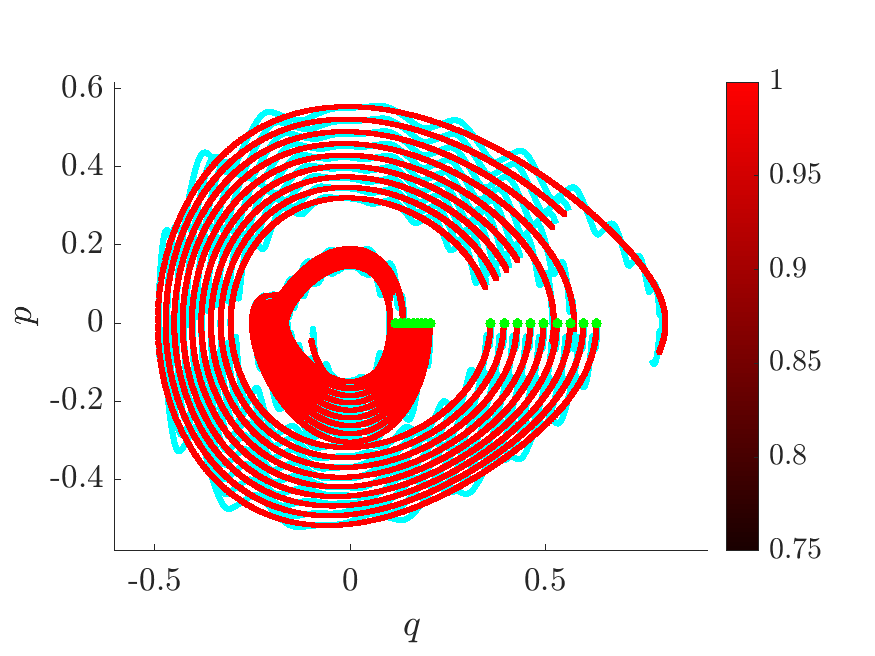} &
		\includegraphics[trim={0 0 15 10},clip,width=0.33\textwidth]{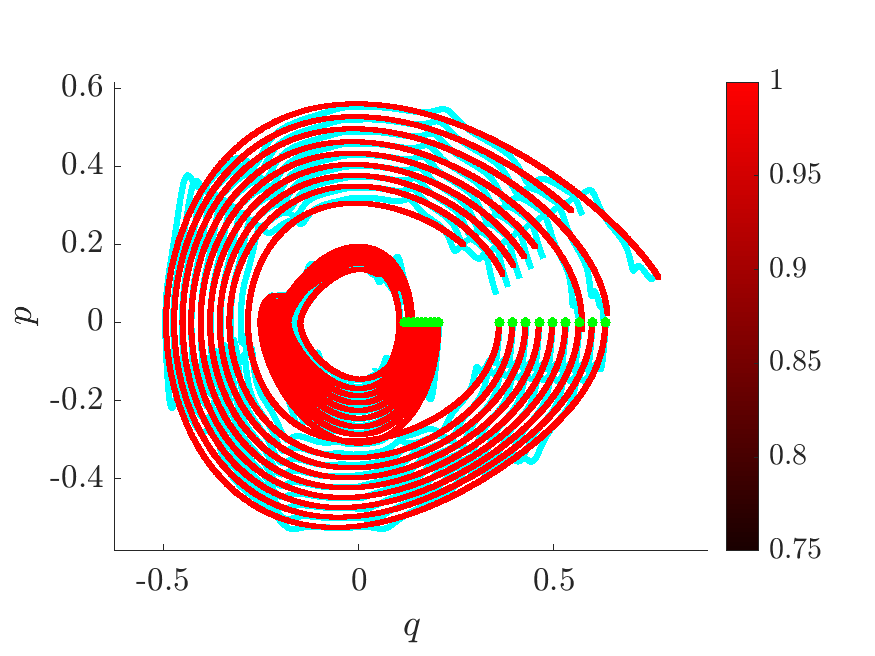} \\
\fbox{Example 3}
$\vep\in\{0.01,0.03\}$ &\hspace{1.5mm}\includegraphics[trim={30 0 30 10},clip,width=0.32\textwidth]{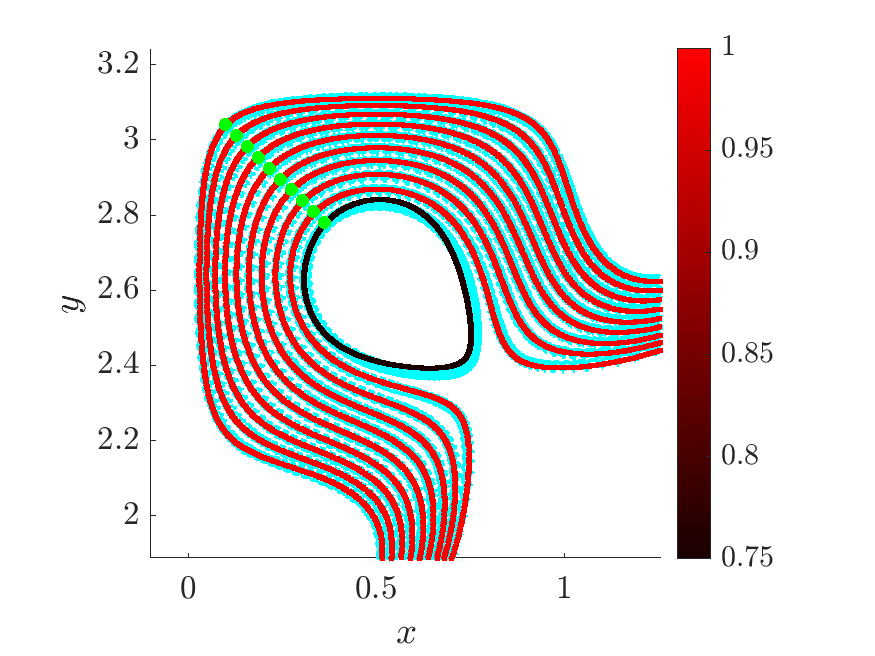} &\hspace{1mm}
		\includegraphics[trim={30 0 30 10},clip,width=0.32\textwidth]{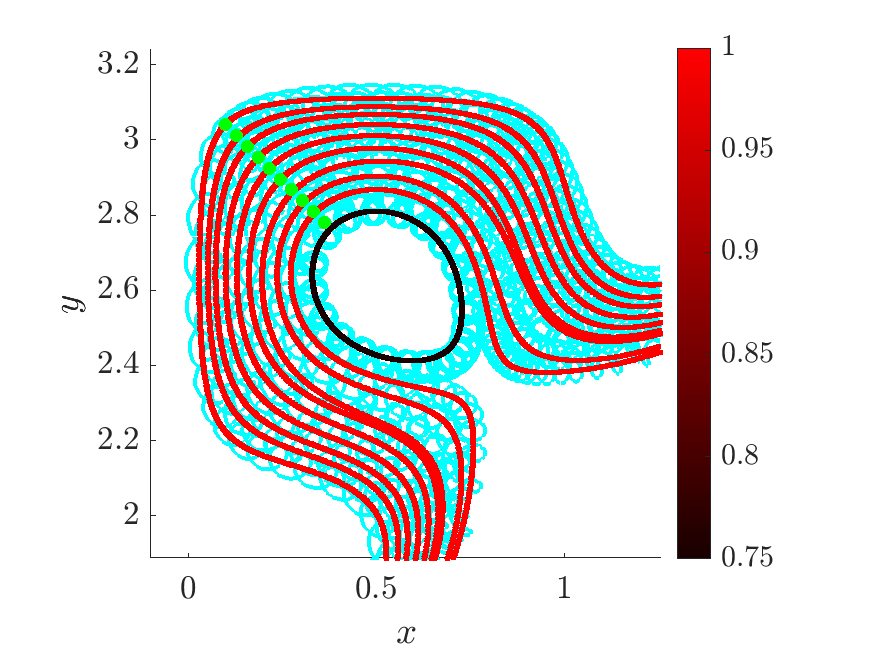} \\
\fbox{Example 4}
$\vep\in\{0.01,0.03\}$ &\includegraphics[trim={30 0 30 10},clip,width=0.32\textwidth]{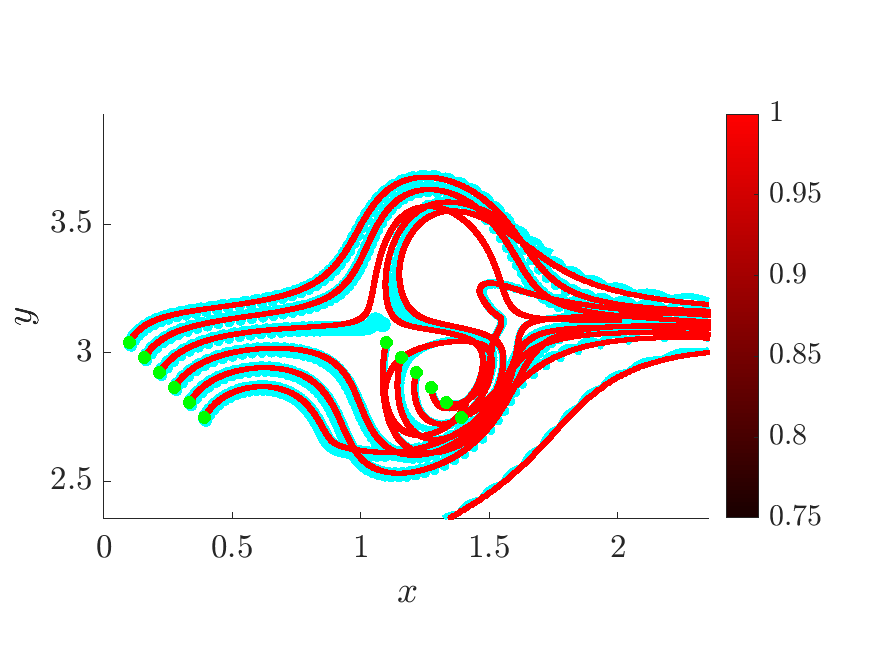} &
		\includegraphics[trim={30 0 30 10},clip,width=0.32\textwidth]{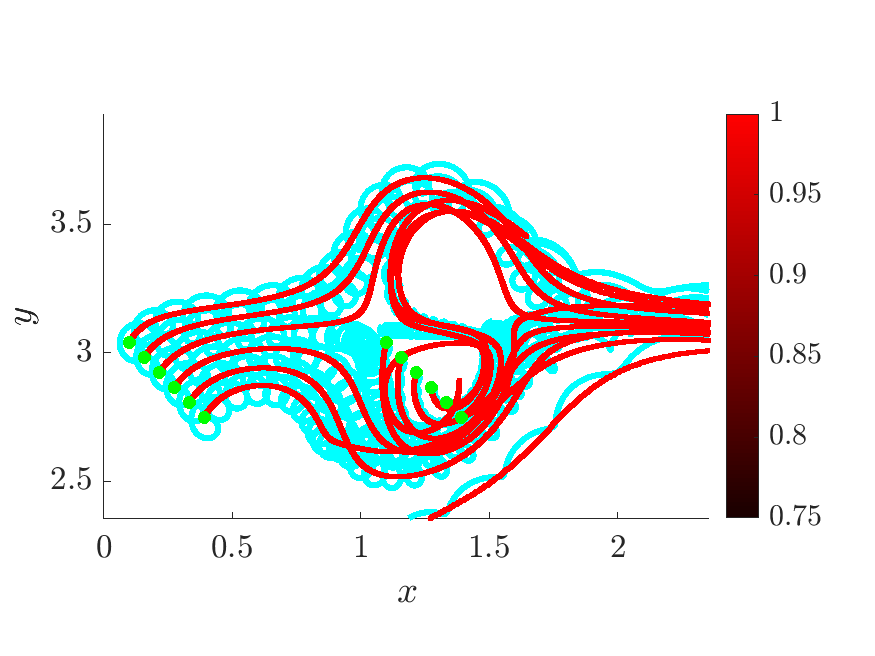} \\ \hline
\end{tabular}
\end{center}
\caption{TPR values for each example (defined in \eqref{tpr}). Each trajectory in red is simulated from the learned Hamiltonian system $\Hredhat$ and is colored according to its TPR. Green dots indicate initial conditions and the training data is plotted in cyan. The correct model is recovered in the vast majority of cases ($\text{TPR}=1$), with exceptions still corresponding to accurate learned dynamics (see Figures \ref{fig:exp1_st}-\ref{fig:exp4_st}). Note also that the right column corresponds to the black curves in Figure \ref{fig:mt}.}
\label{fig:tpr}
\end{figure}

In Figure \ref{fig:tpr} we report the TPR score (defined in \eqref{tpr}) for each trajectory, indicating that in the vast majority of cases WSINDy recovers the correct terms in $\Hred$ even under significant perturbations imparted by the fast scale. In the left column we see that the mild perturbative regime yields a TPR of 1 (i.e.\ perfect model recovery) for all cases except the inner-most trajectories of Examples 1 and 3, which are highly degenerate, nearly-circular orbits. Moreover, trajectories from these misspecified models show qualitative agreement with the true dynamics. It is most surprising that in Example 1 there exist trajectories that only enclose one of the relevant fixed points and still enable identification of $\Hred$. In the extreme perturbative regime (right column of Figure \ref{fig:tpr}), we recover TPR=1 in all cases except indices $\{8,9,10,12\}$ of Example 1 and index 10 in Example 3, yet the learned trajectories are still accurate. (Note that the initial conditions ($z(0)$ indices) referred to in Figure \ref{fig:mt} correspond to the green dots in Figures \ref{fig:exp1_st}-\ref{fig:exp4_st}, with increasing $z(0)$ index corresponding to increasing magnitude $|z(0)|$ in Examples 1 and 2 and to $z(0)$ moving diagonally downwards in Examples 3 and 4.)

\subsubsection{Pointwise agreement with ${\normalfont \Hred}$ ($\Delta H$)}

\begin{table}
\begin{center}
\begin{tabular}{l|llll}
  & Ex1 & Ex2  & Ex3  & Ex4 \\ \hline
(min $\Delta H$, max $\Delta H$), mild & $(0.003,0.012)$, outlier $0.087$ &  $(0.004,0.006)$ &  $(0.003,0.012)$ &  $(0.006,0.018)$ \\  
(min $\Delta H$, max $\Delta H$), extreme & $(0.014,0.053)$ &  $(0.021,0.033)$ & $(0.012,0.035)$ & $(0.017,0.044)$  
\end{tabular}
\end{center}
\caption{Range of $\Delta H$ values across initial conditions for each example.}
\label{tab:H}
\end{table}

Comparing the top three rows with the bottom rows of Figures \ref{fig:exp1_st}-\ref{fig:exp4_st}, we observe that for each example, initial condition, and perturbative regime, the learned trajectory provides excellent qualitative agreement with the $\Hred$ dynamics. That is, the level curves of $\Hredhat$ are very close to those of $\Hred$. More quantitatively, the range of $\Delta H$ is given in Table \ref{tab:H}. For the mild perturbative regime (left columns of Figures \ref{fig:exp1_st}-\ref{fig:exp4_st}), the method produces a maximum $\Delta H$ of $0.012$ (with the exception of one outlier trajectory with $\Delta H = 0.087$), and for the extreme regime the maximum is $\Delta H = 0.053$. To summarize, given a single trajectory from $H_\vep$, the learned Hamiltonian $\Hredhat$ using WSINDy agrees with $\Hred$ over a large region of phase space, and with very mild dependence on the perturbative regime: $\Delta H$ is upper-bounded by $\vep$ in almost all cases.

\subsubsection{Forward simulation errors ($\Delta \Zbf^\mu$,$\Delta \Zbf$,$\Delta \Zbf^\star$)}

\begin{table}
\begin{center}
\begin{tabular}{l|llll}
  & Ex1 & Ex2  & Ex3  & Ex4 \\ \hline
(min $\Delta \Zbf^\mu$, max $\Delta \Zbf^\mu$), mild & $(0.004,0.082)$ & $(0.007,0.015)$ &  $(0.001,0.022)$ &  $(0.036,0.419)$ \\  
(min $\Delta \Zbf^\mu$, max $\Delta \Zbf^\mu$), extreme & $(0.015,0.315)$ &  $(0.024,0.044)$ & $(0.005,0.128)$ & $(0.043,0.435)$  \\
(min $\Delta \Zbf$, max $\Delta \Zbf$), mild & $(0.015,0.124)$ & $(0.056,0.064)$ &  $(0.005,0.018)$ &  $(0.031,0.314)$ \\  
(min $\Delta \Zbf$, max $\Delta \Zbf$), extreme & $(0.069,0.579)$ & $(0.065,0.067)$ & $(0.019,0.145)$ & $(0.050,0.468)$  \\
(min $\Delta \Zbf^*$, max $\Delta \Zbf^*$), mild & $(0.015,0.084)$ &  $(0.058,0.065)$ &  $(0.005,0.019)$ &  $(0.015,0.238)$ \\  
(min $\Delta \Zbf^*$, max $\Delta \Zbf^*$), extreme & $(0.067,0.409)$ &  $(0.071,0.080)$ & $(0.018,0.068)$ & $(0.082,0.408)$  
\end{tabular}
\end{center}
\caption{Range of forward simulation errors across initial conditions for each example.}
\label{tab:Z}
\end{table}

From Figures \ref{fig:exp1_st}-\ref{fig:exp4_st} and Table \ref{tab:Z} we observe that in the mild regime, Examples 1-3 yield $\CalO(10^{-2})$ errors for both $\Delta \Zbf^\mu$ (agreement with $\Hred$) and $\Delta \Zbf$ (agreement with $H_\vep$). Example 4 exhibits much larger forward simulation errors due to the chaotic nature of the dynamics (yet leads to $\CalO(10^{-2})$ values for $\Delta H$ and TPR=1 in all cases).

The right columns of Figures \ref{fig:exp1_st}-\ref{fig:exp4_st} display $\Delta \Zbf^\mu$ and $\Delta \Zbf$ for the extreme perturbative regime, where larger values of $\Delta \Zbf^\mu$ and $\Delta \Zbf$ can be observed for Examples 1 and 4 (Figures \ref{fig:exp1_st} and \ref{fig:exp4_st}), while Examples 2 and 3 accurately capture both the reduced dynamics of $\Hred$ and the full dynamics of $H_\vep$.  The method performs exceptionally well on Example 2, offering better agreement with $H_\vep$ than that provided by $\Hred$, which is especially surprising because of the mixing of slow and fast scales (see Figure \ref{fig:exp2_fft}). For Example 1, the larger forward simulation errors are due to minor phase differences which lead to rapid accumulation of errors (despite close qualitative agreement), while Example 4 suffers again because of its chaotic nature. 

Ultimately, our findings indicate that coarse-grained models provided by WSINDy have the potential to be very useful surrogate models in forward simulations. However, Example 1 indicates that capturing the phase of a slow oscillator for arbitrarily many periods may require additional constraints on the model. On the other hand, all examples indicate that level sets of the reduced Hamiltonian are readily captured by the present method. As well as providing insight into the energy landscape, forward simulations may incorporate knowledge about the level sets. We will pursue these lines of research in a future work.  


\begin{figure}
\begin{center}
\begin{tabular}{m{0.08\textwidth}@{}m{0.35\textwidth}@{}m{0.35\textwidth}}
& \centering\fbox{$\vep=0.01$}  
& \begin{center}\fbox{$\vep=0.05$}\end{center} \\
\fbox{$\Delta H$}  &	\includegraphics[trim={30 0 30 10},clip,width=0.32\textwidth]{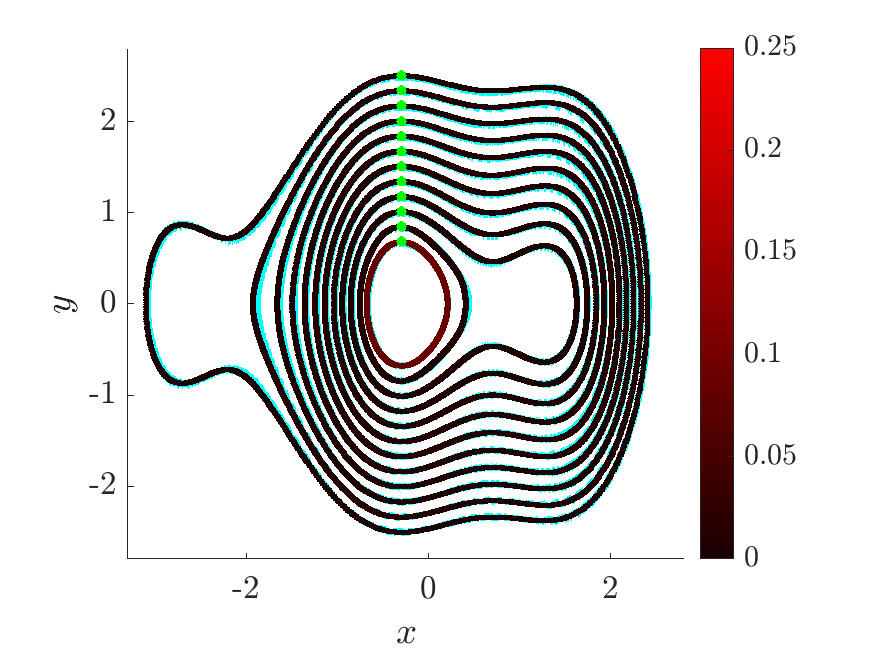} &
		\includegraphics[trim={30 0 30 10},clip,width=0.32\textwidth]{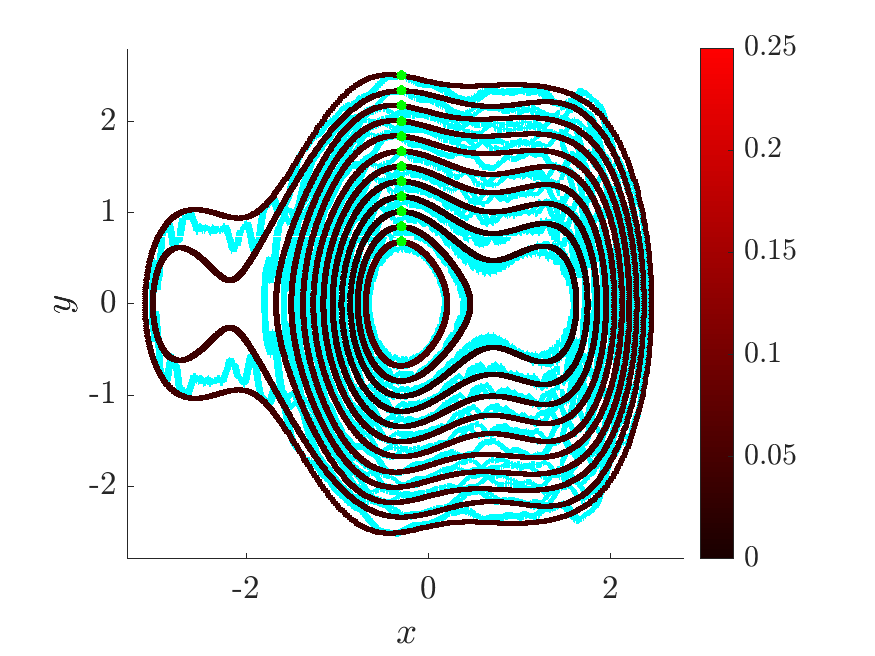} \\
\fbox{$\Delta \Zbf^\mu$} &		\includegraphics[trim={30 0 30 10},clip,width=0.32\textwidth]{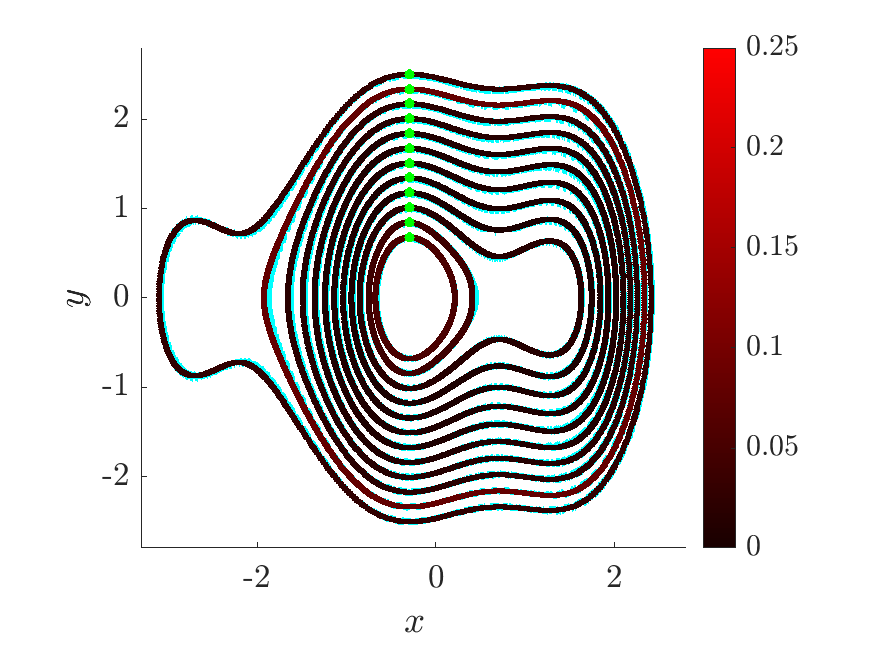} &
		\includegraphics[trim={30 0 30 10},clip,width=0.32\textwidth]{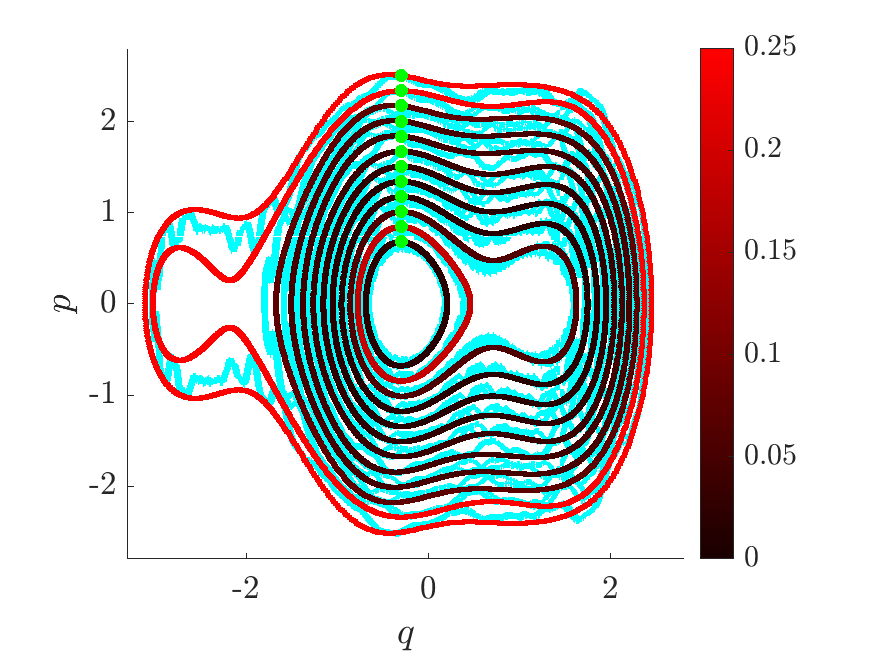} \\
\fbox{$\Delta \Zbf$} &		\includegraphics[trim={30 0 30 10},clip,width=0.32\textwidth]{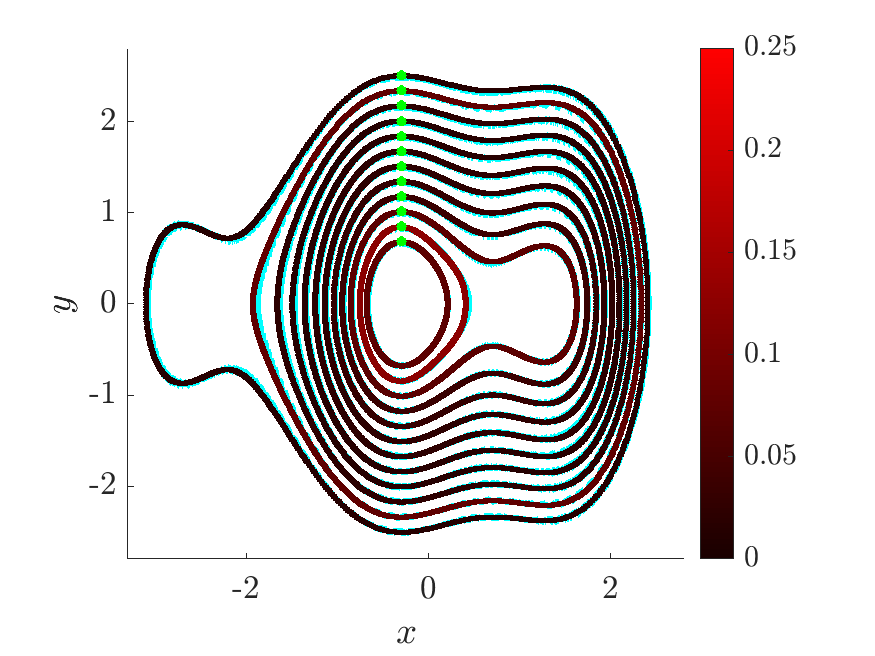} &
		\includegraphics[trim={30 0 30 10},clip,width=0.32\textwidth]{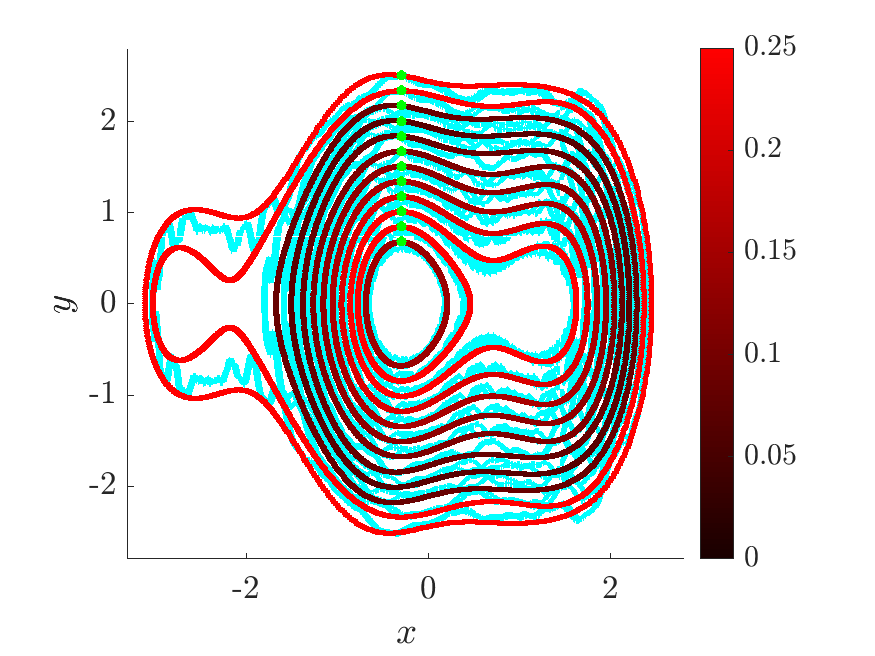} \\ \hline
	 \fbox{$\Delta \Zbf^\star$} &	\includegraphics[trim={30 0 30 10},clip,width=0.32\textwidth]{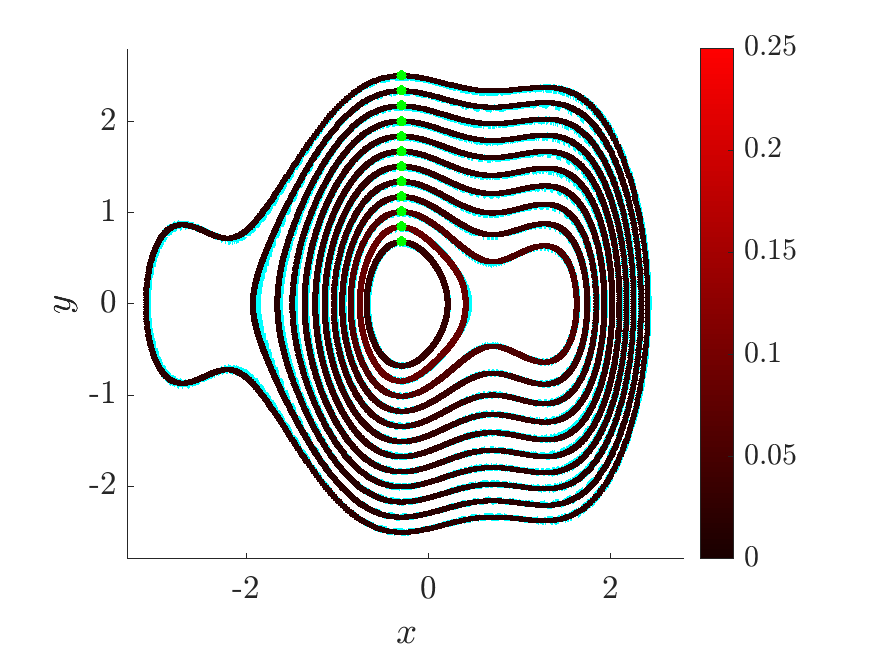} &
		\includegraphics[trim={30 0 30 10},clip,width=0.32\textwidth]{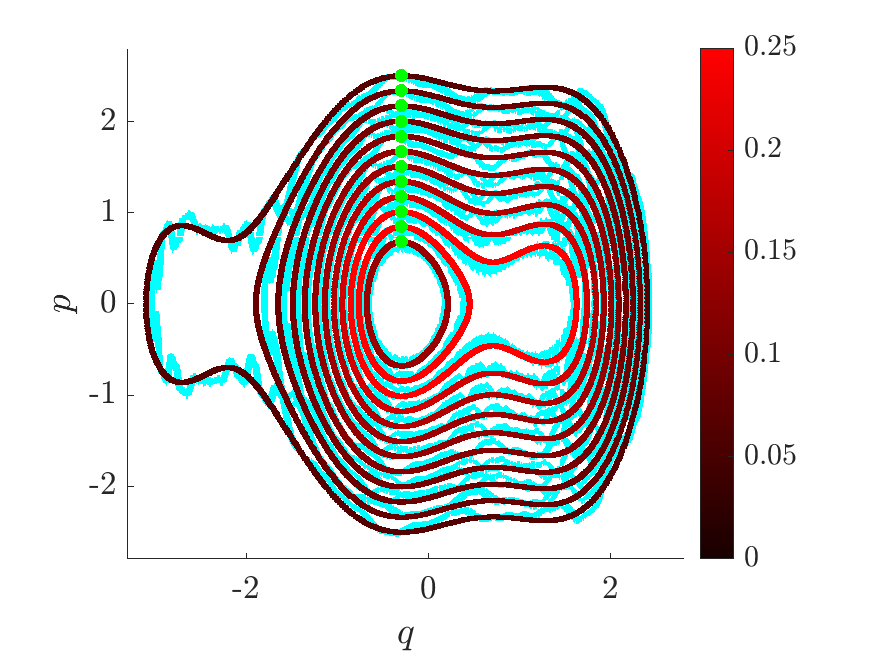} \\ \hline
\end{tabular}
\end{center}
\caption{Error statistics for Example 1. Each trajectory in red is simulated from the learned Hamiltonian system $\Hredhat$ and is colored according to the error metric on the left (defined in \eqref{tpr}-\eqref{delZstar}). The training data is plotted in cyan. Green dots indicate initial conditions. Left: results for $\vep= 0.01$, right: results for $\vep=0.05$.}
\label{fig:exp1_st}
\end{figure}

\begin{figure}
\begin{center}
\begin{tabular}{m{0.08\textwidth}@{}m{0.35\textwidth}@{}m{0.35\textwidth}}
& \centering\fbox{$Q(0)=\frac{\pi}{2}$}  
& \begin{center}\fbox{$Q(0)=\frac{31\pi}{32}$}\end{center} \\
\fbox{$\Delta H$}  &		\includegraphics[trim={30 0 20 10},clip,width=0.32\textwidth]{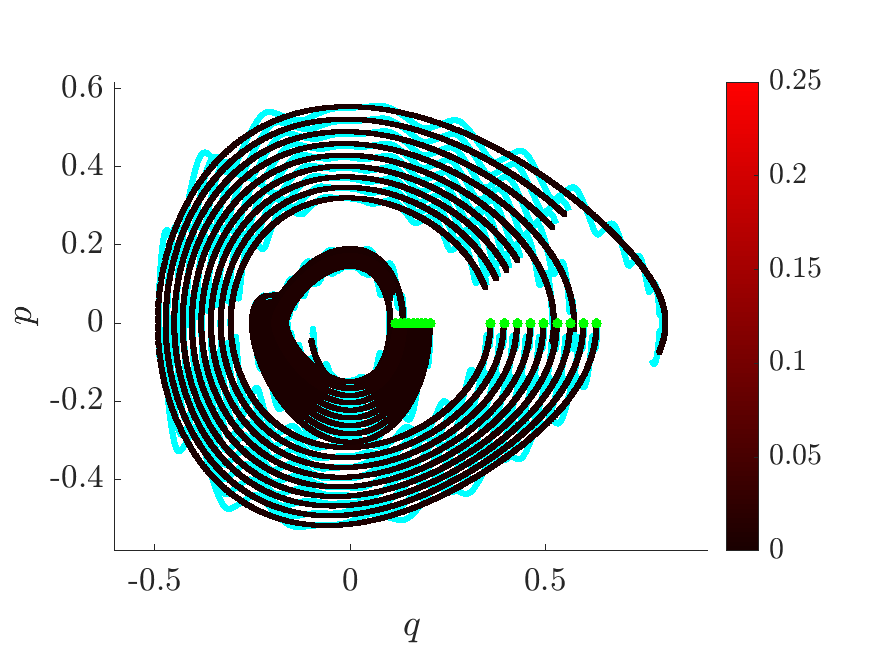} &
		\includegraphics[trim={30 0 20 10},clip,width=0.32\textwidth]{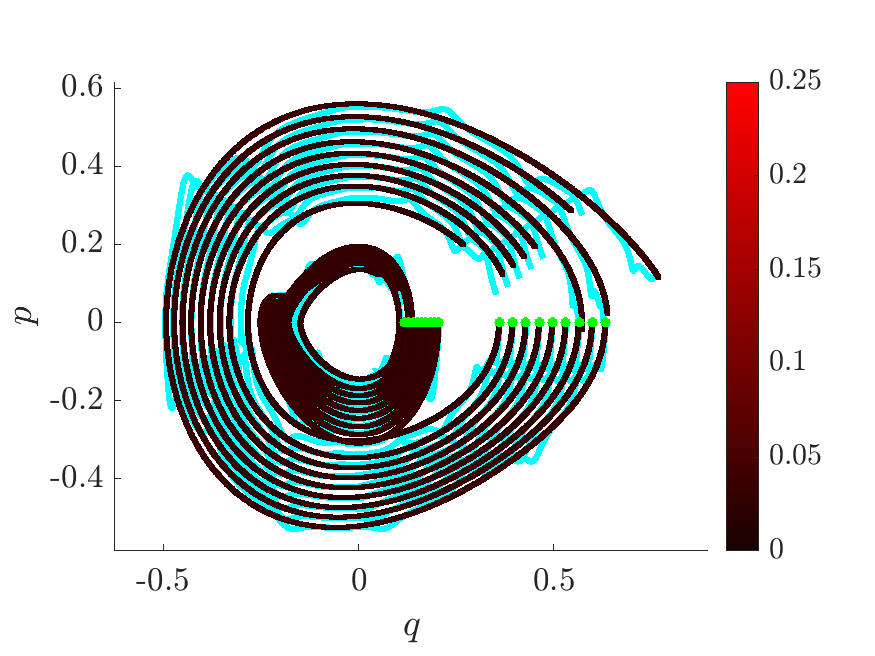} \\
\fbox{$\Delta \Zbf^\mu$}  &	\includegraphics[trim={30 0 20 10},clip,width=0.32\textwidth]{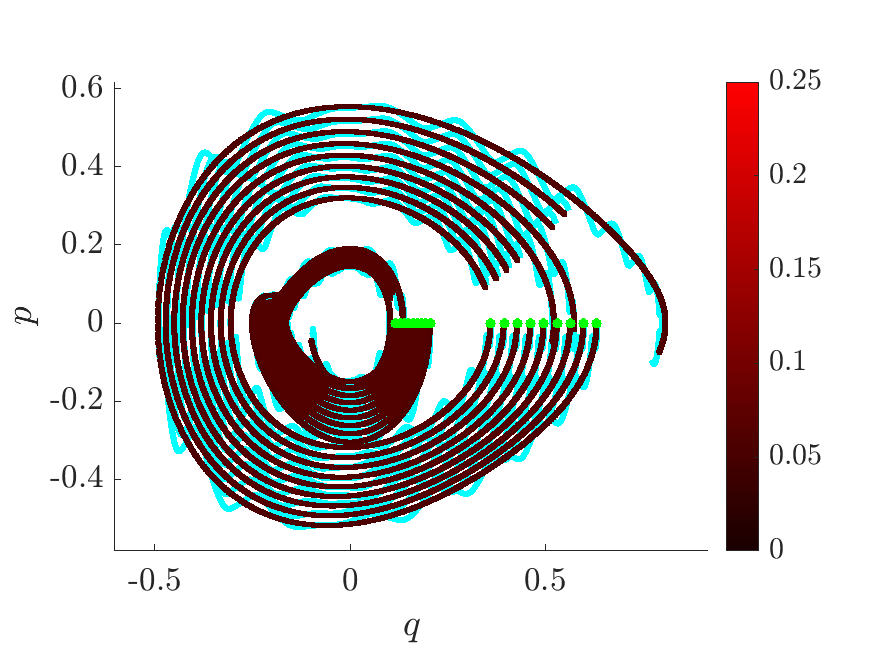} &
		\includegraphics[trim={30 0 20 10},clip,width=0.32\textwidth]{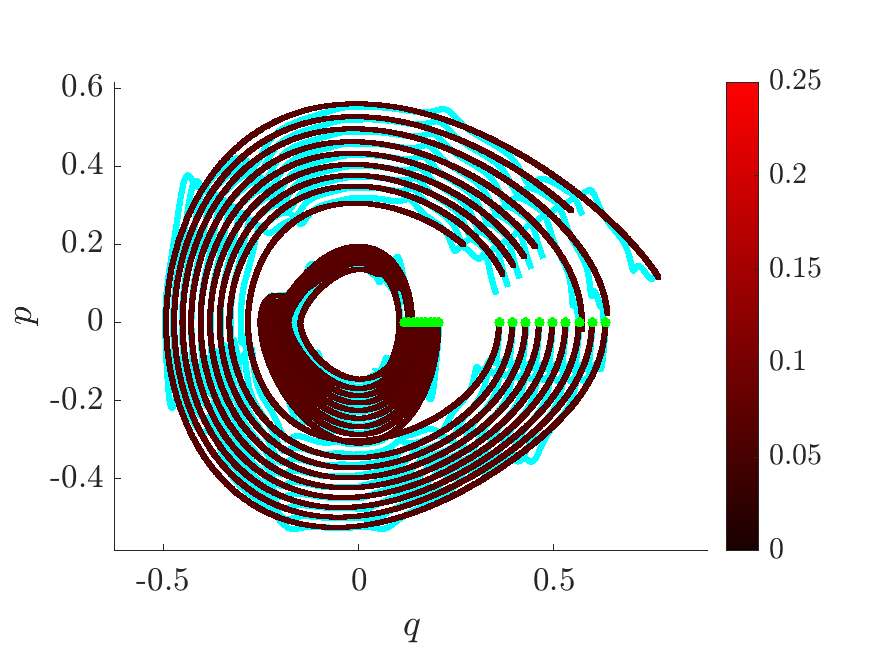} \\
\fbox{$\Delta \Zbf$}  &		\includegraphics[trim={30 0 20 10},clip,width=0.32\textwidth]{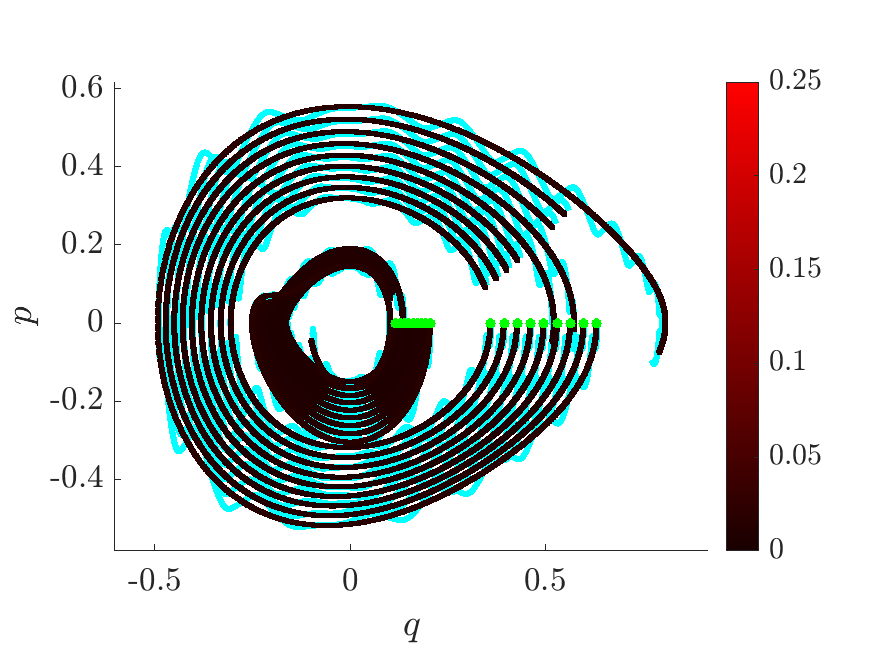} &
		\includegraphics[trim={30 0 20 10},clip,width=0.32\textwidth]{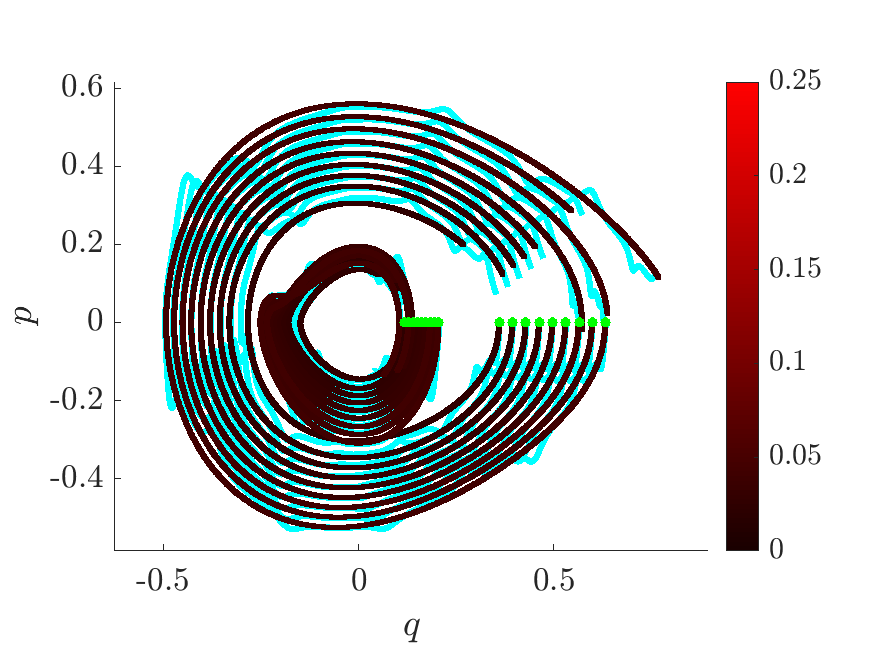} \\ \hline
	\fbox{$\Delta \Zbf^\star$} &	\includegraphics[trim={30 0 20 10},clip,width=0.32\textwidth]{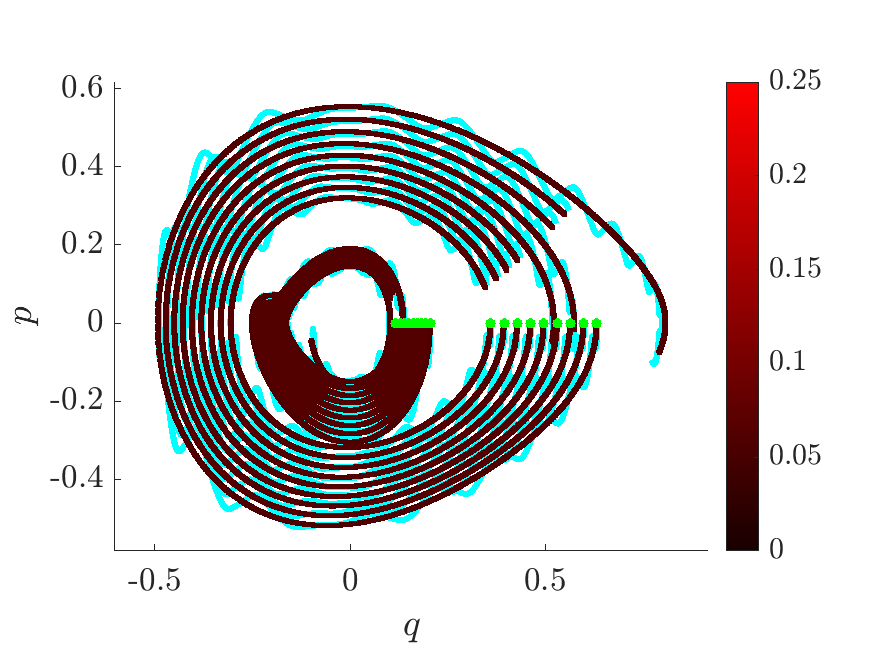} &
		\includegraphics[trim={30 0 20 10},clip,width=0.32\textwidth]{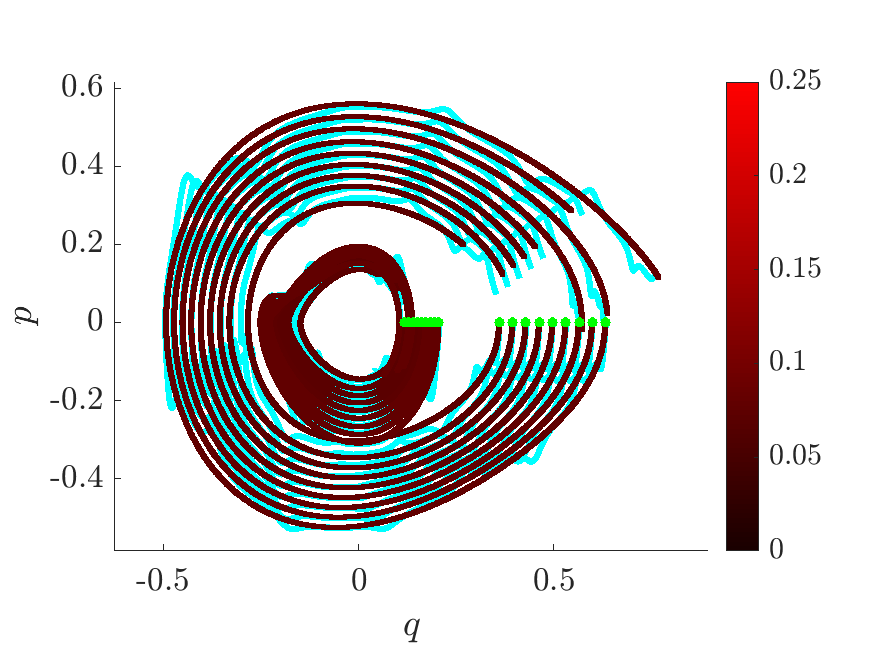} \\ \hline
\end{tabular}
\end{center}
\caption{Error statistics for Example 2 at $\vep=0.05$. Each trajectory in red is simulated from the learned Hamiltonian system $\Hredhat$ and is colored according to the error metric on the left (defined in \eqref{tpr}-\eqref{delZstar}). The training data is plotted in cyan. Green dots indicate initial conditions. Left: results for $Q(0) = \frac{\pi}{2}$, right: results for $Q(0) = \frac{31\pi}{32}$.}
\label{fig:exp2_st}
\end{figure}

\begin{figure}
\begin{center}
\begin{tabular}{m{0.08\textwidth}@{}m{0.35\textwidth}@{}m{0.35\textwidth}}
& \centering\fbox{$\vep=0.01$}  
& \begin{center}\fbox{$\vep=0.03$}\end{center} \\
\fbox{$\Delta H$}  &		\includegraphics[trim={30 0 30 10},clip,width=0.32\textwidth]{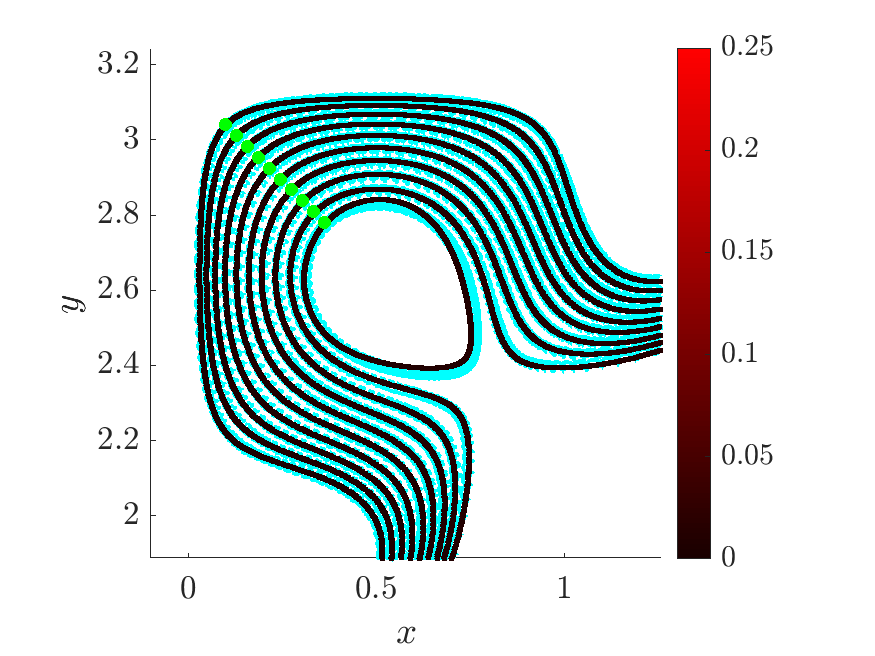} &
		\includegraphics[trim={30 0 30 10},clip,width=0.32\textwidth]{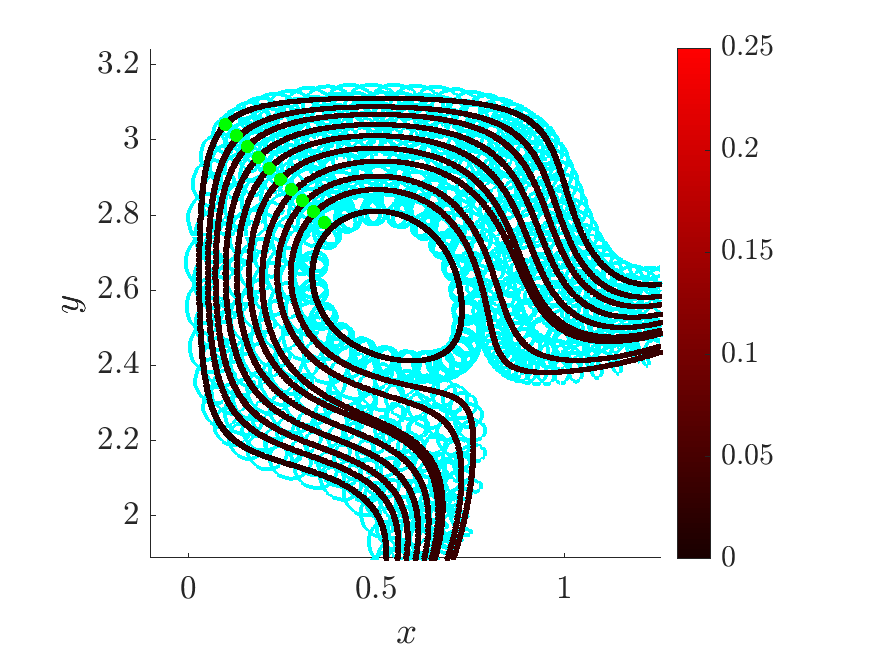} \\
\fbox{$\Delta \Zbf^\mu$}  &	\includegraphics[trim={30 0 30 10},clip,width=0.32\textwidth]{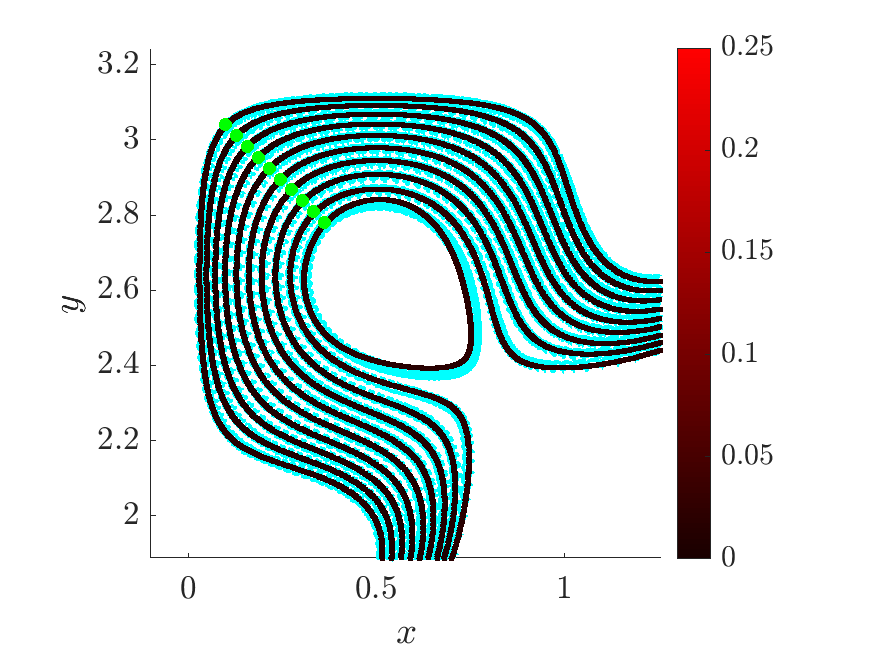} &
		\includegraphics[trim={30 0 30 10},clip,width=0.32\textwidth]{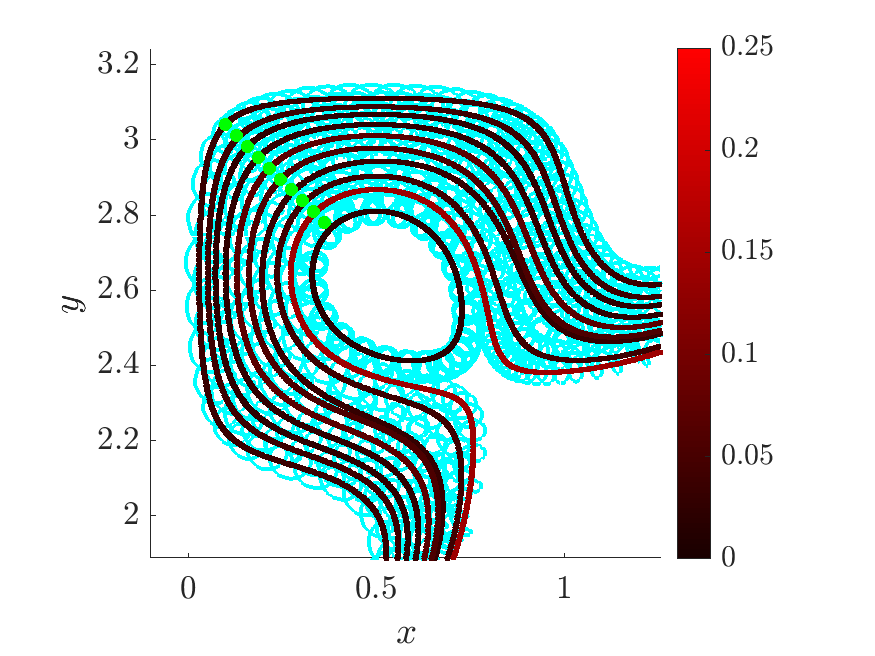} \\
\fbox{$\Delta \Zbf$}  &		\includegraphics[trim={30 0 30 10},clip,width=0.32\textwidth]{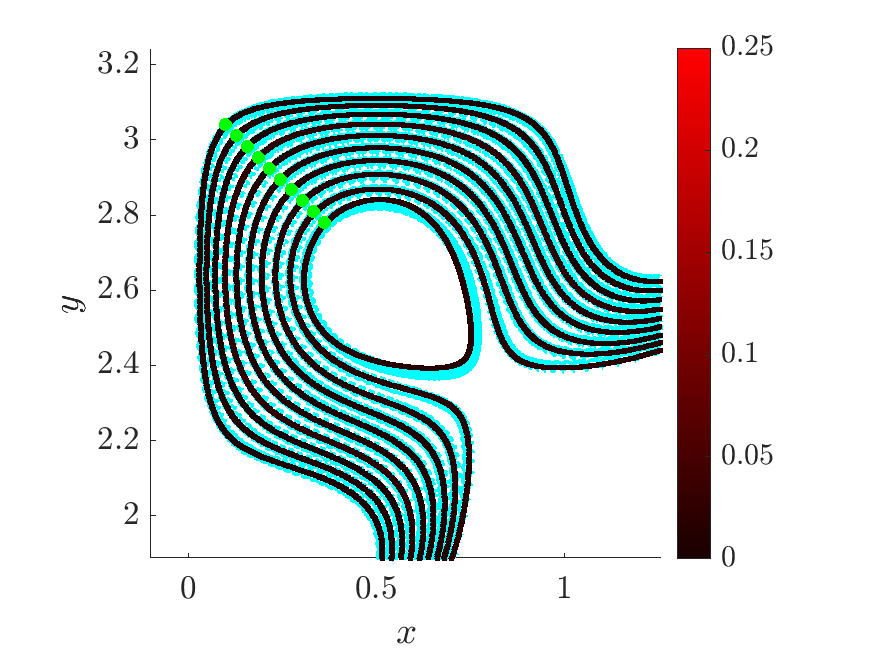} &
		\includegraphics[trim={30 0 30 10},clip,width=0.32\textwidth]{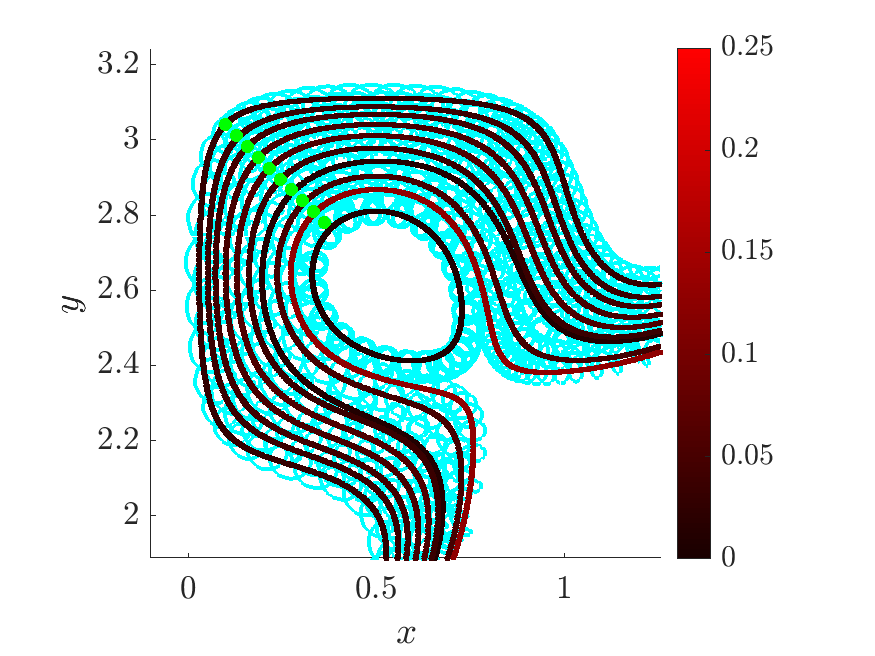} \\ \hline
\fbox{$\Delta \Zbf^\star$}  &	\includegraphics[trim={30 0 30 10},clip,width=0.32\textwidth]{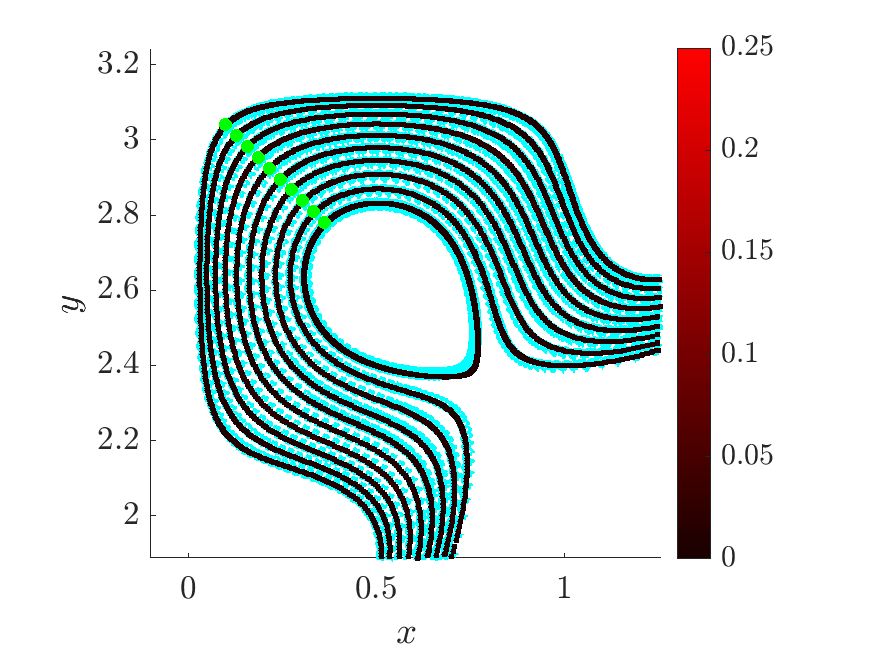} &
		\includegraphics[trim={30 0 30 10},clip,width=0.32\textwidth]{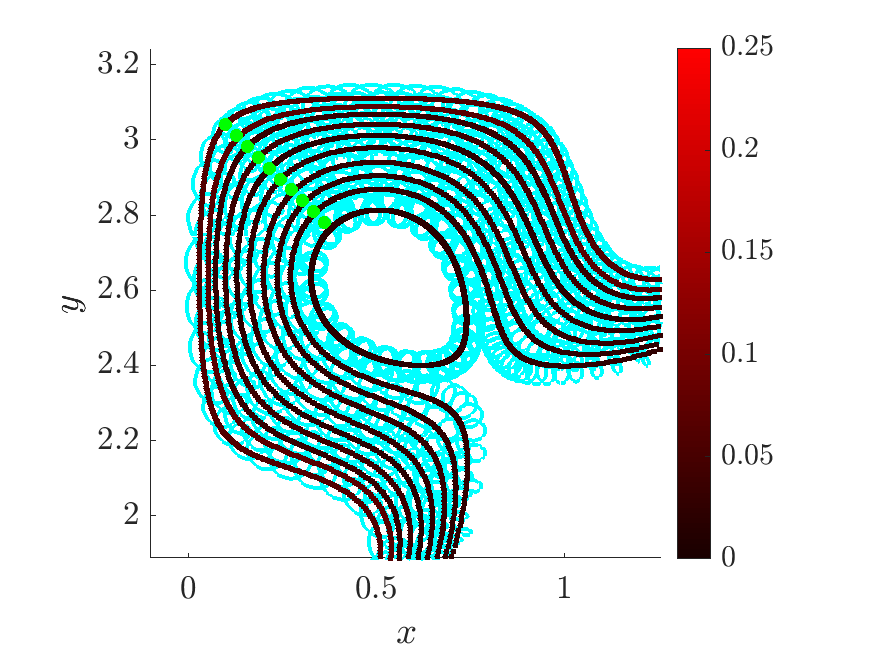} \\ \hline
\end{tabular}
\end{center}
\caption{Error statistics for Example 3. Each trajectory in red is simulated from the learned Hamiltonian system $\Hredhat$ and is colored according to the error metric on the left (defined in \eqref{tpr}-\eqref{delZstar}). The training data is plotted in cyan. Green dots indicate initial conditions. Left: results for $\vep= 0.01$, right: results for $\vep=0.03$.}
\label{fig:exp3_st}
\end{figure}

\begin{figure}
\begin{center}
\begin{tabular}{m{0.08\textwidth}@{}m{0.35\textwidth}@{}m{0.35\textwidth}}
& \centering\fbox{$\vep=0.01$}  
& \begin{center}\fbox{$\vep=0.03$}\end{center} \\
\fbox{$\Delta H$}  &		\includegraphics[trim={5 0 5 10},clip,width=0.32\textwidth]{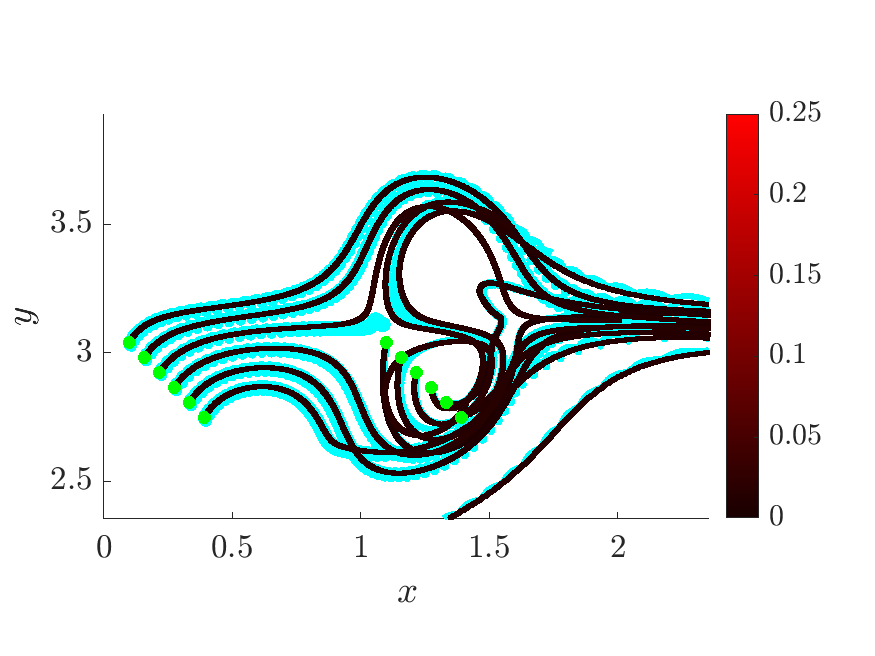} &
		\includegraphics[trim={5 0 5 10},clip,width=0.32\textwidth]{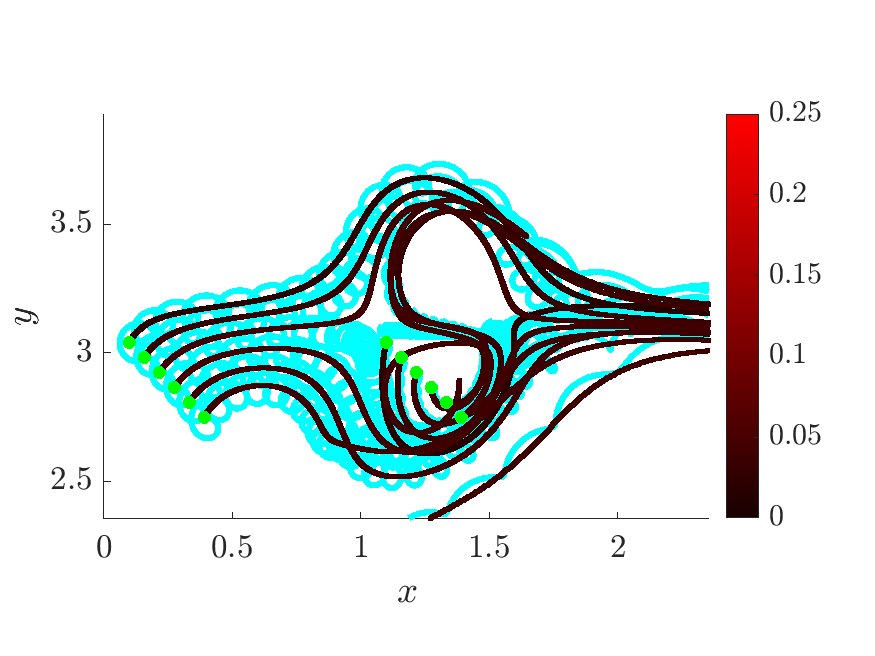} \\
\fbox{$\Delta \Zbf^\mu$}  &	\includegraphics[trim={5 0 5 10},clip,width=0.32\textwidth]{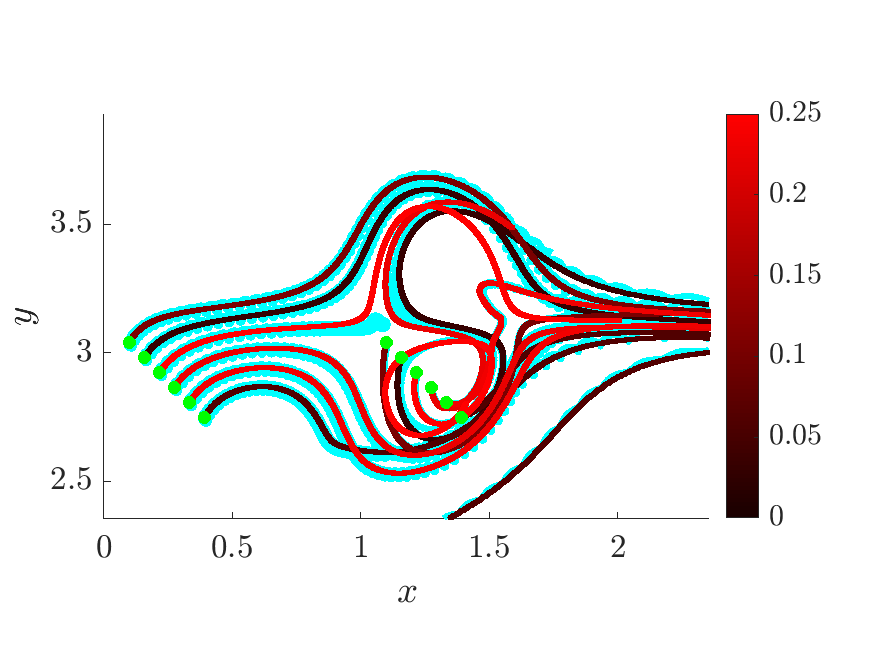} &
		\includegraphics[trim={5 0 5 10},clip,width=0.32\textwidth]{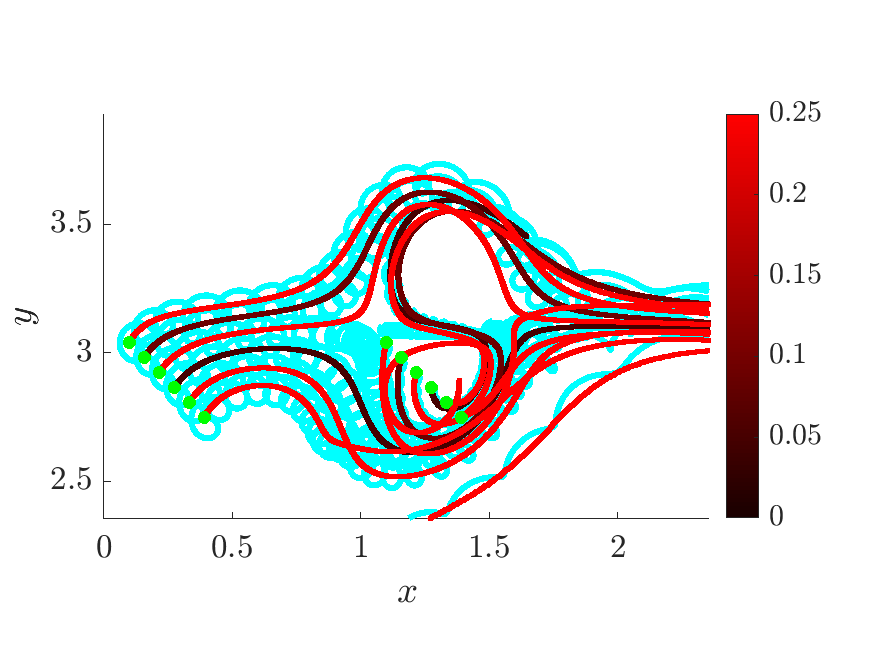} \\
\fbox{$\Delta \Zbf$} &		\includegraphics[trim={5 0 5 10},clip,width=0.32\textwidth]{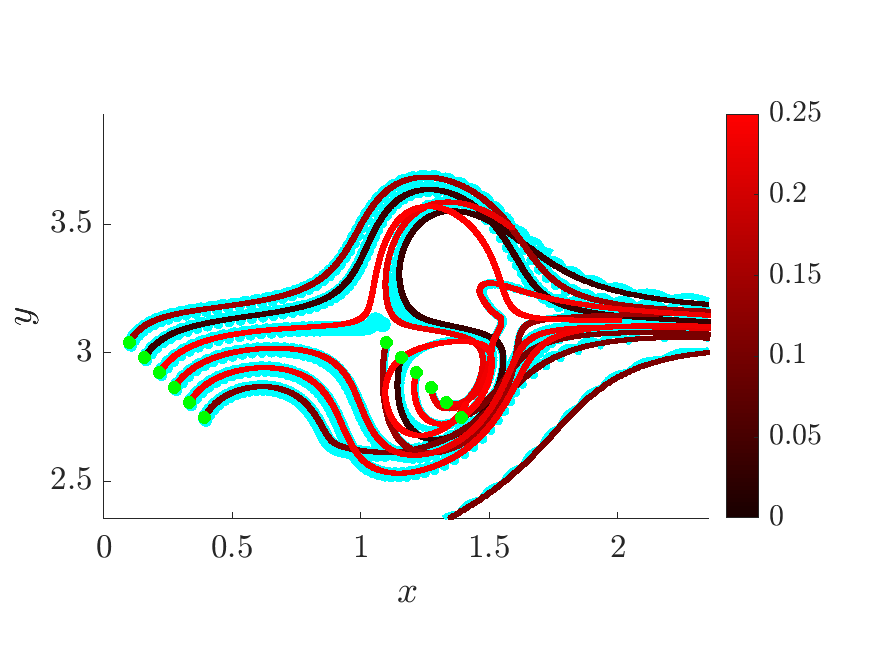} &
		\includegraphics[trim={5 0 5 10},clip,width=0.32\textwidth]{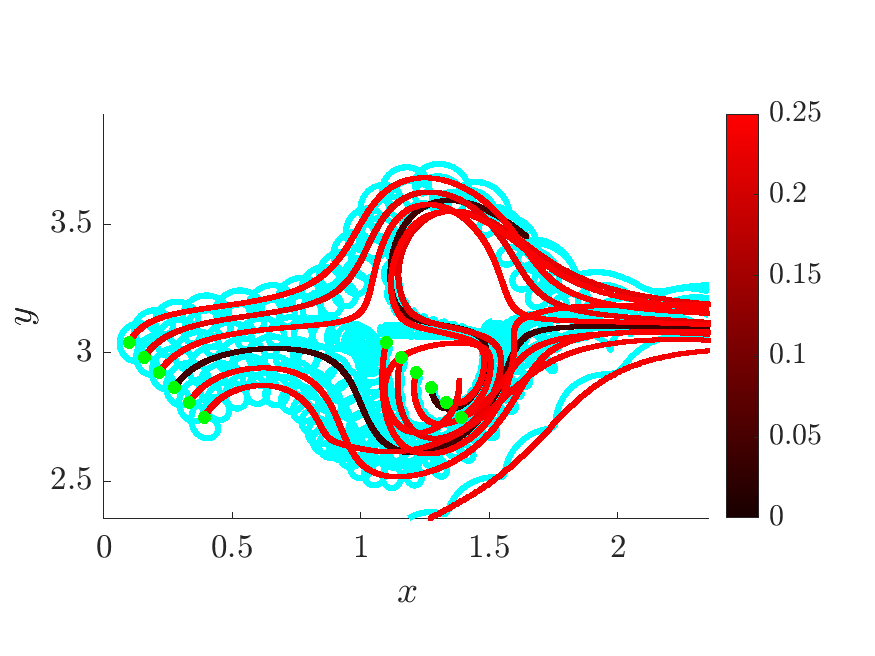} \\ \hline
	\fbox{$\Delta \Zbf^\star$}  &	\includegraphics[trim={5 0 5 10},clip,width=0.32\textwidth]{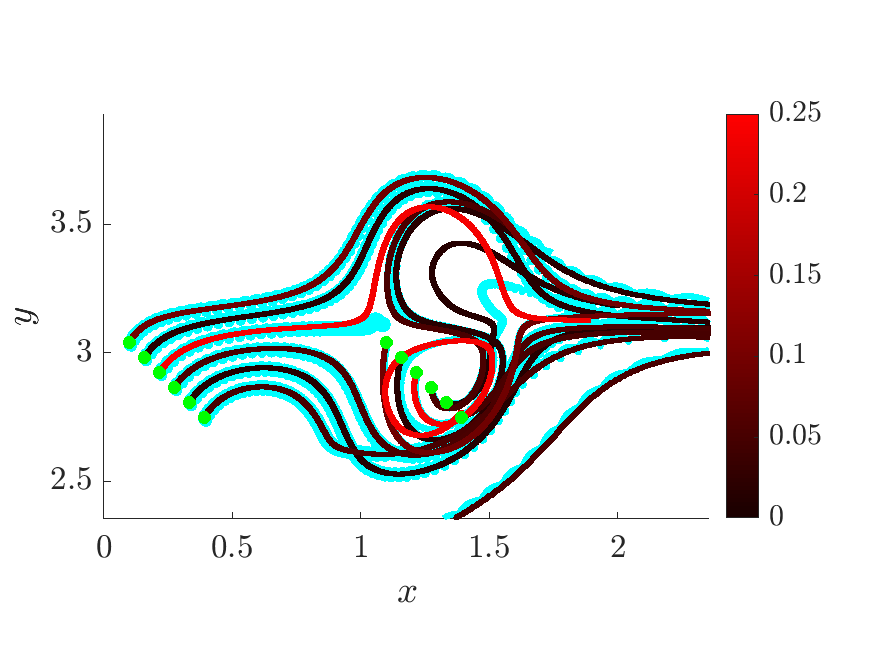} &
 
		\includegraphics[trim={5 0 5 10},clip,width=0.32\textwidth]{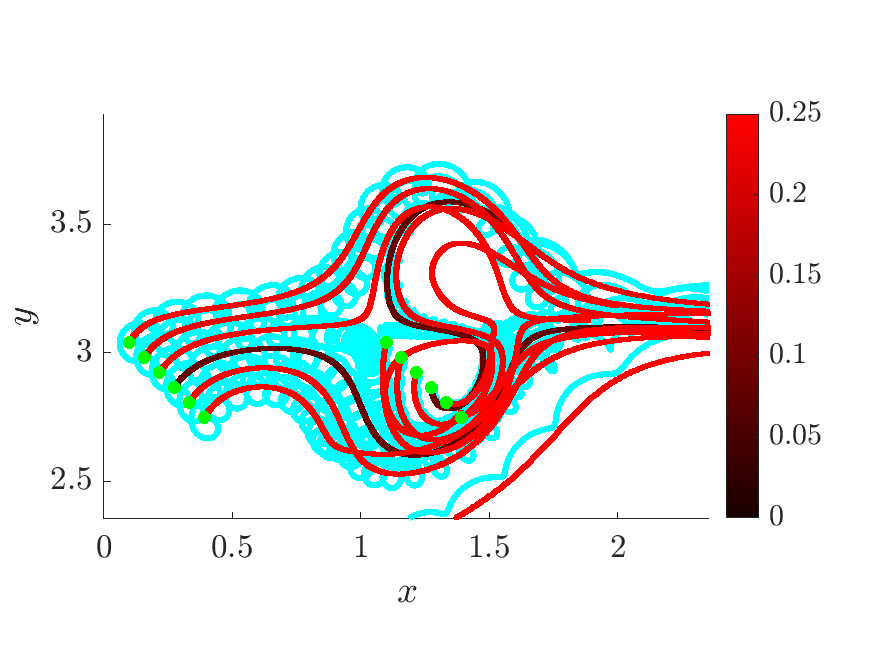} \\ \hline
\end{tabular}
\end{center}
\caption{Error statistics for Example 4. Each trajectory in red is simulated from the learned Hamiltonian system $\Hredhat$ and is colored according to the error metric on the left (defined in \eqref{tpr}-\eqref{delZstar}). The training data is plotted in cyan. Green dots indicate initial conditions. Left: results for $\vep= 0.01$, right: results for $\vep=0.03$.}
\label{fig:exp4_st}
\end{figure}

\subsection{Results: nonzero extrinsic noise ($\sigma_{NR}>0$)}

We now survey the performance of WSINDy in recovering $\Hred$ from data with extrinsic (measurement) noise, with quantitative results presented in Figure \ref{fig:noise_stats}. For each example we take a fixed trajectory from the extreme perturbative regime (defined above) for which WSINDy recovers the correct model without noise (TPR=1), and we add various levels of Gaussian white noise. We apply WSINDy to the noisy data $\Zbf$ and average results over 100 independent trials. Example trajectories with $10\%$ noise are included in Figures \ref{fig:noise_1}-\ref{fig:noise_4}, with the noisy data $\Zbf$ in red, the clean data $\Zbf^\star$ in black, the true reduced data $\Zbf^\mu$ in blue, and the learned reduced data $\widehat{\Zbf}^\mu$ in green. It should be noted that combined extrinsic noise and extreme intrinsic perturbations presents a significant challenge, and the results in Figure \ref{fig:noise_stats} improve greatly under less severe corruption levels.

From Figure \ref{fig:noise_stats} we observe similar trends for each example in the low-medium noise regime. For each example the method is robust up until at least $\sigma_{NR} = 10^{-1.5}\approx 0.032$, or $\approx 3\%$ noise providing on average $\text{TPR}>0.95$ and $\Delta H$ well below $10\%$. For larger noise levels, results vary by performance metric and example. The remainder of this section discusses these variations for the $10\%$ noise case, examples of which are given in Figures \ref{fig:noise_1}-\ref{fig:noise_4}.  

At $10\%$ noise ($\sigma_{NR} = 0.1$), WSINDy is still able to recover the model structure of Examples 1 and 3 very well, with TPR$\in [0.9,1]$, yet for Example 3 this coincides with large errors $\Delta H$ and $\Zbf^\mu$. We can see from the dynamics in Figure \ref{fig:noise_3} that although the model is identified correctly, with reasonable qualitative agreement between the green and blue curves, the large value of $\Delta \Zbf^\mu$ is explained by the warped contour followed by the green curve. This is an artifact of errors in the coefficients $\what$ compared to $\wstar$ due to noise. On the other hand, Example 1 performs well in all three categories TPR, $\Delta H$, $\Delta \Zbf^\mu$, but has significantly larger errors compared to $\Zbf^\star$, as indicated by $\Delta \Zbf$. This is due to accumulation of errors from a slight phase differences between $\Zbf^\mu$ and $\widehat{\Zbf}^\mu$ (see Figure \ref{fig:noise_1}, right). Correcting for this effect will be essential for utilizing WSINDy in related forward simulations. Overall, the fact that WSINDy still identifies the correct model, given combined extrinsic noise and intrinsic fast-scale dynamics, suggests that the method is well-suited for scientific discovery with these mixed effects. 

At $10\%$ noise Example 2 exhibits the lowerest TPR, but this coincides with excellent agreement with $\Hred$ as measured by $\Delta H$ and $\Delta \Zbf^\mu$, both with average values less than $4\%$. Hence, in this case large noise leads to identification of a model with slightly different terms, but still with acceptable accuracy compared to $\Hred$. Figure \ref{fig:noise_stats} (bottom right, orange curve) shows that on average the WSINDy model performs nearly as well as the analytical reduced dynamics given by $\Hred$. Since $\Hred$ is only a leading-order approximation, a natural next line of inquiry would be to examine the connection between ``misspecified'' models and next-order corrections. 

For Example 4, 10\% noise is clearly outside of the feasible recovery regime, with misspecified models (TPR<1) leading to large values of $\Delta H$ and significant forward simulations errors $\Delta \Zbf^\mu$ combined with chaotic effects. This can be overcome by considering multiple trajectories (not shown here), which serves to recover $\text{TPR}=1$, yet accuracy issues with $\Delta H$ remain. We conjecture that some form of variance reduction is needed. We aim to investigate the applicability of WENDy (Weak-form estimation of nonlinear dynamics) \cite{BortzMessengerDukic2023BullMathBiol} which has been demonstrated to reduce regression errors due to extrinsic noise.

\begin{figure}
\begin{center}
\begin{tabular}{@{}c@{}c@{}}
$\boxed{\text{TPR}}$ & $\boxed{\Delta H}$\\
\includegraphics[trim={0 0 0 0},clip,width=0.42\textwidth]{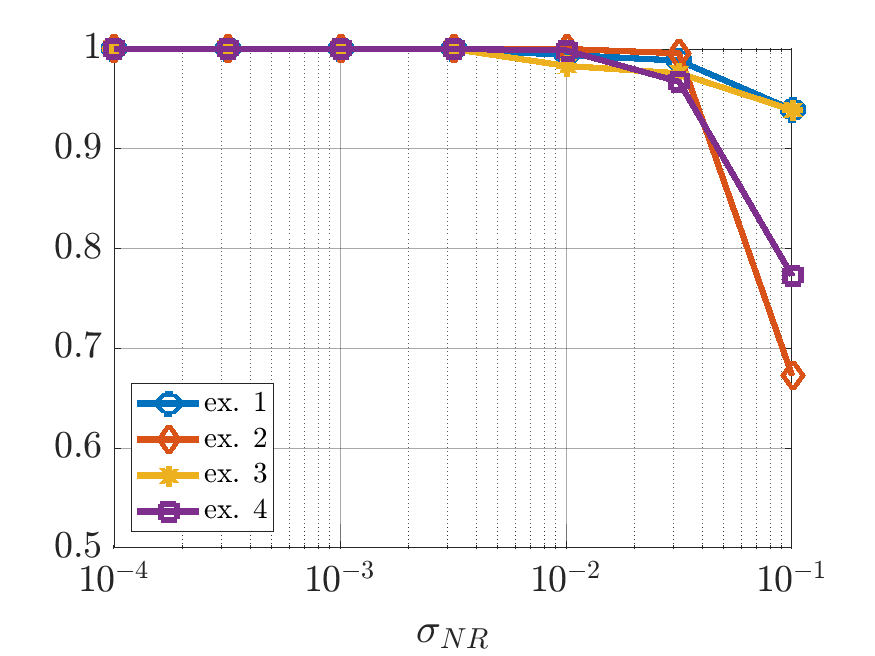} &
		\includegraphics[trim={0 0 0 0},clip,width=0.42\textwidth]{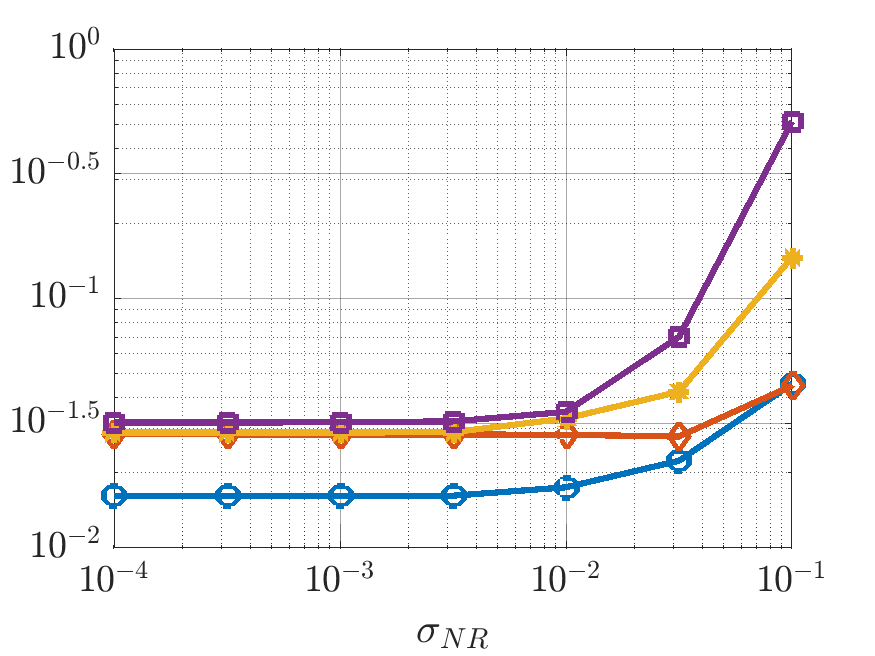} \\
$\boxed{\Delta \Zbf^\mu}$ & $\boxed{\Delta \Zbf}$\\
\includegraphics[trim={0 0 0 0},clip,width=0.42\textwidth]{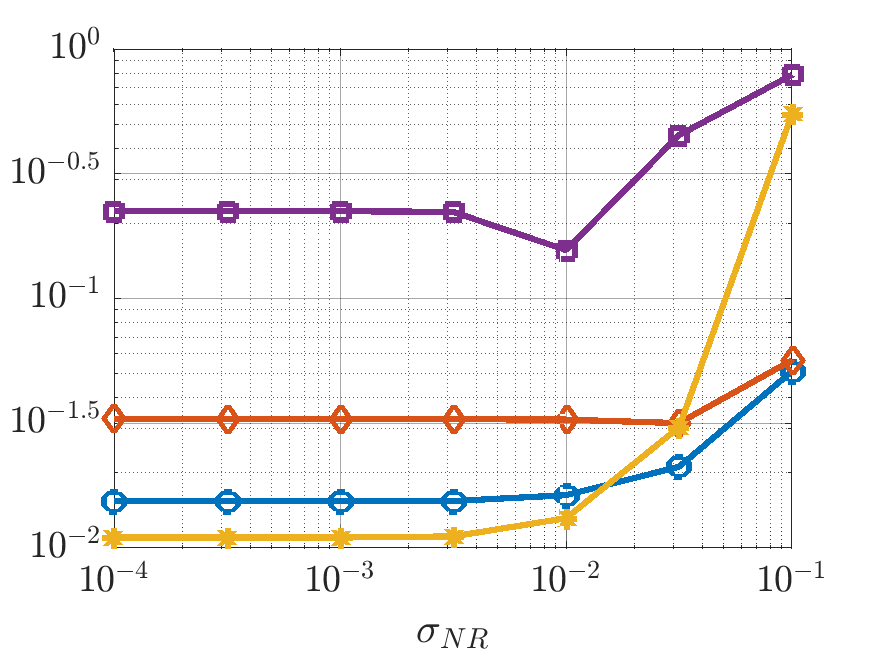} &
		\includegraphics[trim={0 0 0 0},clip,width=0.42\textwidth]{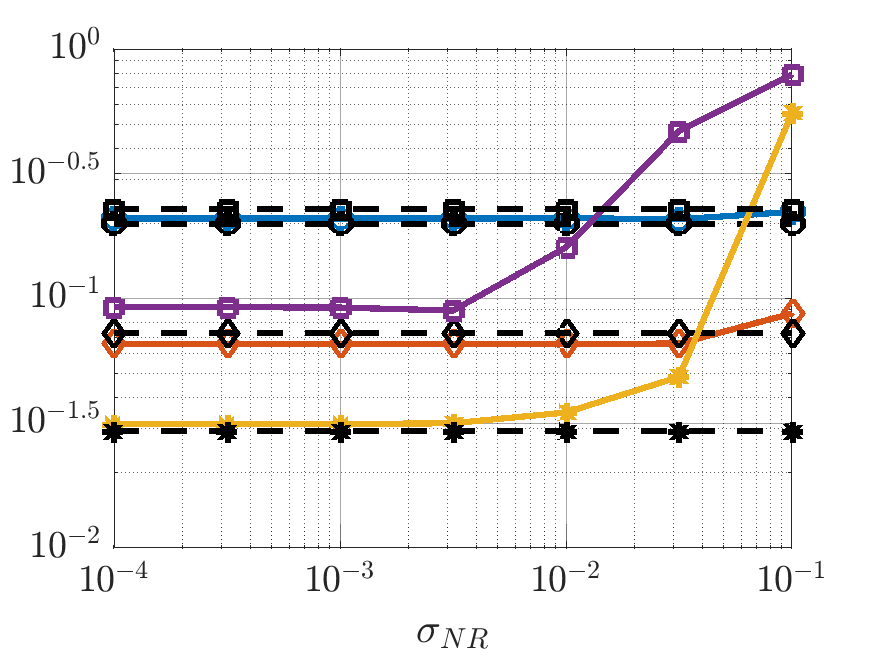} \\ \hline
\end{tabular}
\end{center}
\caption{Results for WSINDy applied to noisy data $\Zbf = \Zbf^\star+\eta$.}
\label{fig:noise_stats}
\end{figure}

\begin{figure}
\begin{center}
\begin{tabular}{@{}c@{}c@{}}
\includegraphics[trim={0 0 0 0},clip,width=0.5\textwidth]{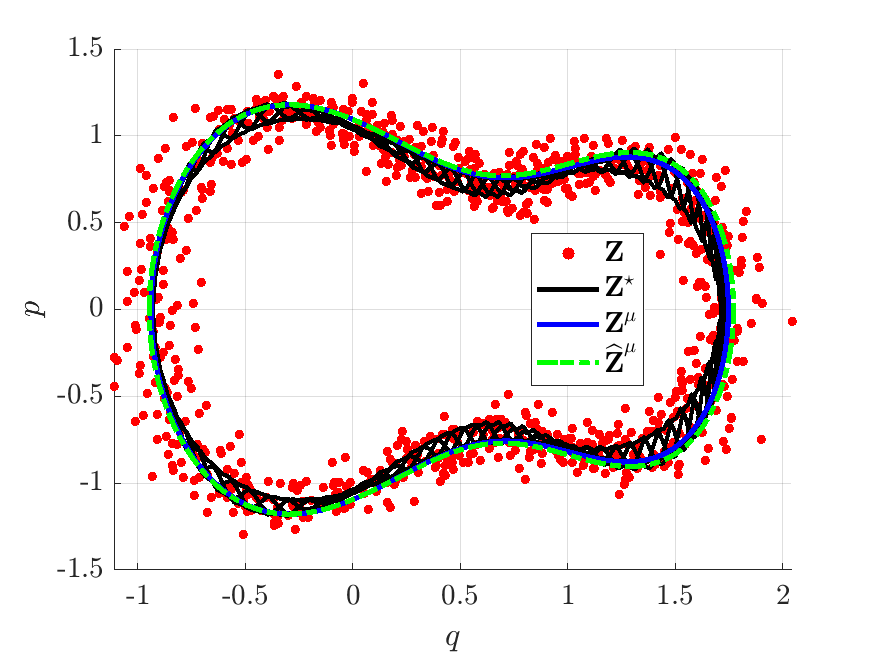} &
		\includegraphics[trim={0 0 0 0},clip,width=0.5\textwidth]{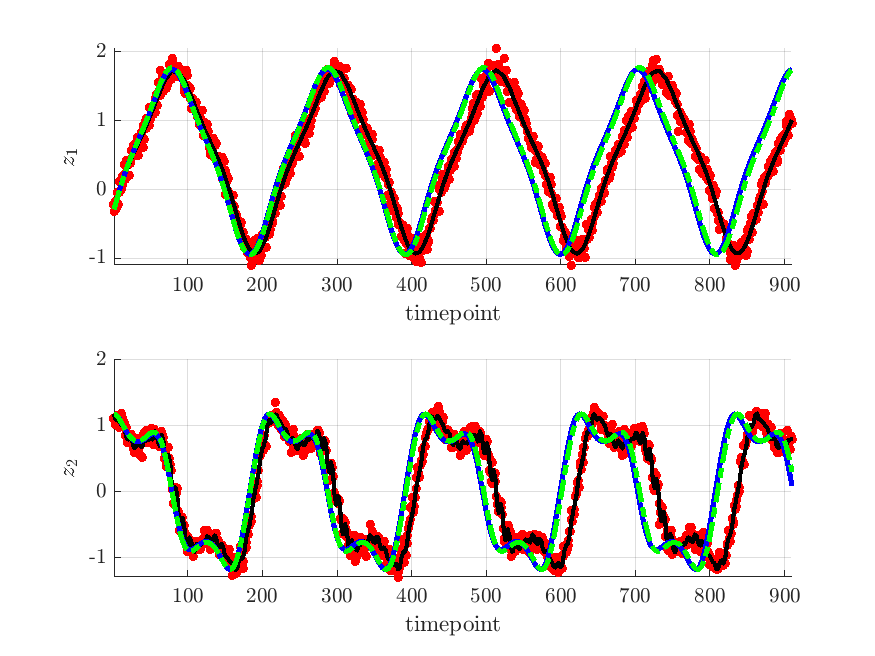} \\ \hline
\end{tabular}
\end{center}
\caption{Visualization of noisy data (in red) and recovered model (in green) for Example 1, as well as the clean data $\Zbf^\star$ and simulation from the true reduced system $\Hred$.}
\label{fig:noise_1}
\end{figure}

\begin{figure}
\begin{center}
\begin{tabular}{@{}c@{}c@{}}
\includegraphics[trim={0 0 0 0},clip,width=0.5\textwidth]{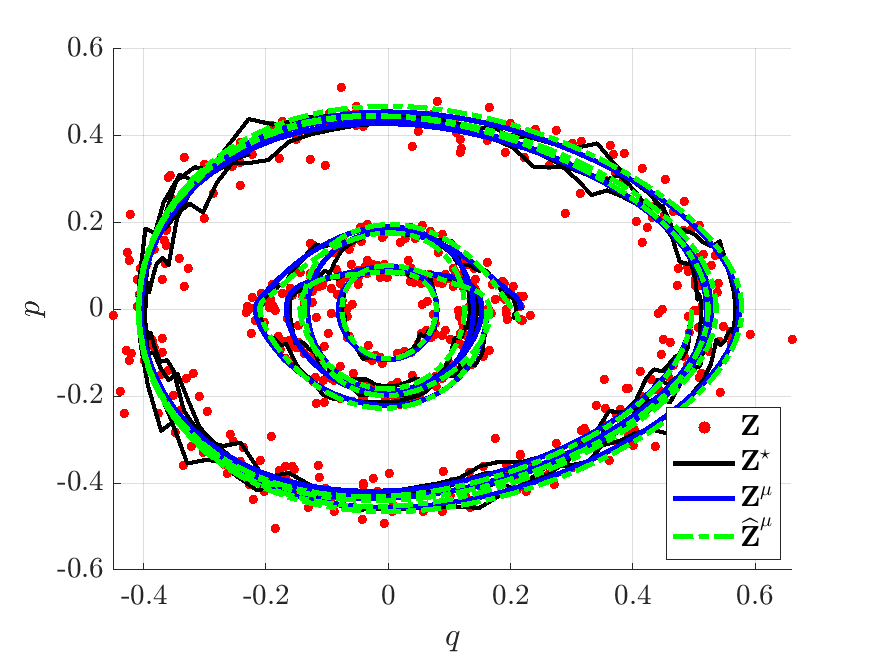} &
		\includegraphics[trim={0 0 0 0},clip,width=0.5\textwidth]{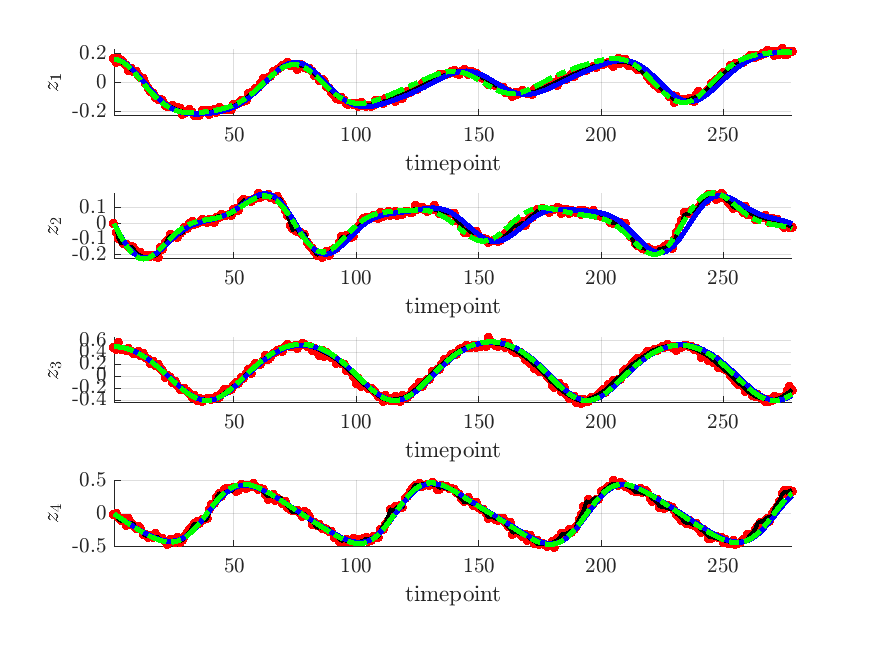} \\ \hline
\end{tabular}
\end{center}
\caption{Visualization of noisy data (in red) and recovered model (in green) for Example 2, as well as the clean data $\Zbf^\star$ and simulation from the true reduced system $\Hred$.}
\label{fig:noise_2}
\end{figure}

\begin{figure}
\begin{center}
\begin{tabular}{@{}c@{}c@{}}
\includegraphics[trim={0 0 0 0},clip,width=0.5\textwidth]{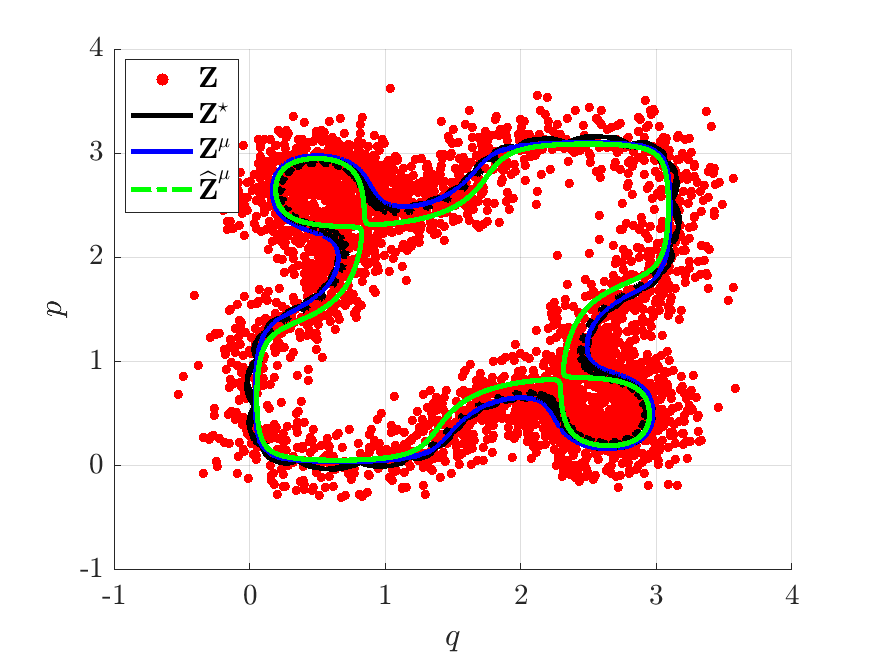} &
		\includegraphics[trim={0 0 0 0},clip,width=0.5\textwidth]{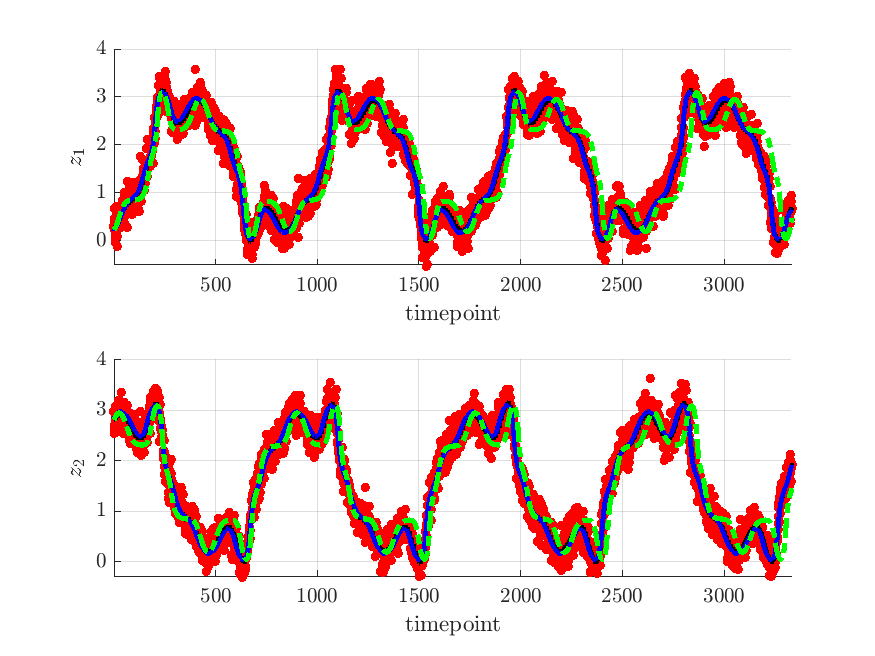} \\ \hline
\end{tabular}
\end{center}
\caption{Visualization of noisy data (in red) and recovered model (in green) for Example 3, as well as the clean data $\Zbf^\star$ and simulation from the true reduced system $\Hred$.}
\label{fig:noise_3}
\end{figure}

\begin{figure}
\begin{center}
\begin{tabular}{@{}c@{}c@{}}
\includegraphics[trim={0 0 0 0},clip,width=0.5\textwidth]{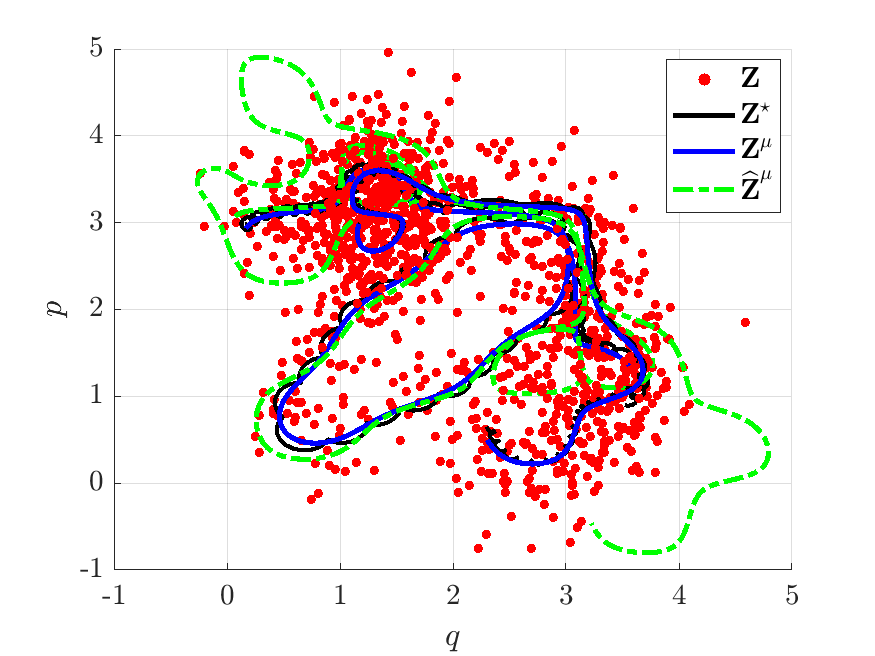} &
		\includegraphics[trim={0 0 0 0},clip,width=0.5\textwidth]{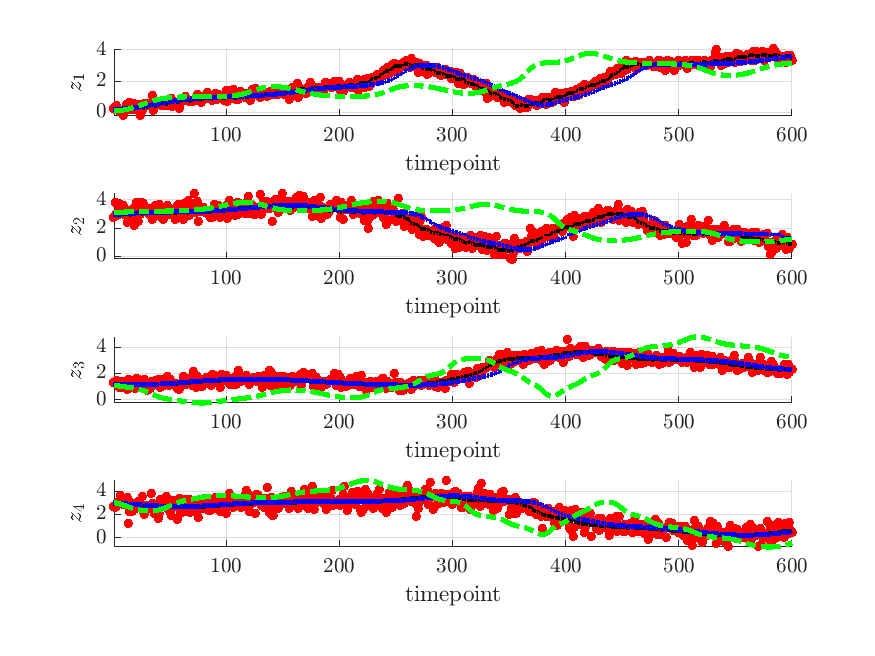} \\ \hline
\end{tabular}
\end{center}
\caption{Visualization of noisy data (in red) and recovered model (in green) for Example 4, as well as the clean data $\Zbf^\star$ and simulation from the true reduced system $\Hred$.}
\label{fig:noise_4}
\end{figure}

\section{Conclusions}

In this article we have described a weak-form equation learning approach to coarse-graining Hamiltonian systems using the WSINDy algorithm. We have provided substantial evidence that appropriate coarse-grained models with Hamiltonian structure can be identified from noisy data simply by choosing a proper weak formulation. This is significant due to the nontrivial analytical procedures involved in deriving reduced Hamiltonian systems using nearly-periodic reduction techniques. The dictionary learning approach naturally enforces structure preservation, and the method is highly efficient, requiring no forward solves of candidate models. The output is a human-readable equation for the reduced-order Hamiltonian which can then be used in down-stream tasks. 

The fact that a suitable weak formulation enables one to identify the reduced-order Hamiltonian over all of phase space using only a single (noisy) trajectory is surprising, and to the best of the authors' knowledge not found in the literature. An obvious next direction is to leverage this in combination with black-box methods (neural networks, Gaussian processes, etc.), and to adapt the test functions to target different models in a hierarchical fashion, as demonstrated in Figures \ref{fig:full_sims} and \ref{fig:full_results}. We aim also in a future work to quantify the sufficiency of information needed to identify the reduced Hamiltonian, as results from degenerate cases (trajectories that only encircle one of the relevant fixed points, see Examples 1 in Figure \ref{fig:tpr}) imply that the reduced Hamiltonian is accessible from regions of phase with seemingly low levels of information. At the other extreme, combining information from multiple trajectories (as typically done in black-box learning approaches) is an obvious means of increasing performance. Answering questions of data and information sufficiency will be crucial to developing methods of identifying higher-order (in $\vep$) coarse-grained models. 

Finally, combined with findings in \cite{MessengerBortz2022PhysicaD}, a major implication of our results is that the weak form opens the door to a new generation of coarse-graining algorithms. We aim to explore the use of these algorithms in surrogate modeling to reduce the number of full simulations of dynamics involving disparate scales, which are computationally prohibitive.

\section*{Data availability}

The datasets and code used in the present work are available from the corresponding author (DAM) upon reasonable request.


\bibliography{mathbioCU.bib,wsindyhamiltonian.bib}



\section*{Acknowledgements (not compulsory)}

This research was supported by a US DOE Mathematical Multifaceted Integrated Capability Center
(MMICC) grant to Los Alamos National Laboratory (supporting JWB) and to the University of Colorado (DE-SC0023346 to DMB). This work also utilized the Blanca condo computing resource at the University of Colorado Boulder. Blanca is jointly funded by computing users and the University of Colorado Boulder.
\section*{Author contributions statement}


D.A.M.: conceptualization, formal analysis, investigation, methodology, software, validation, visualization, writing—original draft, writing—review and editing; J.W.B.: conceptualization, formal analysis, writing—review and editing; D.M.B.: conceptualization, funding acquisition, methodology, project administration, resources, supervision, validation, writing—review and editing. All authors gave final approval for publication and agreed to be held accountable for the work performed therein.

\section*{Additional information}

\textbf{Competing interests} The authors attest to no competing interests. 


\section*{Supplemental Information}
\appendix

\section{Robust cornerpoint method}\label{app:cornerpt}

In \cite{MessengerBortz2021JComputPhys}, the authors developed a method of choosing the test function hyperparameters using a cornerpoint algorithm applied to the Fourier spectrum of the data. This approach was taylored to the case of extrinsic noise, in which case the cumulative sum of the power spectrum of the data can be approximated be two line segments, the intersection of which (the cornerpoint) closely marks the changepoint from signal-dominated to noise-dominated Fourier modes. For nearly periodic systems without extrinsic noise, the two line segment assumption is not accurate, however, a feasible test functions support width $T_\phi$ can still be computed using a slightly altered cornerpoint approach, simply by taking the minimum of the rotated cumulative sum of the power spectrum. We also extend the method in \cite{MessengerBortz2021JComputPhys} to work with any test function, not just those approximated by Gaussian distributions.

To this end, denote the discrete Fourier transform (DFT) of the $n$th component of the data $\Zbf$ by 
\[\CalF[\Zbf_n](k) := \sum_{j=0}^{m-1}\exp\left(-\frac{2\pi i}{m}jk\right)\Zbf_n(t_j)\]
 and let its cumulative sum (taken over the negative wavenumbers) be denoted by  
\[\Hbf_n(k) = \sum_{k'=-m/2}^{k} \left\vert\CalF[\Zbf_n](k')\right\vert, \quad k\in \left\{-m/2,\dots,0\right\}.\]
For convenience we will assume $m$ is even as the odd case is nearly identical. With high probability, $\Hbf_n$ will lie below the line 
\[k\to \Ybf_n(k):=\left(\frac{k_{\max}-k}{k_{\max}+(m/2)}\right)\Hbf_n(-m/2)+\left(\frac{k+(m/2)}{k_{\max}+(m/2)}\right)\Hbf_n(k_{\max})\]
where $k_{\max} = \argmax_{k\in \{-m/2,\dots,0\}}|\CalF[\Zbf_n](k)|$ is the wavenumber of the largest Fourier coefficient. A tie is broken by the maximizing $k$ with the largest magnitude. Moreover, $\Hbf_n$ will be approximately convex for $k\in \{-m/2,\dots,k_{\max}\}$. For instance, convexity and $\Hbf_n\leq \Ybf$ over this region are both true if $\CalF[\Zbf_n]$ exhibits any decay of the form $|\CalF[\Zbf_n](k)|\sim |k|^{-s}$ for $s>0$. We define the cornerpoint of $\Hbf_n$ as the point $(k_n^*,\Hbf_n(k_n^*))$ that maximizes the distance between $\Hbf_n(k)$ and $\Ybf_n(k)$ over $k\in\{-m/2,\dots,k_{\max}\}$ in the total least squares sense. A practical computation of this is the following. Let $R(\theta)$ be the rotation matrix in the $(k,\Hbf_n)$ plane such that 
\[\begin{pmatrix} 0 \\ 1 \end{pmatrix}\cdot \Rbf(\theta)\begin{pmatrix} k \\ \Ybf_n(k) \end{pmatrix} = 0.\]
That is, $\Rbf(\theta)$ rotates $\Ybf_n$ parallel to the $k$-axis. Then 
\[k_n^*:= \argmin_{k\in\{-m/2,\dots,k_{\max}\}} \left[\Rbf(\theta)\begin{pmatrix} k \\ \Hbf_n(k) \end{pmatrix}\right]_2\]
where $[\cdot]_2$ indicates the second component of the vector. This process is visualized in Figure \ref{fig:corner}, and in this work, we let the final cornerpoint be the average $k^* = \frac{1}{2N}\sum_{n=1}^{2N}k^*_n$, used for all observed variables. Figure \ref{fig:corner} displays that the previous method is more prone to labeling irrelevant modes as signal-dominated (2nd row), but overall the performance of both methods is comparable. 

\begin{figure}
\begin{center}
\begin{tabular}{@{}c@{}c@{}}
\includegraphics[trim={0 0 0 0},clip,width=0.45\textwidth]{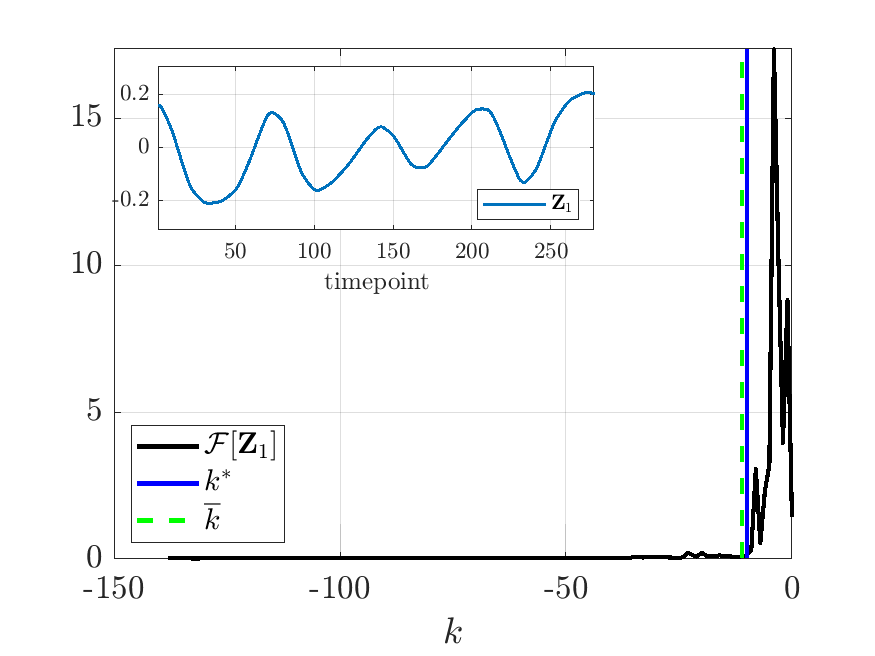} &
		\includegraphics[trim={0 0 0 0},clip,width=0.45\textwidth]{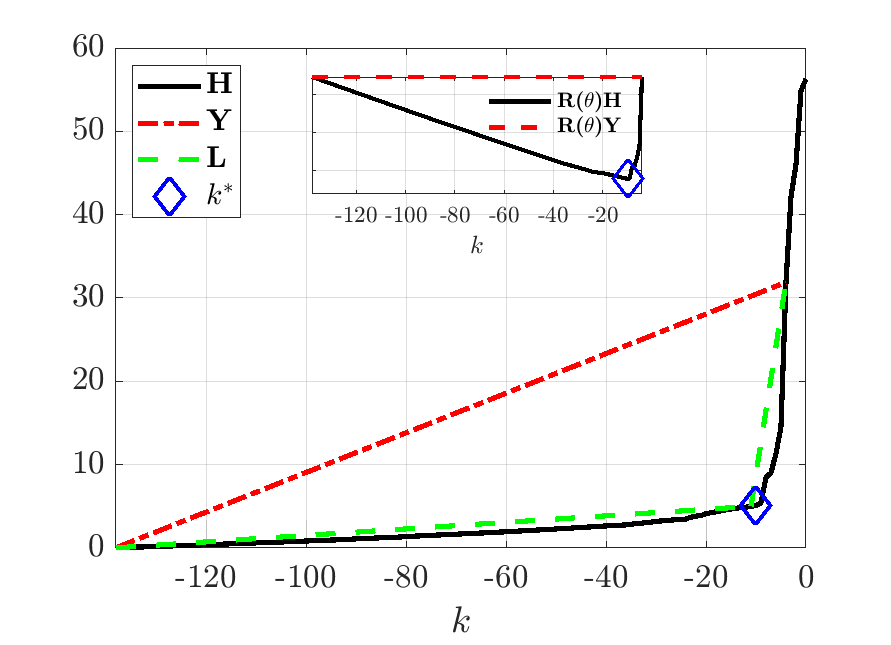} \\
\includegraphics[trim={0 0 0 0},clip,width=0.45\textwidth]{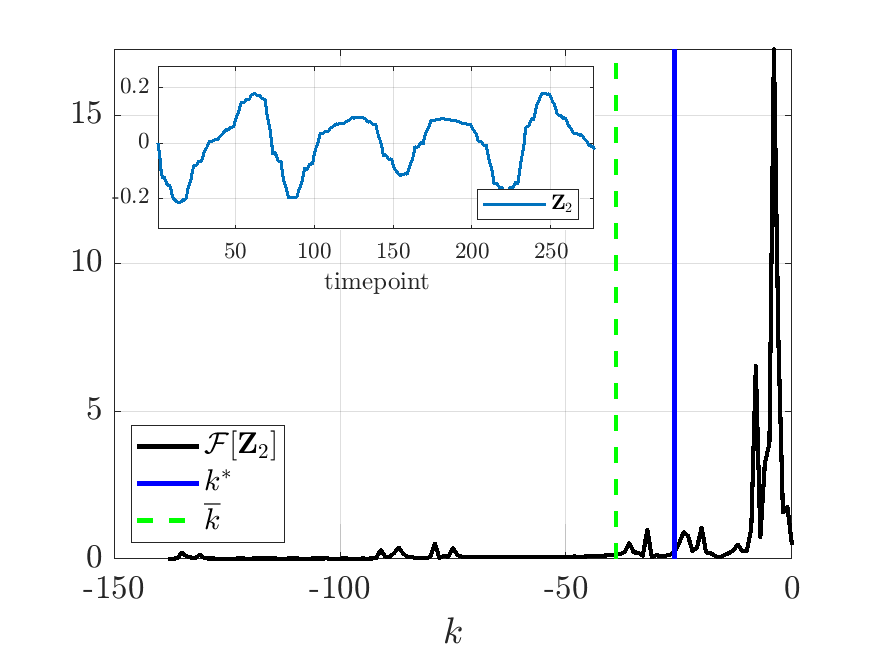} &
		\includegraphics[trim={0 0 0 0},clip,width=0.45\textwidth]{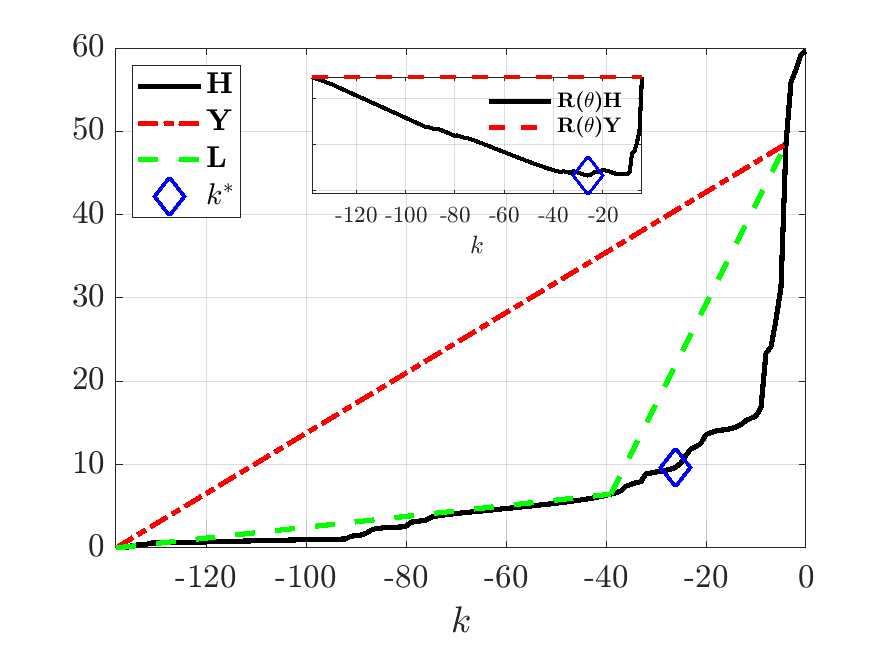}
\end{tabular}
\end{center}
\caption{Visualization of cornerpoint detection and comparison to the method used in \cite{MessengerBortz2021JComputPhys}. The data $\Zbf$ is taken from Example 2 with $\vep=0.05$ and $Q(0)=\frac{31\pi}{32}$. The top and bottom rows display coordinates $1$ and $2$, the configuration and momentum variables for the first oscillator. The data and DFTs are plotted in the left column, along with markers for the detected cornerpoints, while the right column shows the cumulative sums $\Hbf_n$ together with the line $\Ybf$ and the piecewise linear approxmation $\Lbf$ used in \cite{MessengerBortz2021JComputPhys} (inset plots show the rotated $\Hbf$ which is minimized at $k^*$). The value $k^*$ is the cornerpoint detected by the current method, while $\overline{k}$ is the value detected by the method in \cite{MessengerBortz2021JComputPhys}. In the top row, the two methods nearly agree, while the more perturbed dynamics in the bottom row cause the $\overline{k}$ to be significantly larger, leading to more irrelevant modes entering the recovered model.}
\label{fig:corner}
\end{figure}

\section{MSTLS algorithm}\label{app:MSTLS} 

To solve for a sparse coefficient vector $\what$ such that $\Gbf\what \approx \bbf$, we use the MSTLS algorithm proposed in \cite{MessengerBortz2021JComputPhys} that uses the original STLS algorithm from \cite{BruntonProctorKutz2016ProcNatlAcadSci} with  variable thresholding and then performs a line search for the sparsity threshold $\lambda$. By heterogeneous thresholding, we mean that instead of employing the typical hard thresholding operator
\begin{equation}\label{hardthresh}
H_\lambda(\wbf)_j=\begin{cases} \wbf_j, & |\wbf_j|\geq \lambda \\ 0, & \text{otherwise,}\end{cases}
\end{equation}
which treats all columns of $\Gbf$ equally, we define the variable hard-thresholding operator
\begin{equation}\label{hardthresh}
H_{\Lbf,\Ubf}(\wbf)_j=\begin{cases} \wbf_j, & L_j\leq |\wbf_j|\leq U_j \\ 0, & \text{otherwise}\end{cases}
\end{equation}
for specified upper and lower bound vectors $\Ubf,\Lbf \in \Rbb^J$ (for a $J$-term library $\Hbb$). This is particular useful for enforcing not only that the coefficients $\wbf_j$ stay within a certain range, but also that the term magnitudes $\|\Gbf_j\wbf_j\|$ have a reasonable contribution to the dynamics given by $\bbf$.

In this work, we define $\Lbf$ and $\Ubf$ to enforce that $\what_j$ and $\what_j\Gbf_j$ are comparable to the best 1-term solution. That is, denoting 
the projection operator by $\Pbf$, we define
\[\hat{j} :=\argmax_j \|\Pbf_{\Gbf_j}\bbf\|_2, \qquad \hat{w} := \frac{|\bbf\cdot \Gbf_{\hat{j}}|}{\|\Gbf_{\hat{j}}\|_2^2}, \qquad \hat{G} := \frac{\hat{w}\|\Gbf_{\hat{j}}\|_2}{\|\bbf\|_2}\]
where $\hat{w}$ and $\hat{G}$ are the coefficient and relative term magnitudes of the best 1-term projection of $\bbf$ onto the library $\Hbb$. We then let
\begin{align}
L_j(\lambda) &= \lambda^{-1}\max\{\hat{w},\hat{G}/\|\Gbf_j\|_2\} \\
U_j(\lambda) &= \lambda \min\{\hat{w},\hat{G}/\|\Gbf_j\|_2\}
\end{align}
which implies that the nonzero entries of $H_{\Lbf(\lambda),\Ubf(\lambda)}(\wbf)$ are the entries of $\wbf$ satisfying
\[\lambda \hat{w} \leq |\wbf_j|\leq \lambda^{-1}\hat{w} \qquad \&\qquad
\lambda \hat{G} \leq \frac{|\wbf_j|\|\Gbf_j\|_2}{\|\bbf\|_2}\leq \lambda^{-1}\hat{G}.\]
The MSTLS algorithm then proceeds as follows. With initial guess $\wbf^{(0)} = \Gbf^\dagger\bbf$ and support set $S^{(0)} = \{1,\dots,J\}$ the inner
STLS loop for fixed $\lambda$ is defined by
\begin{equation}\label{STLS}
(\text{STLS})\qquad \begin{dcases} 
S^{(\ell+1)} =\supp{H_{\Ubf(\lambda),\Lbf(\lambda)} (\wbf^{(\ell)})} \\
\hspace{-0.05cm}\wbf^{(\ell+1)} = \argmin_{\supp{\wbf}\subset S^{(\ell+1)}}\nrm{\Gbf\wbf-\bbf}_2^2
\end{dcases}
\end{equation}
which is run until termination, which must occur in maximum $J$ iterations \cite{ZhangSchaeffer2019MultiscaleModelSimul}. Denoting the solution by $\text{STLS}(\Gbf,\bbf;\, \lambda) = \wbf^\lambda$ (where $\wbf^0=\wbf^{(0)}$), we then define the auxiliary loss function $\CalL$ by
\begin{equation}\label{lossfcn}
\CalL(\lambda) = \frac{\nrm{\Gbf(\wbf^\lambda-\wbf^0)}_2}{\nrm{\Gbf\wbf^0}_2}+\frac{\nrm{\wbf^\lambda}_0}{\nrm{\wbf^0}_0}, \quad \wbf^\lambda = \text{STLS}(\Gbf,\bbf;\,\lambda).
\end{equation}
Finally, we define MSTLS with candidate sparsity thresholds $\pmb{\lambda}$ by
\begin{equation}\label{MSTLS2}
(\text{MSTLS})\qquad \begin{dcases} 
\hspace{0.01cm} \widehat{\lambda} = \min\left\{\lambda\in \pmb{\lambda} \ :\ \CalL(\lambda) = \min_{\lambda\in \pmb{\lambda}} \CalL(\lambda)\right\}\\
\widehat{\wbf}  =\text{STLS}(\Gbf,\bbf;\,\widehat{\lambda}).\end{dcases}
\end{equation}
We denote by $\text{MSTLS}(\Gbf,\bbf,\pmb{\lambda}):=\what$ the output weight vector in \eqref{MSTLS2}. The collection of thresholds $\pmb{\lambda}$ can be chosen by the user. Following the strategy seen to be successful in \cite{MessengerBortz2021JComputPhys}, we let $\pmb{\lambda}$ be 100 equally log-spaced values from $10^{-4}$ to 1. This process is visualized in Figure \ref{fig:mstlsloss}.

\begin{figure}
  \begin{center}
    \includegraphics[width=0.48\textwidth]{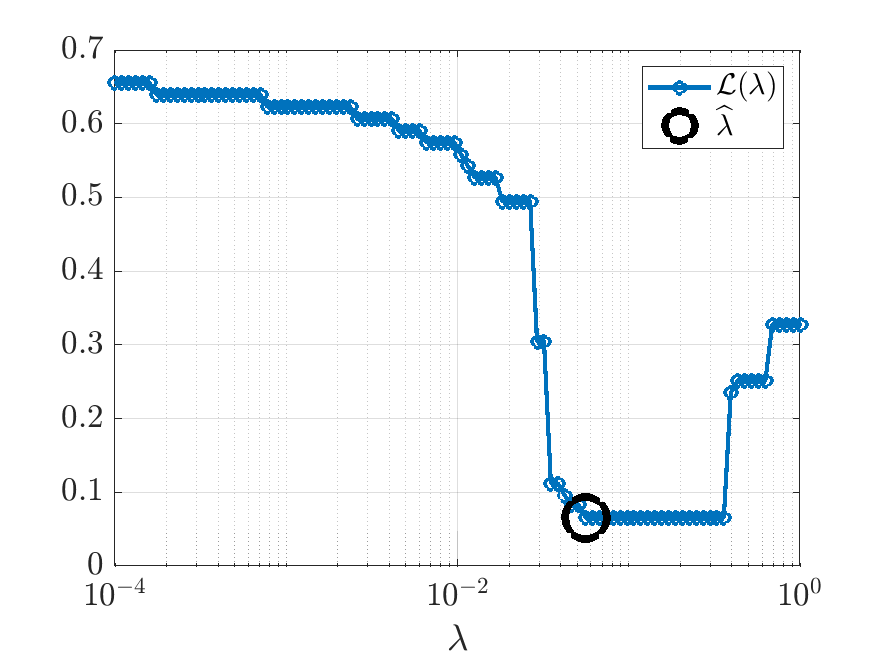}
  \end{center}
  \caption{Example profile of the loss function $\CalL(\lambda)$ employed in MSTLS \eqref{MSTLS2} for the data in Figure \ref{fig:noise_1}.}
  \label{fig:mstlsloss}
\end{figure}

\section{Initial conditions for forward simulations}\label{app:IC} 
In order to measure agreement between forward simulations of $\Hred$ and $\Hredhat$, we first use a data-driven nonlinear least squares approach to find an adequate set of initial conditions. For all time intervals $I_k = [t_0,t_k]$, $k=2,\dots,100$, we fit a polynomial of degree 2 through the data $\Zbf(I_k)$ and let $z(0)$ be the value of this polynomail at $t_0$. We then simulate the system of interest (defined by $\Hred$ or $\Hredhat$) starting from the selected $z(0)$ for 1/8 of the total time of the available data. The initial condition $z(0)$ that minimizes the error between these partial forward simulations and the data $\Zbf$ over the 99 time intervals $I_k$ is chosen for forward simulations.

\end{document}

%% file: notation.tex
\newcommand{\Hbb}{\mathbb{H}}

\newcommand{\Nbb}{\mathbb{N}}

\newcommand{\Rbb}{\mathbb{R}}

\newcommand{\Vbb}{\mathbb{V}}

\newcommand{\Abf}{\mathbf{A}}
\newcommand{\bbf}{\mathbf{b}}

\newcommand{\Gbf}{\mathbf{G}}

\newcommand{\Hbf}{\mathbf{H}}

\newcommand{\Jbf}{\mathbf{J}}

\newcommand{\Lbf}{\mathbf{L}}

\newcommand{\Pbf}{\mathbf{P}}

\newcommand{\Rbf}{\mathbf{R}}

\newcommand{\tbf}{\mathbf{t}}

\newcommand{\Ubf}{\mathbf{U}}

\newcommand{\wbf}{\mathbf{w}}

\newcommand{\Ybf}{\mathbf{Y}}

\newcommand{\Zbf}{\mathbf{Z}}

\newcommand{\vep}{\varepsilon}
\newcommand{\ep}{\eta}

\newcommand{\CalF}{{\mathcal{F}}}

\newcommand{\CalH}{{\mathcal{H}}}

\newcommand{\CalL}{{\mathcal{L}}}

\newcommand{\CalO}{{\mathcal{O}}}
\newcommand{\CalQ}{{\mathcal{Q}}}

\newcommand{\CalR}{{\mathcal{R}}}

\newcommand{\nrm}[1]{\left\Vert {#1} \right\Vert}

\newcommand{\supp}[1]{\text{supp}\left(#1\right)}

\newcommand{\ind}[1]{\mathbbm{1}_{#1}}

\renewcommand{\tilde}[1]{\widetilde{#1}}

\DeclareMathOperator*{\argmin}{argmin}
\DeclareMathOperator*{\argmax}{argmax}

\newcommand{\lan}{\left\langle}
\newcommand{\ran}{\right\rangle}

\newcommand{\wstar}{\wbf^\star}
\newcommand{\what}{{\widehat{\wbf}}}

\newtheorem{defn}{Definition}[section]

\newcommand*{\dt}[1]{%
  \accentset{\mbox{\large{\hspace{0.1mm}.\hspace{0.75mm}}}}{#1}}

\newcommand{\Hav}{H_\vep^*}
\newcommand{\HavLam}{\CalH_{\mu}^*}
\newcommand{\Hred}{\CalH^{\mu}_0}
\newcommand{\Oav}{\Omega_\vep^*}
\newcommand{\OavLam}{\Omega_{\mu}^*}
\newcommand{\Ored}{\Omega^{\mu}_0}
\newcommand{\Xav}{X_\vep^*}
\newcommand{\XavLam}{X_{\mu}^*}
\newcommand{\Xred}{X^{\mu}_0}
\newcommand{\RLam}{R_{\mu}}
\newcommand{\Mred}{M^{\mu}_0}
\newcommand{\Hredhat}{\widehat{\CalH}^\mu_0}
\newcommand{\Xredhat}{\widehat{X}^{\mu}_0}